\documentclass{article}

\usepackage[sorting=nyt, 
backend = bibtex, 
style=alphabetic, 
backref, 
maxbibnames=20, 
doi=false,
url=false]{biblatex}

\usepackage[utf8]{inputenc}

\usepackage{xparse}

\usepackage[T1]{fontenc}
\usepackage{amsmath,amssymb,amsthm}
\usepackage{stmaryrd}
\usepackage{mathtools}%

\usepackage{graphicx}
\graphicspath{{figs/}}
\usepackage{url}
\usepackage[section]{placeins}
\usepackage[dvipsnames]{xcolor}

\usepackage{tikz}
\usepackage{quantikz}

\usepackage{mathrsfs}
\usepackage{dsfont}
\usepackage{setspace}
\usepackage[export]{adjustbox}
\usepackage{makecell}
\usepackage{enumitem}
\setlist{noitemsep}

\usepackage{algorithm2e}
\usepackage[most]{tcolorbox}

\usepackage{longtable}

\usepackage{physics}
\usepackage[pdfusetitle]{hyperref}
\usepackage[capitalize,nameinlink]{cleveref}

\usepackage{comment}

\theoremstyle{definition}

\usepackage[framemethod=TikZ]{mdframed}
\mdfsetup{roundcorner=6pt,innerleftmargin=5pt, 
innerrightmargin=5pt}

\newmdtheoremenv{definition}{Definition}[section]
\newmdtheoremenv{lemma}[definition]{Lemma}
\newmdtheoremenv{theorem}[definition]{Theorem}
\newmdtheoremenv{proposition}[definition]{Proposition}
\newmdtheoremenv{corollary}[definition]{Corollary}
\newmdtheoremenv{remark}[definition]{Remark}
\newmdtheoremenv{example}[definition]{Example}
\newmdtheoremenv{fact}[definition]{Fact}
\newmdtheoremenv{claim}[definition]{Claim}

\numberwithin{equation}{section}

\newcommand{\splitatcommas}[1]{%
  \begingroup
  \ifnum\mathcode`,="8000
  \else
    \begingroup\lccode`~=`, \lowercase{\endgroup
      \edef~{\mathchar\the\mathcode`, \penalty0 \noexpand\hspace{0pt plus 1em}}%
    }\mathcode`,="8000
  \fi
  #1%
  \endgroup
}

\usepackage{thm-restate}

\newcommand{\SWAP}{\mathsf{SWAP}}
\newcommand{\CNOT}{\mathsf{CNOT}}

\newcommand{\CZ}{\mathsf{CZ}}

\newcommand{\MEASX}{\mathsf{M}_{\PAULIX}}
\newcommand{\INITZ}{\mathsf{INIT}_{\PAULIZ}}
\newcommand{\INITX}{\mathsf{INIT}_{\PAULIX}}
\newcommand{\INIT}{\mathsf{INIT}}
\newcommand{\MEASZ}{\mathsf{M}_{\PAULIZ}}

\newcommand{\PAULIX}{\mathsf{X}}
\newcommand{\PAULIY}{\mathsf{Y}}
\newcommand{\PAULIZ}{\mathsf{Z}}
\newcommand{\IDENT}{\mathsf{I}}
\newcommand{\HAD}{\mathsf{H}}
\newcommand{\PHAS}{\mathsf{S}}

\newcommand{\MF}{\mathcal{F}}

\newcommand{\dsl}{\llbracket}
\newcommand{\dsr}{\rrbracket}
\newcommand{\supp}{\operatorname{supp}}

\usepackage{CJKutf8}

\newcommand{\mixedstate}[1]{\ensuremath{\mathsf{D}\left(#1 \right)}}

\NewDocumentCommand\cqState{O{n_C}O{n_Q}}{%
  \ensuremath{\mathcal{CQ}\left(#1,#2\right)}
}

\NewDocumentCommand\mixedstateArb{m}{\ensuremath{\mathsf{D}\left(#1\right)}}
\NewDocumentCommand\mixedstateType{m}{\ensuremath{\mathsf{D}\left(\mathcal{H}_{#1}\right)}}

\newcommand{\unphysicalstate}[1]{\linearOp{#1}}

\NewDocumentCommand\gadget{}{\ensuremath{\mathsf{Gad}}}
\NewDocumentCommand\qcode{o}{\IfNoValueTF{#1}{\ensuremath{{\mathcal{Q}}}}{\ensuremath{{\mathcal{Q}^{(#1)}}}}}

\NewDocumentCommand\compCodeType{o}{\ensuremath{\mathsf{CompCode}_{#1}}}

\NewDocumentCommand\qGate{}{\ensuremath{\mathsf{g}}}
\NewDocumentCommand\weightenum{mm}{\ensuremath{\mathcal{W}{\left(#1;~#2\right)}}}
\NewDocumentCommand\fault{o}{%
  \IfNoValueTF{#1}
   {\ensuremath{\mathbf{f}}}
   {\ensuremath{\mathbf{#1}}}%
}

\NewDocumentCommand\circuitMap{m}{\ensuremath{\mathsf{map}[#1]}}

\DeclareMathOperator{\rowspan}{rowsp}

\DeclareMathOperator*{\wtsum}{\boxplus}

\NewDocumentCommand\nodeIn{o}{\IfNoValueTF{#1}{\text{in}}{\ensuremath{\text{in}(#1)}}}
\NewDocumentCommand\nodeOut{o}{\IfNoValueTF{#1}{\text{out}}{\ensuremath{\text{out}(#1)}}}
\NewDocumentCommand\nodeTime{o}{\IfNoValueTF{#1}{\mathsf{T}}{\ensuremath{\mathsf{T}(#1)}}}
\NewDocumentCommand\nodeTimeInv{o}{\IfNoValueTF{#1}{\mathsf{T}^{-1}}{\ensuremath{\mathsf{T}^{-1}(#1)}}}
\NewDocumentCommand\edgeConn{o}{\IfNoValueTF{#1}{\mathsf{conn}}{\ensuremath{\mathsf{conn}(#1)}}}

\NewDocumentCommand\nodeGate{o}{\IfNoValueTF{#1}{\mathsf{gate}}{\ensuremath{\mathsf{gate}(#1)}}}
\NewDocumentCommand\nodeGateP{o}{\IfNoValueTF{#1}{\mathsf{gate'}}{\ensuremath{\mathsf{gate'}(#1)}}}
\NewDocumentCommand\nodeFaulty{o}{\IfNoValueTF{#1}{\widetilde{\mathsf{gate}}}{\ensuremath{\widetilde{\mathsf{faulty}(#1)}}}}

\NewDocumentCommand\edgeSrc{o}{\IfNoValueTF{#1}{\mathsf{src}}{\ensuremath{\mathsf{src}(#1)}}}
\NewDocumentCommand\edgeDest{o}{\IfNoValueTF{#1}{\mathsf{dest}}{\ensuremath{\mathsf{dest}(#1)}}}
\NewDocumentCommand\edgeType{o}{\IfNoValueTF{#1}{\mathsf{type}}{\ensuremath{\mathsf{type}(#1)}}}

\NewDocumentCommand\typeQ{}{\ensuremath{\mathsf{Q}}}
\NewDocumentCommand\typeC{}{\ensuremath{\mathsf{C}}}
\NewDocumentCommand\cqvar{}{\ensuremath{\mathsf{cq}}}
\NewDocumentCommand\typeVar{o}{\IfNoValueTF{#1}{\ensuremath{\mathsf{t}}}{\ensuremath{\mathsf{t}_{#1}}}}
\NewDocumentCommand\typeSet{}{\ensuremath{\mathcal{T}}}

\NewDocumentCommand\linearOp{om}{\IfNoValueTF{#1}{\ensuremath{\mathsf{L}\left(#2\right)}}{\ensuremath{\mathsf{L}\left(#1, #2\right)}}}
\NewDocumentCommand\linearOpType{om}{\IfNoValueTF{#1}{\ensuremath{\mathsf{L}\left(#2\right)}}{\ensuremath{\mathsf{L}\left(#1, #2\right)}}}

\NewDocumentCommand\superOp{om}{\IfNoValueTF{#1}{\ensuremath{\mathcal{L}\left(#2\right)}}{\ensuremath{\mathcal{L}\left(#1, #2\right)}}}
\NewDocumentCommand\superOpType{om}{\IfNoValueTF{#1}{\ensuremath{\mathcal{L}\left(#2\right)}}{\ensuremath{\mathcal{L}\left(#1, #2\right)}}}

\NewDocumentCommand\F{o}{\IfNoValueTF{#1}{\ensuremath{\mathbb{F}_2}}{\ensuremath{\mathbb{F}_{#1}}}}

\NewDocumentCommand\adjgraph{m}{\ensuremath{G_{\mathsf{adj}}[#1]}}

\DeclareMathOperator{\im}{im}

\NewDocumentCommand{\surfacecode}{m}{\ensuremath{\mathsf{SC}{({#1})}}}
\NewDocumentCommand{\sctype}{mm}{\ensuremath{\mathsf{SC}_{{#2}}^{({#1})}}}

\NewDocumentCommand{\ecgadget}{o}{\IfNoValueTF{#1}{\ensuremath{\mathsf{EC}}}{\ensuremath{\mathsf{EC}^{{#1}}}}}
\NewDocumentCommand{\ldpcCode}{m}{\ensuremath{\mathsf{qLDPC}({#1})}}
\NewDocumentCommand{\ldpctype}{mm}{\ensuremath{\mathsf{qLDPC}_{{#2}}^{({#1})}}}
\newcommand{\inject}{\mathsf{INJECT}}
\newcommand{\unencode}{\mathsf{UNENCODE}}
\newcommand{\grow}{\mathsf{GROW}}
\newcommand{\shrink}{\mathsf{SHRINK}}
\newcommand{\mz}{\mathsf{MEAS}_{\PAULIZ}}
\newcommand{\poly}{\mathsf{poly}}

\NewDocumentCommand\encoder{o}{\IfNoValueTF{#1}{\ensuremath{\mathsf{enc}}}{\ensuremath{\mathsf{enc}_{#1}}}}
\NewDocumentCommand\decoder{o}{\IfNoValueTF{#1}{\ensuremath{\mathsf{dec}}}{\ensuremath{\mathsf{dec}_{#1}}}}

\NewDocumentCommand\Rencoder{o}{\IfNoValueTF{#1}{\ensuremath{{\mathsf{renc}}}}{\ensuremath{{\mathsf{renc}}_{#1}}}}
\NewDocumentCommand\Rdecoder{o}{\IfNoValueTF{#1}{\ensuremath{{\mathsf{rdec}}}}{\ensuremath{{\mathsf{rdec}}_{#1}}}}

\newcommand{\hilbert}{\mathcal{H}}

\newcommand{\error}{\mathcal{E}}
\newcommand{\errorin}{\mathcal{E}_{\mathrm{in}}}

\newcommand{\baderrors}{\mathcal{B}}
\newcommand{\badlcerrors}{\mathcal{B}^{\lightcone}}

\newcommand{\outerror}{B^{\text{out}}}
\newcommand{\reserror}{\baderrors^{(R)}}
\newcommand{\badfaults}{\mathcal{F}}

\newcommand{\clusterset}{\mathcal{CG}}

\newcommand{\power}{P}
\newcommand{\dev}[1]{\precapprox_{#1}}

\newcommand{\encU}{U^\dagger}
\newcommand{\decU}{U}

\newcommand{\contract}{\kappa}
\newcommand{\contractC}{C_\contract}
\newcommand{\contractD}{D_\contract}
\newcommand{\contractV}{V_\contract}
\newcommand{\contractE}{E_\contract}
\newcommand{\contractt}{t_\contract}
\newcommand{\contractConn}{\edgeConn_\contract}
\newcommand{\contractType}{\edgeType_\contract}
\newcommand{\contractGate}{\nodeGate_\contract}
\newcommand{\contractT}{\nodeTime_\contract}

\newcommand{\bundle}{\beta}
\newcommand{\bundleE}{E_\bundle}

\newcommand{\noise}{\mathcal{N}}

\newcommand{\environ}{\Omega}
\newcommand{\environC}{C_\environ}

\newcommand{\environV}{V_\environ}
\newcommand{\environE}{E_\environ}

\newcommand{\Cwithenv}{C_{\mathsf{env}}}

\newcommand{\g}{\mathcal{G}}

\newcommand{\depth}{T}
\newcommand{\lightcone}{\mathcal{LC}}

\newcommand{\recover}{\mathcal{R}}

\newcommand{\cc}{\mathcal{CO}}

\NewDocumentCommand\intype{o}{\IfNoValueTF{#1}{\ensuremath{\typeVar^{\text{in}}}}{\ensuremath{\typeVar^{\text{in}}_{#1}}}}

\NewDocumentCommand\outtype{o}{\IfNoValueTF{#1}{\ensuremath{\mathsf{t}^{\text{out}}}}{\ensuremath{\mathsf{t}^{\text{out}}_{#1}}}}

\NewDocumentCommand\ctype{o}{\IfNoValueTF{#1}{\ensuremath{{\mathsf{ctype}}}}{\ensuremath{{\mathsf{ctype}}_{#1}}}}

\NewDocumentCommand\spec{o}{\IfNoValueTF{#1}{\ensuremath{{\mathsf{spec}}}}{\ensuremath{{\mathsf{spec}}_{#1}}}}

\NewDocumentCommand\inspec{o}{\IfNoValueTF{#1}{\ensuremath{{\mathsf{spec}^{\text{in}}}}}{\ensuremath{{\mathsf{spec}}_{#1}^{\text{in}}}}}

\NewDocumentCommand\outspec{o}{\IfNoValueTF{#1}{\ensuremath{{\mathsf{spec}^{\text{out}}}}}{\ensuremath{{\mathsf{spec}}_{#1}^{\text{out}}}}}

\newcommand{\syntype}{\typeVar[\mathsf{syn}]}

\newcommand{\synd}{\upsilon}
\newcommand{\Fst}{\mathcal{F}_{\mathsf{spacetime}}}

\usepackage[intoc]{nomencl}
\makenomenclature 
\usepackage{etoolbox}

\renewcommand\nomgroup[1]{%
  \item[] %
  \vspace{.6\baselineskip}%
  \hspace*{-\labelwidth}\hspace*{-\labelsep}%
  {\bfseries %
  \ifstrequal{#1}{A}{Weight Enumerators}{%
  \ifstrequal{#1}{B}{Types and Operators}{%
  \ifstrequal{#1}{C}{Circuits and Faults}{%
  \ifstrequal{#1}{D}{Fault-Tolerant Gadgets}{%
  \ifstrequal{#1}{E}{Error Correction for CSS Codes}{%
  \ifstrequal{#1}{F}{Gadgets for qLDPC Codes}{}
  }}}}}%
  }%
  \par\vspace{.2\baselineskip}%
}

\newcommand{\bulk}{\mathsf{bulk}}
\newcommand{\Jbulk}{J_{\mathsf{bulk}}}
\newcommand{\Ibulk}{I_{\mathsf{bulk}}}
\newcommand{\Jd}{J^{\mathrm{d}}}
\newcommand{\Js}{J^{\mathrm{s}}}
\newcommand{\Hbulk}{H_{\mathsf{bulk}}}
\newcommand{\Hext}{H_{\mathsf{ext}}}

\newcommand{\ein}{e_{\text{in}}}
\newcommand{\fin}{f_{\text{in}}}
\newcommand{\cin}{c_{\text{in}}}

\newcommand{\cft}{C_{\mathsf{FT}}}
\newcommand{\ratio}{\Gamma}
\newcommand{\cig}{\mathcal{C}_{i,+}}
\newcommand{\cib}{\mathcal{C}_{i,-}}
\newcommand{\C}{\mathcal{C}}

\makeatletter
\newenvironment{construction}[1][Construction]{\par
  \pushQED{\(\blacktriangleleft\)}%
  \normalfont \topsep6\p@\@plus6\p@\relax
  \trivlist
  \item\relax
        {\itshape
    #1\@addpunct{.}}\hspace\labelsep\ignorespaces
}{%
  \popQED\endtrivlist\@endpefalse
}
\makeatother

\makeatletter
\newenvironment{proofclaim}[1][Proof of claim]{\par
  \pushQED{\(\blacksquare\)}%
  \normalfont \topsep6\p@\@plus6\p@\relax
  \trivlist
  \item\relax
        {\itshape
    #1\@addpunct{.}}\hspace\labelsep\ignorespaces
}{%
  \popQED\endtrivlist\@endpefalse
}
\makeatother

\usepackage{geometry}

\title{Composable Quantum Fault-Tolerance}
\author{
Zhiyang He (Sunny)\thanks{Department of Mathematics, Massachusetts Institute of Technology. \texttt{szhe@mit.edu}.} \and 
Quynh T. Nguyen\thanks{Harvard University. \texttt{qnguyen@g.harvard.edu}.} \and 
Christopher A. Pattison\thanks{Simons Institute for the Theory of Computing, University of California, Berkeley. \texttt{cpattison@berkeley.edu}.}
\thanks{Institute for Quantum Information and Matter, California Institute of Technology}
}

\begin{document}
\maketitle

\begin{abstract}
    Proving threshold theorems for fault-tolerant quantum computation is a burdensome endeavor with many moving parts that come together in relatively formulaic but lengthy ways. 
    It is difficult and rare to combine elements from multiple papers into a single formal threshold proof, due to the use of different measures of fault-tolerance.
    In this work, we introduce \textit{composable fault-tolerance}, a framework that decouples the probabilistic analysis of the noise distribution from the combinatorial analysis of circuit correctness, and enables threshold proofs to compose independently analyzed gadgets easily and rigorously.
    Within this framework, we provide a library of standard and commonly used gadgets such as memory and logic implemented by constant-depth circuits for quantum low-density parity check codes and distillation.
    As sample applications, we explicitly write down a threshold proof for computation with surface code and re-derive the constant space-overhead fault-tolerant scheme of~\cite{gottesman2014fault}
    using gadgets from this library.
    We expect that future fault-tolerance proofs may focus on the analysis of novel techniques while leaving the standard components to the composable fault-tolerance framework, with the formal proof following the intuitive ``napkin math'' exactly.

\end{abstract}

\section{Introduction}
Quantum fault-tolerance is arguably the central object of study in modern quantum error correction research.
Informally, theorems for quantum fault-tolerance (hereafter ``threshold theorems'') state that a circuit \(C\) can be mapped to a new circuit \(\cft\) that produces an output distribution close to that of \(C\) even in the presence of some amount of noise.

Threshold theorems have a long history dating back to the mid-90s~\cite{shor1996fault,knill1996threshold,aharonov1997fault,kitaev1997quantum} as early works justified the plausibility of large-scale fault-tolerant quantum computation.
These formal proofs are still relevant today for a variety of reasons.
Threshold theorems state how all the pieces (gadgets) of a particular fault-tolerance scheme will ``fit together.''
This allows users to be sure that no detail is missing from the scheme even if the analytically proven threshold is not quantitatively precise. 
Such a proof along with numerical estimation of the threshold is the best evidence for fault-tolerance without simulating or running the full circuit experimentally.
As we develop novel low-overhead gadgets for fault-tolerant computation, threshold theorems should be proven to illustrate how the new techniques change the full picture.

Unfortunately, proofs of threshold theorems remain a tedious and technical exercise where many details must be handled in relatively standard ways in order to wrap the novel portion of a construction.
Worse yet, proofs of threshold theorems are largely \emph{non-composable}: Fault-tolerance constructions are frequently written more or less monolithically with few formally written intermediate results suitable for reuse in other works.

Here, we introduce a framework for composable fault-tolerance where complicated fault-tolerant circuits and schemes may be assembled out of simple ones in a rigorous yet essentially black-box way.
On a high-level, we decouple the probabilistic analysis of noise models from the combinatorial analysis of circuit correctness.
This decoupling is enabled by the weight enumerator formalism (\cref{sec:weight-enumerators}), which bounds the set of circuit locations corrupted by noise using the notion of \textit{bad sets}. 
We briefly elaborate on this central mechanism of our formalism.

On the combinatorial side, consider a gadget (which is a circuit) with a set of corrupted locations and an input state with a set of corrupted qubits. 
We say that a circuit has \textit{failed} if the set of corrupted locations includes a \textit{bad fault path} and call a state \textit{bad} if the set of corrupted qubits includes a \textit{bad error support}.\footnote{Bad fault paths and bad error supports are collections of sets defined with respect to the gadget (including a decoder) and generally ``witness'' some configuration of faults that may lead to failure of the gadget.}
Then, the correctness properties of a gadget can be easily specified combinatorially: \text{A successful execution of the circuit on a good input state results in a good output state}. 
This specification of the input/output behavior is at the heart of our composablility -- two gadgets with compatible definitions of bad error supports can now be analyzed independently and later composed to become a larger gadget with no regard to the probability distribution of noise.\footnote{In particular, whether the distribution of noise across the two gadgets is independent or not is irrelevant.} 
The specification of bad fault paths and error supports for this larger gadget can be easily computed from the specifications of the individual gadgets in the weight enumerator formalism. 

On the probabilistic side, for noise models such as locally stochastic noise, given the definition of bad sets from the gadgets, the probability of noise corrupting a set of locations that contain a bad set can be bounded by the weight enumerator polynomial.
Moreover, computation of bad sets (as needed for composition of gadgets) corresponds straightforwardly to computation with polynomials. 
Existence of a threshold then follows straightforwardly if the polynomial decays exponentially in some parameter such as code distance.

Our main technical contribution is the rigorous formulation of these definitions and concepts, especially the combinatorial specification of gadgets.
The weight enumerator formalism, as well as a precursor collection of combinatorial definitions, were introduced and utilized in~\cite{nguyen2024quantum} to prove threshold for a low-overhead fault-tolerant scheme.
This work keeps the same weight enumerator algebra while re-designing and generalizing the core definitions, introducing a complete edition of the composable fault-tolerance framework. 
In particular, we introduce more precise definitions, tools for manipulations and correctness proofs of fault-tolerant circuits, as well as a library of examples.

We demonstrate the general applicability of our composable framework with instructive examples.
\begin{enumerate}[topsep = 0pt]
    \item We build a library of standard and commonly used gadgets such as error correction in quantum low-density parity check codes using \(d\)-rounds of syndrome extraction, logical gates implemented by constant depth circuits, and distillation. These gadgets are analyzed independently and combinatorially, formulated in terms of bad sets.
    
    \item Utilizing gadgets from this library, we assemble a 
    constant space-overhead quantum fault-tolerance scheme, providing an alternative proof of the results in \cite{tamiya2024polylog} and \cite{gottesman2014fault} up to differences in the classical computation model. In place of the concatenated codes fault-tolerant schemes used in \cite{tamiya2024polylog,gottesman2014fault}, we use a threshold theorem for surface codes built with magic state distillation and transversal gates. 
    To our knowledge, this is the first explicit proof that surface code using magic state distillation has a constant noise threshold.\footnote{We note that this is essentially a formalization of a proof sketch in \cite{dennis2002topological} using magic state distillation \cite{bravyi2005universal}.}

    \item To facilitate concatenated coding and fault-tolerant schemes, we prove a \textit{level reduction} theorem similar to the one in \cite{aliferis2005quantum} where a fault tolerant circuit can be shown to noisily simulate another fault-tolerant circuit on the logical level. 
    The bad sets and weight enumerator polynomials for such a concatenated circuit (equivalently gadget) are easily computable as the composition of the bad sets and polynomials of its component gadgets.
\end{enumerate}

With the composable fault-tolerance framework, we expect that future threshold proofs of fault-tolerant schemes may focus on the combinatorial analysis of error propagation and correction in novel gadgets, while leaving the standard components and rigorous formality to the framework.
For instance, while previously schemes may prove that a fault-tolerant circuit with locally stochastic noise on the input will produce the correct output subject to a (different) locally stochastic noise, future schemes may analyze such a gadget fully combinatorially with no regard to the probabilistic distribution of noise. 
Moreover, a novel fault-tolerant gadget which implements a particular logical operation can now be easily extended into a full fault-tolerant universal computation scheme by composing with standard gadgets. Such extensions are often invoked cursorily in the literature to support the utility of novel gadgets, without rigorous analysis of noise distribution. 
The composable fault-tolerance framework provides structure to fill in these rigorous details and supports new constructions with a unified formal foundation.

\subsection{Reader's guide}
While this paper is lengthy, we have attempted to ensure that readers will benefit from reading only subsets of the contents.
Footnotes are provided when subtle details of definitions are relevant or to explain generalizations.
A nomenclature table is included at the end of the paper.
On a light reading, we recommend readers read \cref{sec:weight-enumerators} (defining the bad sets, weight enumerator polynomials and their algebra) carefully while skipping \cref{def:composition} until it becomes relevant.
\Cref{sec:circuit-formalism} may be skimmed.
It formalizes a relatively intuitive notion of a quantum circuit in the presence of additional classical or quantum processing.\footnote{For conceptual simplicity, we are not restricting the classical computation circuit.
The definitions carry over to the case of a restricted depth classical computation per layer of quantum operations although the constructions given here usually do not carry over without modifications to the algorithms used for classical processing.}
We also formally define the notion of faults, which captures how noise may corrupt gates in a circuit.\footnote{Note that a fault is not a probabilistic distribution of noise (or a noise model per se), but rather an instantiation of a noise model corrupting a circuit.}
The tools in later part of the paper work with circuits subject to particular faults.

\Cref{sec:gadgets} contains the core definitions of gadgets\footnote{By gadget, we mean a quantum circuit that may act on or output a quantum register that has been shown to perform some useful operation.} suitable for composability, facilitated by the formulation of bad error supports and fault paths.
We state and prove the parallel and sequential composition of gadgets (\cref{lemma:gadget-composition}), as well as the \textit{level reduction} theorem (\cref{lemma:level-reduction}) that allows a fault-tolerant quantum circuit to simulate another one.\footnote{This provides a proof of a claim made in \cite{nguyen2024quantum} about the composability of certain fault-tolerance schemes (in particular, those with ``friendly'' gadgets such as used in concatenated code fault-tolerance).}
We recommend that readers read and interpret the statements of definitions and lemmas on an intuitive level, leaving the detailed mathematical formalities and proofs to a deeper read. 

In \cref{sec:circuit-correctness}, we perform the probabilistic analysis which bounds the failure probability of fault-tolerant circuits under various noise models by the weight enumerator polynomials of the (combinatorial) bad sets, which are computed from the composition of gadgets. The analysis for locally stochastic noise is straightforward. 
We further study coherent noise, and present partial progress on obtaining a constant threshold bound for computation using qLDPC codes with sublinear distance such as surface codes.
At the moment, we are able to show a threshold \emph{for computation} that is vanishing polylogarithmically with the circuit volume.

\Cref{sec:qldpc-gadgets} is where we begin building the library of standard gadgets with combinatorial analysis.
We begin by formulating the bad error supports for a LDPC code, and then construct error correction gadgets and transversal gate gadgets analogous to those of \cite{gottesman2014fault} under the assumption of a minimum-weight decoder.
We recommend that readers read the statements of the definitions and lemmas, and relate them to those in \cref{sec:gadgets}.
The proofs involve some intuitive yet technically involved combinatorics, which would be relevant examples for readers who wish to use our framework to prove fault-tolerance for novel gadgets. They may be skipped on a light reading.

\cref{sec:qldpc-threshold-theorems} is where we showcase the applicability of the composable fault-tolerance framework by assembling independently analyzed gadgets from \cref{sec:qldpc-gadgets} into threshold proofs for fault-tolerant schemes in a black-box fashion.
The first threshold theorem is for surface codes using magic state distillation and transversal gates.
Utilizing this fault-tolerant scheme, we provide an alternative proof of constant space overhead quantum fault-tolerance closely following \cite{tamiya2024polylog} and \cite{gottesman2014fault}.
We note that \cite{tamiya2024polylog} counts the runtime of the classical computation which introduces additional requirements on their gadgets.
\cite{nguyen2024quantum} reduces the time overhead to essentially logarithmic in a similarly restricted model of classical computation.
Proofs in \cref{sec:qldpc-threshold-theorems} are split into construction and proof.
We suggest that readers focus on the constructions on a first read.
We further note that the constructions and proofs in \cref{sec:qldpc-gadgets} and \cref{sec:qldpc-threshold-theorems} are asymptotic in nature, which means there was no attempt to optimize the constants. 
While such asymptotic results have found applications in complexity theory, it is our hope that asymptotic analysis of fault-tolerant should provide evidence and potential guidance to practical constructions of large-scale fault-tolerant quantum computers, much as the recent studies on theoretical and practical qLDPC codes have demonstrated.

Finally, we remark that the current framework and manuscript is the result of extensive revisions, as formulating a rigorous and applicable set of definitions necessarily encounters many subtleties. 
Seemingly obvious definitions may either be inaccurate or insufficient to prove important results such as gadget composition or level reduction.
Consequently, some of the core definitions, such as those surrounding gadgets in \cref{sec:gadgets}, are quite technical and may not seem intuitive on first sight. 
We try our best to point out the relevant intuition and subtleties in discussions and footnotes, and kindly ask that readers contact us for any questions or comments.

\subsection{Prior Works}
\cite{christandl2024fault} provide an alternative formalism for quantum fault-tolerance based on channels and approximate simulation as opposed to the exact simulation and combinatorial analysis considered here.
Many of the applications considered here have equivalent statements in \cite{christandl2024fault}.
Our formalism is heavily inspired from \cite{aharonov1997fault} and \cite{kitaev1997quantum}.
It is possible to include the ``extended rectangles'' from \cite{aliferis2005quantum} however, we elected not to in order to simplify the presentation.
The ``level reduction'' theorem (\cref{lemma:level-reduction}) is a vast generalization of the one proved in \cite{aliferis2005quantum}.

\cite{nguyen2024quantum} introduced an early version of the weight enumerator formalism in the minimal form required to prove the result.
Our work greatly expands the framework in an approachable way and provides justification of a claim made about recursive simulation of fault-tolerance schemes required for amplification of the threshold to constant in~\cite{nguyen2024quantum}.

\subsection{Future Directions}
We expect the composable fault-tolerance framework to serve as the formal foundation in a large number of future fault-tolerance results.
We briefly discuss some potential directions here.
\begin{enumerate}[topsep = 0pt]
    \item Constant coherent noise thresholds are known for concatenated constructions~\cite{aharonov1997fault,aliferis2005quantum,christandl2024fault}, but can one prove a constant noise threshold against coherent noise without code concatenation? \cite{christandl2024fault} has shown that constant noise threshold is possible when using linear-distance qLDPC codes with a single-shot decoder, but it is reasonable to conjecture that the linear-distance requirement can be dropped. Prior works have analytically shown a $O(1/d)$ coherent noise threshold for distance-$d$ surface codes in the memory setting~\cite{iverson2020coherence}, but numerical simulations suggest the threshold to be a constant~\cite{bravyi2018correcting}.
    \item In this work, we impose no constraints on the connectivity of the physical qubits and allow arbitrary blocks of qubits to interact.
    It would be useful to prove threshold theorems for computation in restricted models of connectivity using techniques such as qubit routing~\cite{gidney2025constant,pattison2023hierarchical} or lattice surgery~\cite{Horsman2012LatticeSurgery}. 
    We expect that the fault-tolerance of code surgery gadgets (see \cref{remark:code-switching}) can be proved similarly to the qLDPC error correction gadget (\cref{lemma:qldpc:ec-gadget}).
    \item Asymptotically, the spacetime overhead of quantum fault-tolerance has been brought down substantially against worst-case circuits.
    An obvious next step would be to customize fault-tolerance schemes to circuits of interest.
    Then, perhaps it is possible to construct practically relevant quantum fault-tolerance schemes with lower asymptotic overhead on particular circuits.
    \item While only a handful of applications require extended rectangles (see \cite{aliferis2005quantum}) and our definitions are compatible with them, we have found it to be notationally burdensome to work with in full generality. Perhaps some improved notation could be developed.
    \item It would be a useful resource to establish a threshold theorem for the surface code under minimal runtime assumptions for the classical decoder e.g. using a parallel window decoder~\cite{skoric2023parallel,tan2023scalable} or the parallel decoder from \cite{takada2025doubly}.
    \item We expect our techniques to also be useful for providing an end-to-end threshold proof of low time overhead quantum fault-tolerance using transversal gates in surface codes \cite{zhou2024algorithmic,serra2025decoding,cain2025fast}.
\end{enumerate}

\subsection{Acknowledgments}
The authors would like to thank 
Alfred,
Thiago Bergamaschi,
Ken Brown,
Kathleen Chang,
Nicolas Delfosse, 
Anirudh Krishna,
Urmila Mahadev,
John Preskill,
Armands Strikis, and
Eugene Tang
for useful discussions and encouragement.
 
This work was done in part while the authors were visiting at the Quantum Algorithms, Complexity, and Fault-Tolerance workshop hosted at the Simons Institute for the Theory of Computing (supported by DOE QSA grant \#FP00010905 and NSF QLCI Grant No. 2016245), the Fault-Tolerant Quantum Technologies workshop hosted at the Centro de Ciencias de Benasque Pedro Pascual (CCBPP), and the Institute for Quantum Information and Matter (IQIM) at Caltech. 
We thank the Simons Institute, the CCBPP, and Caltech IQIM for their hospitality.

Z.H. is supported by the MIT Department of Mathematics and the NSF Graduate Research Fellowship Program under Grant No. 2141064. QTN acknowledges support through the NSF Award No. 2238836 and IBM PhD fellowship.
C.A.P. is currently a Simons-CIQC postdoctoral fellow at the Simons Institute for the Theory of Computing, supported by NSF QLCI Grant 2016245.
C.A.P. acknowledges funding from the U.S. Department of Energy Office of Science, DE-SC0020290.
The Institute for Quantum Information and Matter is an NSF Physics Frontiers Center.

\tableofcontents

\section{Weight enumerator formalism}\label{sec:weight-enumerators}
We now introduce some definitions that will allow us to work with sets of faulty locations in a composable way that we call the ``weight enumerator formalism''.
An earlier version of the weight enumerator formalism appeared in \cite{nguyen2024quantum}.

We begin by defining sets that we will want the support of errors (or faults) to avoid.
Roughly, whenever the errors do not contain any of these sets, we will be able to conclude some property irrespective of what the error is.
One (basic) example of this is for a length-\(n\) repetition code, all subsets of weight \(\lfloor \frac{d-1}{2}\rfloor + 1\).

\begin{definition}[Avoiding sets]
  For a set \(\Omega\) and a family of subsets \(\mathcal{F}\subseteq P(\Omega)\), referred to as \textbf{the bad sets}, a subset \(X \subseteq \Omega\) is said to be \(\mathcal{F}\)-avoiding if it does not contain an element of \(\mathcal{F}\).
  That is, \(X \subseteq \Omega \) is \(\mathcal{F}\)-avoiding if and only if
  \begin{align}
    \forall F \in \mathcal{F},~ F \not\subseteq X~.
  \end{align}
\end{definition}

\begin{definition}[Set lower bound]
    For a set \(\Omega\) and two families of subsets \(\mathcal{F}_1, \mathcal{F}_2\subseteq P(\Omega)\), \(\MF_1\) is said to be a \textbf{set lower bound} of \(\MF_2\) if for every \(F_2\in \MF_2\), there exists \(F_1\in \MF_1\) such that \(F_1\subseteq F_2\). We denote this relation as \(\MF_1\preceq \MF_2\). 
\end{definition}

As it turns out, it is frequently the case that the only data required about the family of sets (for an end-user) is how many sets of each weight there are.
Thus, we can associated these families with polynomials that we refer to as \emph{weight enumerator polynomials} in analogy to the weight enumerator polynomials used in coding theory.
\begin{definition}[Weight enumerators]\label{def:weight-enumerator}
  For a finite set \(\Omega\) and a family of subsets \(\mathcal{F} \subseteq P(\Omega)\), it will be convenient to associate a polynomial \(\weightenum{\mathcal{F}}{x}\) to \(\mathcal{F}\) defined as the sum
  \begin{align}
    \weightenum{\mathcal{F}}{x} = \sum_{w=0}^{\infty} A_w x^w,
  \end{align}
  where \(A_w = |\left\{f \in \mathcal{F} \mid |f| = w\right\}|\) is the number of elements of \(\mathcal{F}\) of weight \(w\).
\end{definition}
\nomenclature[A, 01]{\(\weightenum{\mathcal{F}}{x}\)}{Weight enumerator function of a collection of sets \(\mathcal{F}\), see~\cref{def:weight-enumerator}.}

\begin{example}[Local stochastic noise]
    To see that this data is useful, suppose \(x \in \F^n\) is a bit string with entries \(1\) with probability \(p\) and \(0\) with probability \(1-p\).
    Then, for any family \(\mathcal{F} \subseteq P([n])\) the probability that \(x\) fails to be \(\mathcal{F}\)-avoiding is upper bounded by the evaluation of the corresponding weight enumerator:
    \begin{align}
        \Pr(\text{\(x\) is not \(\mathcal{F}\)-avoiding}) \le \sum_{f \in \mathcal{F}} \Pr(f \subseteq x) \le \weightenum{\mathcal{F}}{p}
    \end{align}
    More generally, this bound holds whenever \(x\) is distributed according to a so-called \emph{local stochastic} distributions.
    Local stochastic distributions obey the following property.
    \begin{align}
      \text{for all }S \subseteq [n] \qquad \Pr(S \subseteq x) \le p^{|S|}
    \end{align}
\end{example}

Once one associates polynomials to these families of bad sets, it is natural to ask whether the polynomial ring operations have any operational meaning. 
For many noise distributions, the error support is not independent on disjoint sets.
However, we can define certain sum and product operations that correspond to union-bounds and independence in the independent-noise case.

\begin{definition}[Ring of bad sets]\label{def:ring_of_bad_sets}
  For a set \(\Omega\) with the decomposition \(\Omega_1\cup \Omega_2 = \Omega\) and families of subsets \(\mathcal{F}_1\subseteq P(\Omega_1)\) and \(\mathcal{F}_2\subseteq P(\Omega_2)\), we define two operations between \(\mathcal{F}_1\) and \(\mathcal{F}_2\) to arrive at a subset \(\mathcal{F}\subseteq \Omega\).
  Let \(\iota_1\) (\(\iota_2\)) be the canonical inclusion map from \(\Omega_1\) (\(\Omega_2\)) to \(\Omega\).

  The first operation is a sort of addition operation.
  \begin{align}
    \mathcal{F}_1\boxplus \mathcal{F}_2 := \iota_1(\mathcal{F}_1) \cup \iota_2(\mathcal{F}_2),
  \end{align}
  where by \(\iota_1(\mathcal{F}_1)\), we mean the application of \(\iota_1\) to each element of the family \(\mathcal{F}_1\).

  When \(\Omega_1\) and \(\Omega_2\) are disjoint, that is \(\Omega = \Omega_1\sqcup \Omega_2\), we define a multiplication operation.%
  \begin{align}
    \mathcal{F}_1\circledast \mathcal{F}_2 := \{\iota_1(f_1) \cup \iota_2(f_2) \mid f_1 \in \mathcal{F}_1,~f_2 \in \mathcal{F}_2\}.
  \end{align}
\end{definition}

\nomenclature[A, 02]{\(\mathcal{F}, \boxplus, \circledast\)}{Ring of bad sets, see~\cref{def:ring_of_bad_sets}. \(\mathcal{F}\) is a collection of bad sets, and \(\boxplus,\circledast\) denote addition and multiplication of such collections.}

The weight enumerators of the sums and products are sums and products of the corresponding weight enumerators.
\begin{proposition}\label{lemma:enumerator-ring} It holds that
  \begin{align}
    \weightenum{\mathcal{F}_1\boxplus \mathcal{F}_2}{x} &= \weightenum{\mathcal{F}_1}{x} + \weightenum{\mathcal{F}_2}{x}, \\
    \weightenum{\mathcal{F}_1\circledast \mathcal{F}_2}{x} &= \weightenum{\mathcal{F}_1}{x} \weightenum{\mathcal{F}_2}{x}. 
  \end{align}
\end{proposition}
\begin{proof}
  A weight \(w\) element of \(\mathcal{F}_1\boxplus \mathcal{F}_2\) is a weight \(w\) element of \(\mathcal{F}_1\) or of \(\mathcal{F}_2\), so the coefficients of the same degree add.
  To prove the second property, note that any weight \(w\) element of \(\mathcal{F}_1\circledast \mathcal{F}_2\) must be the disjoint union of a weight \(w_1\) element of \(\mathcal{F}_1\) and a weight \(w_2\) element of \(\mathcal{F}_2\) where \(w = w_1+w_2\).
\end{proof}

The final operation we will need is a sort of composition.
This should be thought of capturing the bad errors of concatenated codes. 
This operation will appear when we analyze the recursive use of gadgets (\cref{lemma:level-reduction}).
Informally, given a fault-tolerant circuit \(C\) that fails with probability \(P_{L,1}(p)\) and fault-tolerant gadgets that fail with probability \(P_{L,2}(p)\), the circuit given by replacing every gate in \(C\) with a gate gadget of the second fault-tolerance scheme fails with probability \(P_{L,1}(P_{L,2}(p))\).
\begin{definition}[Composition]\label{def:composition}
  For a proposition \(Q\) and a set \(X\), let \(\mathbb{I}_{Q}[X]\) denote \(X\) when \(Q\) holds and \(\varnothing\) when \(\lnot Q\) holds.

  Fix a set \(\Omega\) and \(\mathcal{F}\subseteq P(\Omega)\). For ease of notation, identify \(\Omega \simeq [n]\). Also, consider sets \(\{\omega_i\}_{i \in \Omega}\) and families of sets \(\{\mathcal{S}_i \subseteq P(\omega_i)\}_{i \in \Omega}\) that are indexed by elements of \(\Omega\).
  We define an operation \(\mathcal{F} \bullet \{\mathcal{S}_i\}_i \subseteq P\left(\bigsqcup_{i\in\Omega} \omega_i\right)\) where
  \begin{align}
    \mathcal{F} \bullet \{\mathcal{S}_i\}_i := \boxplus_{f \in \mathcal{F}} \left(\mathbb{I}_{n \in f} [\mathcal{S}_n] \circledast \mathbb{I}_{{n-1} \in f} [\mathcal{S}_{n-1}] \circledast \dots \circledast  \mathbb{I}_{1 \in f} [\mathcal{S}_1]\right).
  \end{align}
  Note that if \(\mathcal{S}_i = \mathcal{S} \subseteq P(\omega)\) for some set \(\omega\), then \(\mathcal{F} \bullet \{\mathcal{S}_i\}_i \subseteq P(\Omega \times \omega)\) corresponds to all sets which for some element \(f \in \mathcal{F}\) are elements of \(\mathcal{S}\) on each row of \(\Omega \times \omega\) where \(f\) is non-trivial.
We will use the notation \(\mathcal{F} \bullet \mathcal{S}\) for this special case.
\end{definition}
\nomenclature[A, 03]{\(\mathcal{F} \bullet \{\mathcal{S}_i\}_i \)}{Composition of families of sets, see~\cref{def:composition}.}
For a number \(n \in \mathbb{N}\), and a family of sets \(\mathcal{S}\), we use the notation \(\mathcal{S}^{\bullet n}\) to denote the \(n\)-fold \(\bullet\) operation with \(\mathcal{S}^{\bullet 0}\) defined to be \(\{\{\centerdot\}\}\) where \(\{\centerdot\}\) is a singleton set.
The element will be clear from context.

\begin{proposition}[Composition of weight enumerators]\label{lemma:composition-upper-bound}
  Using the variables as defined in \cref{def:composition}, when there exists a polynomial \(p(x)\) such that on some interval \(x \in I \subseteq \mathbb{R}_+\) for all \(i \in \Omega\), \(\weightenum{\mathcal{S}_i}{x} \le p(x)\),
  \begin{align}
    \weightenum{\mathcal{F} \bullet \{\mathcal{S}_i\}_i}{x} \le \weightenum{\mathcal{F}}{p(x)}
  \end{align}
  for \(x \in I\).
\end{proposition}
\begin{proof}
  \(\mathcal{F} \bullet \{\mathcal{S}_i\}_i\) is constructed as a sum over a product for each \(f \in \mathcal{F}\).
  Using the weight enumerator sum rule, the weight enumerator is the sum of weight enumerators for each product.
  Each product has \(|f|\) terms and the corresponding weight enumerator can be evaluated using the weight enumerator product rule.
  Applying the restriction \(x \in I\) and the upper bound \(p(x)\), each term in the sum has upper bound \(p(x)^{|f|}\).
  For each \(d \in \mathbb{N}\), the number of terms in the sum proportional to \(p(x)^d\) correspond to the number of elements of \(\mathcal{F}\) of weight \(d\).
\end{proof}

Since avoiding families of bad sets imposes properties on a set, it will be convenient to define a partial order of bad sets.
\begin{definition}[Partial order of bad sets]\label{def:set-partial-order}
    For a set \(\Omega\) and two family of subsets \(\mathcal{F}_1, \mathcal{F}_2\subseteq P(\Omega)\),
    we say that \(\mathcal{F}_1 \preceq \mathcal{F}_2\) if every element of \(\mathcal{F}_2\) is a superset of some element of \(\mathcal{F}_1\).
    \begin{align}
        \forall T \in \mathcal{F}_2,\quad \exists S \in \mathcal{F}_1,\quad S \subseteq T.
    \end{align}
\end{definition}
It is straightforward to see that, for a subset \(X\subseteq \Omega\), if \(X\) is \(\mathcal{F}_1\)-avoiding and \(\mathcal{F}_1 \preceq \mathcal{F}_2\), then \(X\) is also \(\mathcal{F}_2\)-avoiding.

\section{Circuit formalism}\label{sec:circuit-formalism}
In order to work with our fault-tolerant circuits and gadgets in a black-box manner, we need to model a large degree of data that describes them.
This will allow us to rewrite larger fault-tolerant circuits that are built out of smaller fault-tolerant gadgets on a formal level.

\subsection{Networks}\label{sec:circuit}
We start by defining a quantum-classical circuit as a network.
In what follows, we will write definitions that are heavily inspired from tensor networks. However, no knowledge of tensor networks is required.

We will use the notion of a directed graph \(G = (V,E)\) with multi-edges.\footnote{In a multi-graph, there may be multiple edges between a pair of vertices. Thus, the edge set is \(E \subseteq V \times V \times \mathbb{N}\).}
For an edge \(e \in E\), we will say that its source is \(\edgeSrc[e] \in V\) and its destination is \(\edgeDest[e] \in V\).
For a vertex \(v \in V\), we use \(\nodeOut[v]\) to denote the number of edges \(e \in E\) with source \(v = \edgeSrc[e]\).
Likewise, we use \(\nodeIn[v]\) to denote the number of edges \(e \in E\) with destination \(v = \edgeDest[e]\).

For two vertices \(u,v \in V\), we say that \(u\) is an ancestor of \(v\) if there exists a sequence of vertices \(w_1, w_2, w_3, \dots, w_m \in V\) that form a (directed) path from \(u\) to \(v\).
That is, there is a sequence of edges \((u,w_1), (w_1,w_2), \dots, (w_{m-1}, w_m), (w_m,v)\) in \(G\).

\begin{definition}[Network]\label{def:network}
    A \textbf{network} is a labeled directed multi-graph \((V, E, \edgeConn)\) described by edge labels \(\edgeConn \colon E \to \mathbb{N} \times \mathbb{N}\), referred to as the \emph{connections} of the network.
    When there exists an edge \(e\) with source \(v\) and destination \(u\) and connection \((i,j) = \edgeConn[e]\), the edge \(e\) is said to connect output \(i\) of \(v\) to input \(j\) of \(u\).
    The following conditions ensure that this assignment is consistent.
    \begin{itemize}
        \item Every edge \(e \in E\) is labeled by valid input and output indices.
        \begin{align}
            \edgeConn(e) \in [\nodeOut[\edgeSrc[e]]] \times [\nodeIn[\edgeDest[e]]].
        \end{align}
        \item For every vertex \(v \in V\) and every output index \(i \in [\nodeOut[v]]\), there is a unique edge \(e \in E\) with source \(v = \edgeSrc[e]\) and label \(\edgeConn[e]=(i,j)\) for some arbitrary \(j\). 
        \item For every vertex \(u \in V\) and every input index \(j \in [\nodeIn[u]]\), there is a unique edge \(e \in E\) with destination \(u = \edgeDest[e]\) and label \(\edgeConn[e]=(i,j) \) for some arbitrary  \(i\). 
    \end{itemize}

    The vertex set may additionally include two special vertices \(\bot\) and \(\top\) which are referred to as the \textbf{input and output vertices} (IO vertices), respectively.
    We require that \(\bot\) has no inputs and \(\top\) has no outputs, and refer to the set \(V \setminus \{\bot, \top\}\) as the \textbf{non-IO vertices} of the network.
\end{definition}
\nomenclature[B, 01]{\(G = (V, E)\)}{A directed, acyclic multigraph that represents a network, and later a circuit. See~\cref{def:network}.}
\nomenclature[B, 02]{\(\edgeSrc(e),\edgeDest(e)\)}{Source and destination vertices of a directed edge \(e\).}
\nomenclature[B, 03]{\(\nodeIn(v),\nodeOut(v)\)}{Edges that have \(v\) as destination and edges that have \(v\) as source.}
\nomenclature[B, 04]{\(\edgeConn\)}{Labeling of edges in the network that specifies the sources and destinations of edges.}
\nomenclature[B, 05]{\(\top,\bot\)}{Special input and output vertices of a network.}

A network should be thought of as a directed graph between vertices where each edge connects a particular output of a vertex to a particular input of some other vertex.
Note that this is very nearly the same data that defines a tensor network: In a tensor network, the connections denote contractions of tensor legs labeled ``locally'' by indices.
The IO vertices will eventually be placeholders for quantum or classical input and outputs.

Our networks will represent circuits, whose operations (e.g., gates) correspond to the vertices, so they will need a notion of time.
\begin{definition}[Foliated network]\label{def:foliated-network}
    A foliated network \(G = (V,E,\edgeConn)\) of depth \(D\) is a network with a partition of vertices, denote the ``time'' of a vertex, \(\nodeTime \colon V \to [D]\), such that every edge \(e \in E\) has source and destinations at consecutive times. 
    \begin{align}
        \nodeTime[\edgeSrc[e]] + 1 = \nodeTime[\edgeDest[e]]
    \end{align}
    The foliation will be kept implicit until needed.
    For a time \(t \in [D]\), the set of vertices at time \(t\), referred to as a \textbf{timestep}, is \(\nodeTimeInv[t] \subseteq V\).
\end{definition}
\nomenclature[B, 06]{\(\nodeTime \colon V \to [D]\)}{Time labeling of vertices in a network. \(D\) is the depth of the network. For every \(t\in [D]\), the set of vertices at time \(t\) is called a timestep.}
Note that in a foliated network, cycles are not permitted and all outputs of a vertex are connected to inputs of vertices at the following time.
One can straightforwardly see that a foliated network is a \textbf{directed acyclic graph}: For any vertex \(v \in V\), every ancestor \(u \in V\) of \(v\) satisfies \(\nodeTime[u] < \nodeTime[v]\).
Thus \(\nodeTime[\cdot]\) induces a partial order on the set of vertices.
This partial order can be extended to a (possibly non-unique) total order, known as a \textbf{topological sort}, such that if \(u,v \in V\) satisfies \(u \le v\), then \(v\) is not an ancestor of \(u\).
For convenience, we will refer to a arbitrary (but fixed) topological sort \((v_1, \dots, v_{|V|})\) as the ordered vertices of \(G\) and require \(v_1 = \bot\) and \(v_{|V|} = \top\).

\subsection{Types}

Each connection in a network must carry additional information.
We refer to this as the ``type'' of a connection in analogy to types in computer programming languages.
A type will specify the Hilbert space associated with a connection, as well as whether the Hilbert space is of classical or quantum nature.

Preservation of types during rewriting operations will ensure that each rewrite is valid as well as completely specify any ambiguity.

\begin{definition}[Types]\label{def:type}
  Fix a Hilbert space \(\hilbert\) of \(n\) qubits and \(\cqvar \in \{\typeQ, \typeC\}\) (for quantum or classical).
  The corresponding \textbf{type} \(\typeVar\) is the tuple \((\hilbert, \cqvar)\).
  We say that such a type is quantum or classical, respectively.
  We use \(\mathcal{H}_{\typeVar}\) to denote the Hilbert space \(\mathcal{H}\) associated with the type \(\typeVar\).
  The set of all types is denoted \(\typeSet\), and, for a finite sequence of types \(\typeVar[\bullet] = (\typeVar[1], \dots, \typeVar[N])\), we use \(\mathcal{H}_{\typeVar[\bullet]}\) to denote \(\mathcal{H}_{\typeVar[1]} \otimes \dots \otimes \mathcal{H}_{\typeVar[N]}\).
\end{definition}
\nomenclature[B, 07]{\(\typeVar,\typeVar[\bullet]\)}{Type of a connection, and a sequence of types, see~\cref{def:type}.}

To each Hilbert space we will implicitly associate a complete set of orthogonal states that we call \textbf{computational basis states}.

We will want to describe computations with a classical component.
However, we need to model the state of such computations quantum-mechanically.
These are the classical-quantum states.
\begin{definition}[Classical-quantum states]
  Let \(\typeVar[\bullet]\) be a sequence of types, and decompose \(\mathcal{H}_{\typeVar[\bullet]}\) into two subsystems \(CQ\) where the subsystems of classical type are in \(C\) and the subsystems of quantum type are in \(Q\).
  A pure state \(\ket{\psi} \in \mathcal{H}_{\typeVar[\bullet]}\) is said to satisfy the type \(\typeVar[\bullet]\) if \(C\) is in a classical basis state.
  That is, for each classical subsystem \(i \in C\), there is a computational basis state \(\ket{x_i} \in \mathcal{H}_{t_i}\) such that
  \begin{align}
    \tr_Q \ketbra{\psi}{\psi} = \otimes_{i \in C} \ketbra{x_i}{x_i}
  \end{align}
  In other words, the classical subsystem is in a well defined computational basis state.
  A mixed state \(\rho\)
  is said to satisfy the type \(\typeVar[\bullet]\) if it is a convex combination of pure states satisfying \(\typeVar[\bullet]\).
  We use the notation \(\mixedstateType{\typeVar[\bullet]}\) or simply \(\mixedstate{\typeVar[\bullet]}\) to refer to the set of all mixed states of type \(\typeVar[\bullet]\).
\end{definition}
\nomenclature[B, 08]{\(\mathcal{H}_{\typeVar[\bullet]}, \mixedstateType{\typeVar[\bullet]}\)}{Pure states and mixed states of type \(\typeVar[\bullet]\). We abbreviate \(\mixedstateType{\typeVar[\bullet]}\) as \(\mixedstate{\typeVar[\bullet]}\). }

Note that these states are separable between the classical and quantum subsystems.

We now introduce operators and superoperators.
For a complete discussion see \cite{kitaev2002classical}.
\begin{definition}[Operators and Superoperators]\label{def:operators}
    Let \(\mathcal{H}\) and \(\mathcal{H}'\) be two finite-dimensional Hilbert spaces.
    We use \(\linearOp[\mathcal{H}']{\mathcal{H}}\) to refer to the space of linear operators (i.e. matrices) mapping from \(\mathcal{H}\) to \(\mathcal{H}'\).
    We define the shorthand \(\linearOp{\mathcal{H}} \equiv \linearOp[\mathcal{H}]{\mathcal{H}}\).
    The subset of positive semi-definite (PSD) trace-1 operators (density matrices) on \(\mathcal{H}\) is denoted \(\mixedstateArb{\mathcal{H}}\).

    We define the space of superoperators \(\superOp[\mathcal{H}']{\mathcal{H}}\) to be the set of linear maps on operators that maps from \(\linearOp{\mathcal{H}}\) to \(\linearOp{\mathcal{H}'}\).
    We say that a superoperator \(Q \in \superOp[\mathcal{H}']{\mathcal{H}}\) is completely positive and trace-preserving (CPTP) if
    \begin{itemize}
        \item For all operators \(O \in \linearOp{\mathcal{H}}\), the trace of \(O\) is preserved in the sense that \(\tr(O) = (\tr \circ Q)(O)\).
        \item For all Hilbert spaces \(\mathcal{H}_E\), \(Q\) extended by the identity on the auxillary Hilbert space \(\mathcal{H}_E\) maps PSD operators to PSD operators.
        That is, for all positive semi-definite operators \(O \in \linearOp{\mathcal{H}\otimes \mathcal{H}_E}\), \((Q\otimes\mathsf{Id}_{\mathcal{H}_E})(O)\) is positive semi-definite where \(\mathsf{Id}_{\mathcal{H}_E}\) is the identity superoperator.
    \end{itemize}
\end{definition}

\begin{definition}[Operator Types]\label{def:operator-type}
    For an operator \(O \in \linearOp[\mathcal{H}']{\mathcal{H}}\) and two sequences of types \(\typeVar[\bullet]^{\text{in}} = (\typeVar[1]^{\text{in}}, \ldots, \typeVar[N]^{\text{in}})\) and \(\typeVar[\bullet]^{\text{out}} = (\typeVar[1]^{\text{out}}, \dots, \typeVar[M]^{\text{out}})\), \(O\) is said to have input type \(\typeVar[\bullet]^{\text{in}}\) and output type \(\typeVar[\bullet]^{\text{out}}\) if 
    \begin{itemize}[topsep = 0pt]
        \item \(\mathcal{H} = \mathcal{H}_{\typeVar[\bullet]^{\text{in}}}\) and \(\mathcal{H}' = \mathcal{H}_{\typeVar[\bullet]^{\text{out}}}\)

        \item Every pure state \(\ket{\psi} \in \mathcal{H}_{\typeVar[\bullet]^{\text{in}}}\) of type \(\typeVar[\bullet]^{\text{in}}\) is mapped to a (not necessarily normalized) state \(\ket{\phi} = O \ket{\psi} \in \mathcal{H}_{\typeVar[\bullet]^{\text{out}}}\) of type \(\typeVar[\bullet]^{\text{out}}\).
    \end{itemize}
    The set of operators of input type \(\typeVar[\bullet]^{\text{in}}\) and output type \(\typeVar[\bullet]^{\text{out}}\) will be written as \(\linearOpType[\typeVar[\bullet]^{\text{out}}]{\typeVar[\bullet]^{\text{in}}}\).
\end{definition}
\nomenclature[B, 09]{\(\linearOpType[\typeVar[\bullet]^{\text{out}}]{\typeVar[\bullet]^{\text{in}}}\)}{Set of operators of input type \(\typeVar[\bullet]^{\text{in}}\) and output type \(\typeVar[\bullet]^{\text{out}}\), see~\cref{def:operator-type}.
We use the shorthand \(\linearOp{\mathcal{H}}, \linearOp{\typeVar[\bullet]}\) when the input and output spaces are the same.}

\begin{definition}[Gate]\label{def:gate}
    For two sequences of types \(\typeVar[\bullet]^{\text{in}} = (\typeVar[1]^{\text{in}}, \ldots, \typeVar[N]^{\text{in}})\) and \(\typeVar[\bullet]^{\text{out}} = (\typeVar[1]^{\text{out}}, \dots, \typeVar[M]^{\text{out}})\) and an superoperator \(\g\),
    \(\g\) is said to be a \textbf{gate} with \textbf{input type} \(\typeVar[\bullet]^{\text{in}} \) and \textbf{output type} \(\typeVar[\bullet]^{\text{out}}\) if it has a Kraus representation \(\{K_\mu\}_\mu\) where each Kraus operator \(K_\mu\) has the corresponding input type and output type \(K_\mu \in \linearOpType[\typeVar[\bullet]^{\text{out}}]{\typeVar[\bullet]^{\text{in}}}\).
    For short, we say that the gate has type \((\typeVar[\bullet]^{\text{in}}, \typeVar[\bullet]^{\text{out}})\).
    We use \(\superOpType[\typeVar[\bullet]^{\text{out}}]{\typeVar[\bullet]^{\text{in}}}\) to denote the set of all such superoperators.

    We call \(\g\) a \textbf{physical gate} if it is additionally completely-positive and trace-preserving (CPTP).
\end{definition}
\nomenclature[B, 10]{\(\superOpType[\typeVar[\bullet]^{\text{out}}]{\typeVar[\bullet]^{\text{in}}}\)}{Set of superoperators of input type \(\typeVar[\bullet]^{\text{in}}\) and output type \(\typeVar[\bullet]^{\text{out}}\), see~\cref{def:gate}.}
\nomenclature[B, 11]{\(\g\)}{A gate, which is an element in \(\superOpType[\typeVar[\bullet]^{\text{out}}]{\typeVar[\bullet]^{\text{in}}}\).}

\begin{remark}
    We will use ``non-physical'' gates to describe general and often adversarial models of faults in a circuit, and later prove thresholds against such faults.
\end{remark}

\begin{figure}
    \centering
    \includegraphics[width=0.95\linewidth]{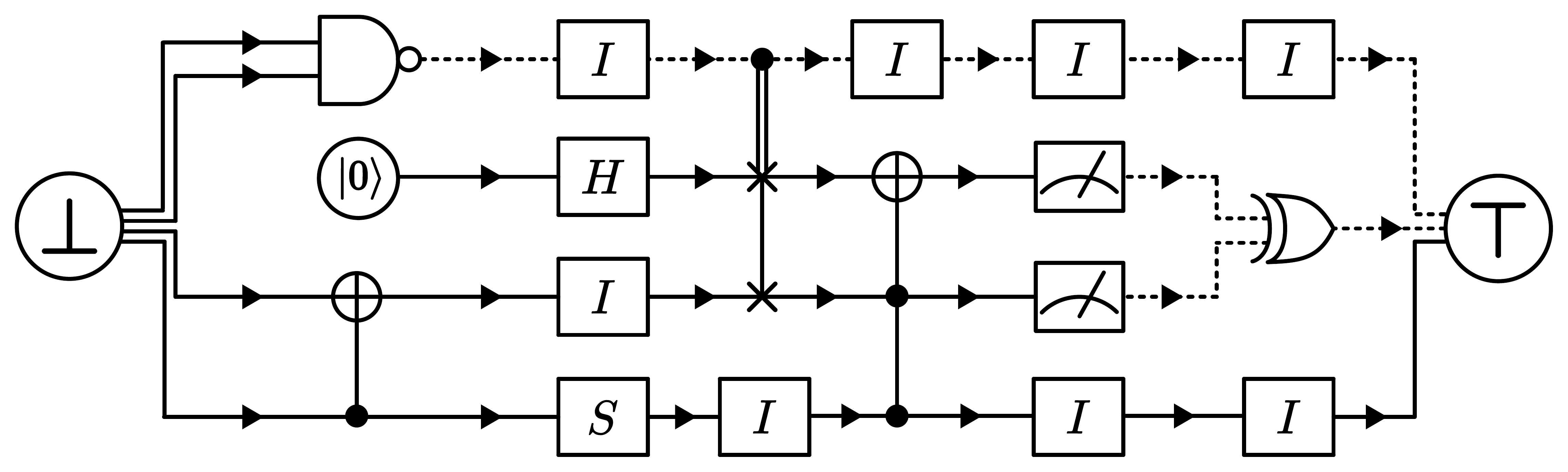}
    \caption{Illustration of a classical-quantum circuit. Taking two qubits and two bits as input and outputting one qubit and two bits.
    Types of edges and connections are annotated in red.}
    \label{fig:cq-circuit}
\end{figure}

We are now ready to state the central definition of our framework, a classical-quantum circuit (``quantum circuit'' or just ``circuit'' for short).
\begin{definition}[Classical-quantum circuit]\label{def:cq-circuit}
    Let \((V,E, \edgeConn)\) be a foliated network on the vertices \(V\) which we refer to as \textbf{locations}.
    Let \(\edgeType \colon E \to \typeSet\) be an assignment of a type to each edge.
    For each non-IO vertex \(v \in V \setminus \{\bot, \top\}\), we assign a classical-quantum gate \(\nodeGate[v]\) (\cref{def:gate}) such that \(\nodeGate[v]\) satisfies the type of the corresponding input and output edges.\footnote{Formally, in the sequence of input types, \(\typeVar[i]^{\text{in}} = \edgeType(e)\) when \(\edgeDest(e) = v\) and \(\edgeConn(e) = (i, \cdot)\). Likewise, in the sequence of output types, \(\typeVar[j]^{\text{out}} = \edgeType(e)\) when \(\edgeSrc(e) = v\) and \(\edgeConn(e) = (\cdot, j)\).}
    For convenience, we define \(\nodeGate[\bot]\) and \(\nodeGate[\top]\) to be the identity -- these vertices correspond to inputs and outputs of the circuit.
    The data \(C = (V,E,\edgeConn, \edgeType, \nodeGate)\) is said to be a classical-quantum circuit.
\end{definition}
\nomenclature[C, 01]{\(\edgeType\)}{A labeling that assigns a type to every connection in a network.}
\nomenclature[C, 02]{\(\nodeGate\)}{A labeling that assigns a gate to every vertex in a network, which has the correct types according to \(\edgeType\).}
\nomenclature[C, 03]{\(C\)}{A classical-quantum circuit is fully specified by \(C = (V,E,\edgeConn, \edgeType, \nodeGate)\), see~\cref{def:cq-circuit}. We often refer to the vertices and edges of a circuit as locations and wires. }

Recall that the depth of a network was defined in \cref{def:foliated-network}.
\begin{definition}[Circuit width]
  For a classical quantum circuit \(C = (V,E,\edgeConn, \edgeType, \nodeGate)\), we say that a vertex is a \textbf{quantum location} if any edge \(e\) incident to it is of quantum type.
  Likewise, a vertex is a \textbf{classical location} if all edges \(e\) incident to it are of classical type.

  The quantum width \(W_t\) of a timestep \(t\) of a circuit is the sum of the maximum over quantum inputs and outputs of each vertex.
  We introduce the following abuse of notation, for a vertex \(v \in V\) and an index \(i \in \nodeIn[v]\), \((v,i)\) uniquely specifies the edge \(e\) with source \(\edgeSrc[e] = v\) and \(\edgeConn[e] = (i,\cdot)\).
  Here, we will label edges by such pairs and an analogous definitions defined for a vertex \(v \in V\) and an index \(j \in \nodeOut[v]\), \((j,v)\).
  \begin{align}
    W_t = \sum_{v \in \nodeTimeInv[t]} \max\left(|\{i \in \nodeIn[v] \mid \edgeType((v,i)) = (\cdot, \typeQ)\}|, |\{j \in \nodeOut[v] \mid \edgeType((j,v)) = (\cdot, \typeQ)\}|\right)
  \end{align}
  The quantum width of a circuit \(W\) is the maximum quantum width in any timestep, \(W = \max_{t \in [D]} W_t\).
  The classical width of a circuit \((W_{\typeC})\) is defined analogously with \(\typeQ\) replaced by \(\typeC\).
  We use the term ``total width'' to refer to the sum \(W +W_{\typeC}\) of quantum and classical widths.
  In what follows, we will simply refer to the quantum width as the width.
\end{definition}
\nomenclature[C, 03]{\(W, W_{\typeC}\)}{The quantum width and classical width of a circuit.}

Note that a classical-quantum circuit is \emph{not} a channel -- many distinct circuits implement the same channel.
However, these circuits may have very different fault-tolerance properties.
We now describe the channel implemented by the circuit.
\begin{definition}[Circuit implementation]\label{def:circuit-imp}
  To every classical-quantum circuit \(C\) specified by the tuple of data \((V,E,\edgeConn, \edgeType, \nodeGate)\), we can define a superoperator in the following way: Let \((v_1, \dots, v_{|V|})\) be the ordered locations (vertices)  of \(C\).
  Then, we define \(\circuitMap{C}\) to be the superoperator
  \begin{align}
    \circuitMap{C} = \nodeGate[v_{|m|}] \circ \nodeGate[v_{m-1}] \circ \dots \circ \nodeGate[v_2] \circ \nodeGate[v_1]
  \end{align}
  Where, for each \(t \in [m]\), \(\nodeGate[v_t]\) is suitably extended by a tensor product with the identity in the canonical way, and for each edge \((u,v) \in E\) with \((i,j) \in \edgeConn\), the \(j\)-th tensor factor in the input Hilbert space of \(\nodeGate[v]\) is identified with the \(i\)-th tensor factor in the output Hilbert space of \(\nodeGate[u]\).
  The ordering of tensor factors for the input \(\nodeGate[\bot]\) and output \(\nodeGate[\top]\) is the same as their output or input, respectively.
\end{definition}
\nomenclature[C, 04]{\(\circuitMap{C}\)}{The map implemented by circuit $C$, see~\cref{def:circuit-imp}.}

In order to work with sets of gates (e.g. gadgets), we will also introduce vertex contractions.
Informally, a circuit contraction is simply a circuit where all the gates in a subset are grouped together into a single big gate.
\begin{definition}[Circuit Contraction]\label{def:circuit-contraction}
    For a classical-quantum circuit \(C = (V,E,\edgeConn, \edgeType, \nodeGate)\) with foliation \(\nodeTime: V\rightarrow [D]\), we can \textbf{contract} vertices together to form a new circuit. 
    More precisely, let \(\contractV\) be a set of vertices and let \(\contract:V\rightarrow \contractV\) be a surjective map. 
    \(\contract\) induces vertex contractions of \(C\): for every vertex \(u\in \contractV\), we contract the vertices \(\contract^{-1}(u)\) which are mapped to \(u\) to form a new network \((\contractV, \contractE)\). Note that $|E|=|\contractE|$.
    We say that the contraction is \textbf{valid} if this new network is acyclic. 
    Vertex contractions induce new labelings \(\contractConn, \contractType, \contractGate\), and therefore a new circuit
    \(\contractC = (\contractV,\contractE,\contractConn,\contractType,\contractGate)\). See~\Cref{fig:contracted-circuit} for an illustration.

    We say that a contraction $\contract$ is \textbf{depth-preserving} if the contracted circuit \(\contractC\) has a foliation \(\contractT:\contractV\rightarrow [D]\) of the same depth \(D\) such that every vertex \(v\in V\) is contracted with other vertices in the same timestep, namely, 
    \begin{align}
        \contractT(\contract(v)) = \nodeTime(v).
    \end{align}
    More generally, we say that a contraction $\contract$ is \textbf{foliation-preserving} if the contracted circuit \(\contractC\) has a foliation \(\contractT:\contractV\rightarrow [\contractD]\) such that for all timesteps $t\in [D]$, there is a new timestep $\contractt\in [\contractD]$ where all vertices in $V$ at timestep $t$ are contracted into. Precisely, 
    \begin{align}
        \forall t\in [D], \exists \contractt \in [\contractD] \quad \contract(\nodeTime^{-1}(t)) \subseteq \contractT^{-1}(\contractt)~.
    \end{align}
    In this work, we require that all contractions are valid.
\end{definition}

\nomenclature[C, 05]{\(\contract, \contractC\)}{contraction map and contracted circuit, see~\cref{def:circuit-contraction}.}

\begin{figure}[t]
    \centering
    \includegraphics[width=0.9\linewidth]{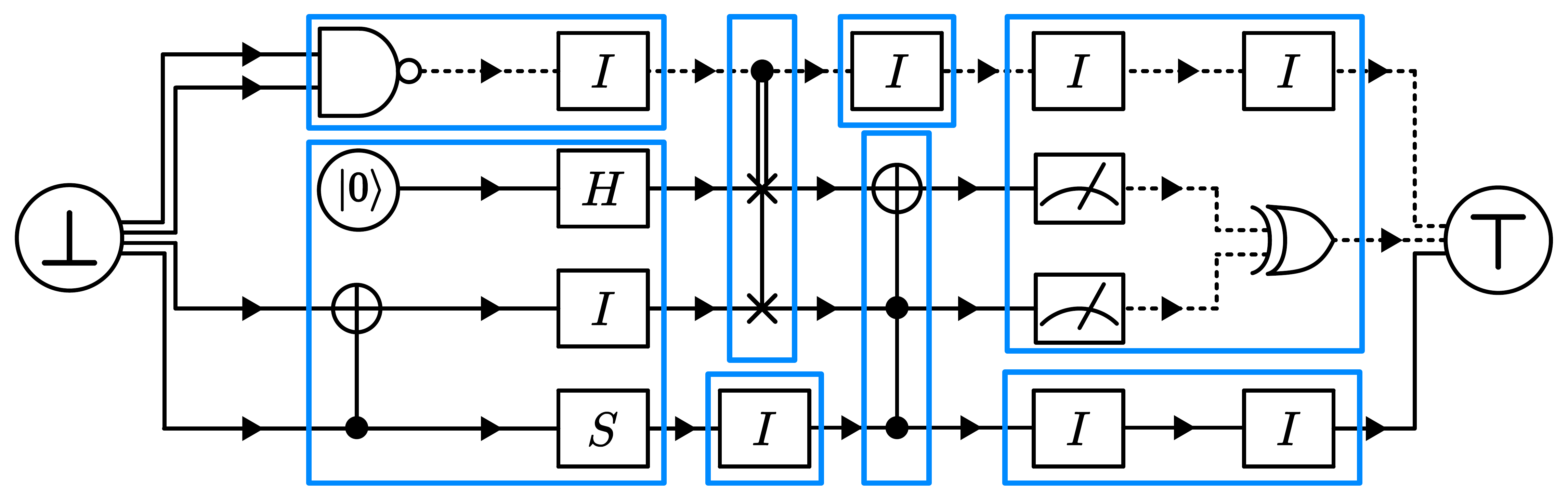}
    \caption{Illustration of a circuit contraction. Note that the indexing of inputs and outputs is modified.}
    \label{fig:contracted-circuit}
\end{figure}

The notion of circuit contraction is employed to define faults below. On a high level, the reason for this is that, in our noise model, the faulty locations of the circuit may be replaced by an arbitrary (global) superoperator.
In order to capture faults that perform operations between parts of the circuit, we allow the fault to act on a contracted circuit.
\begin{definition}[Faults]\label{def:fault}
  Fix a classical-quantum circuit \(C = (V,E,\edgeConn, \edgeType, \nodeGate)\) and ordered locations \((v_1, \dots, v_{|V|})\).
  A \textbf{fault} \(\fault=(\contract, \nodeFaulty)\) consists of a depth-preserving
  contraction map \(\contract \colon V \to \contractV\) and a (faulty) gate map \(\nodeFaulty\) with domain \(\contractV\) such that the contracted circuit \(\contractC = (\contractV,\contractE,\contractConn,\contractType,\contractGate)\) is a valid classical-quantum circuit when the induced gate map \(\contractGate\) is replaced by the faulty gate map \(\nodeFaulty\).
  We say that the support of \(\nodeFaulty\), known as the \textbf{fault path}, is the set of vertices in \(V\) that are contracted non-trivially. Precisely, 
  \begin{align}
    \supp (\nodeFaulty) = \{v \in V \mid \nodeFaulty(\kappa(v)) \ne \nodeGate(v)\}.
  \end{align}
    
  From now on, we use \(\fault\) to refer to a set of data \((\contract, \nodeFaulty)\) compatible with \(C\) and \(C[\fault]\) to refer to the faulty circuit \((\contractV,\contractE,\contractConn,\contractType,\nodeFaulty)\). 
  We write \(\supp (\fault)\) as an alias for \(\supp (\nodeFaulty)\).

  In what follows, we will mostly concern with faults supported only on quantum locations.
  This restriction is not necessary as classical fault-tolerance can be constructed and analyzed in our framework as well.
\end{definition}
\nomenclature[C, 06]{\(\fault = (\contract, \nodeFaulty)\)}{A fault in a circuit is defined by a depth-preserving contraction map \(\contract\) and a set of faulty gates \(\nodeFaulty\) which replaced the ideal gates, see~\cref{def:fault}.}
\nomenclature[C, 07]{\(\supp (\fault)\)}{The fault path of a fault \(\fault\), i.e., the locations of the faulty gates.}

It is common to first prove fault-tolerance against insertions of Pauli errors into the circuit and then reduce the general case to the Pauli error case.
Thus we will define a notion of a \emph{Pauli fault}.
\begin{definition}[Pauli superoperator]\label{def:pauli-superoperator}
    A superoperator \(\error\) is said to be a \textbf{Pauli superoperator} if, for some Pauli operators \(P,P'\) and scalar \(c \in \mathbb{C}\), \(\error\) can be written as
    \begin{align}
        \error(\rho) = c P \rho P'
    \end{align}
    If \(P = P'\) then \(\error\) is said to be a \textbf{diagonal Pauli superoperator}.
\end{definition}
As we allow \(P, P'\) to differ,  \cref{def:pauli-superoperator} will naturally capture how general errors ``decohere'' to Pauli errors (diagonal Pauli superoperators).
In general, the strategy to prove fault tolerance will be to decompose general faults into (non-diagonal) \emph{Pauli faults} which apply Pauli superoperators (see \cref{lemma:fault-decomposition}).
For example, consider the following channel that applies an \(S\) gate with probability \(p\).
\begin{align}
    \rho \mapsto (1-p) \rho + p \PHAS \rho \PHAS^{\dagger}
\end{align}
Defining the coefficient \(\theta = \frac{e^{i \pi / 4}}{\sqrt{2}}\), we can decompose \(\PHAS\) as \(\PHAS =  \theta \mathsf{I} + \theta^*\mathsf{Z}\) and write the channel as a sum of Pauli superoperators
\begin{align}
    \rho \mapsto \left(1-\frac{p}{2}\right) \mathsf{I}\rho \mathsf{I} + \frac{p}{2} \mathsf{Z}\rho \mathsf{Z} + \frac{ip}{2}\mathsf{I}\rho \mathsf{Z} - \frac{ip}{2}\mathsf{Z}\rho \mathsf{I}~.
\end{align}

\begin{definition}[Pauli fault]\label{def:pauli-fault}
    Fix a classical-quantum circuit \(C = (V,E,\edgeConn, \edgeType, \nodeGate)\), and a fault \(\fault = (\kappa \colon V \to V_\kappa, \nodeFaulty)\).
    The fault \(\fault\) is said to be a \textbf{Pauli fault} if the faulty gate at every contracted vertex \(v \in V_\kappa\) can be written as the original gate with Pauli superoperators occurring before and after.
    That is, for some Pauli superoperators \(\mathcal{P}_v\), \(\mathcal{P}_v'\)
    \begin{align}
        \nodeFaulty(v) = \mathcal{P}'_v \circ \nodeGate_\kappa(v) \circ \mathcal{P}_v~.
    \end{align}
    If this property holds for diagonal Pauli superoperators \(\mathcal{P}_v\) and \(\mathcal{P}'_v\), then \(\fault\) is said to be a \textbf{diagonal Pauli fault}.
\end{definition}

While allowing Pauli superoperators to act both before and after an ideal gate may seem somewhat strange, it captures many faults that a less general definition (such as limiting Pauli superoperators to only act after an ideal gate) cannot.
Consider, for example, a fault that replaces a reset gate by the identity gate.
By utilizing extra ancilla (e.g. by swapping qubits to the ancilla, apply any noise superoperator, and swapping back), we can decompose the map of a circuit with an arbitrary fault into a sum of maps of circuits with this ancilla and Pauli faults (see \cref{lemma:fault-decomposition}).
Such ancilla will be part of an \emph{environment circuit}, see~\cref{def:circuit-environment}.

\begin{remark}[Pauli fault]
    It is a crucial property that for Pauli faults, the contraction can always be taken to be trivial.
    More precisely, since all Pauli superoperators can be written as the tensor product of single-qubit Pauli superoperators, there always exists an equivalent fault on the uncontracted vertex set.
    This property will allow us to analyze gadgets in isolation despite faults inserting superoperators with support acting between gadgets.
    The presence of the ideal gate allows the behavior of the faulty circuit to be compared to the ideal one.

    As we will see in \cref{lemma:level-reduction}, it is unfortunately the case that physical level diagonal Pauli faults do not imply that ``logical level'' faults are diagonal or even Pauli.\footnote{For example, the failure of a magic state consumption gadget caused by a measurement error.}
\end{remark}

In order to fully specify the operation of the circuit with noise, we need a model of the environment.
In what follows, the environment should be thought of as a book-keeping method for tracking an unspecified second quantum circuit that is essentially unrestricted.
\begin{definition}[Circuit with environment]\label{def:circuit-environment}
    For a classical-quantum circuit \(C\) specified by the data \((V,E,\edgeConn, \edgeType, \nodeGate)\), the circuit \(C\) \textbf{with environment} refers to an (arbitrary) extension \(\Cwithenv = (V\sqcup \environV,E\sqcup \environE,\edgeConn',\edgeType',\nodeGate')\) of the original circuit \(C\) where the sub-circuit induced by the nodes \(V\) of \(\Cwithenv\) is the original circuit \(C\).
    \(C\) is called the \textbf{computational circuit} of \(\Cwithenv\), and the sub-circuit \(\environC\) induced by the nodes \(\environV\) is called the \textbf{environment circuit} (of \(\Cwithenv\)).

    Frequently, we will assume that the environment circuit is \textbf{closed}, which means it has no input or output edges. In other words, the environment circuit introduces its own workspace and traces out its output at the end.
\end{definition}
\nomenclature[C, 08]{\(\Cwithenv,\environC\)}{Circuit \(C\) with environment refers to an (arbitrary) extension \(\Cwithenv = (V\sqcup \environV,E\sqcup \environE,\edgeConn',\edgeType',\nodeGate')\) of the original circuit \(C\). The sub-circuit \(\environC\) induced by the nodes \(\environV\) is called the {environment circuit} of \(\Cwithenv\).}

For a circuit with environment, its width (depth) refers to the width (depth) of the computational circuit.

The apparent imprecision of the environment circuit is deliberate.
The fault-tolerance statements we prove will hold for arbitrary environment circuits.
Note that as stated, there is no interaction between the computational and environment circuits because the vertex and edge sets are disjoint. The environment circuit and computational circuits may only interact via faults which contracts vertices between the two (otherwise independent) networks.

\begin{remark}[Adversarial faults]
    A fault should be thought of as replacing the gates at locations in the fault path with \emph{arbitrary} superoperators.
    These superoperators, which may be supported on the environment, are not necessarily independent between locations. 
    In other words, they may entangle the qubits of distinct locations of the original circuit.
    This models an arbitrarily powerful adversary with a memory that is permitted access to the part of the computational state supported in the fault every timestep. 
\end{remark}

\section{Fault-tolerant gadgets}\label{sec:gadgets}

For a fault-tolerant circuit, the connections in the circuit carry quantum information encoded in quantum error correcting codes. 
Therefore, a collection of connections could have a \textbf{code type}, which specifies the encoding and protection of the carried information. 
Such code types will be convenient for working with gadgets since a general fault-tolerance scheme may employ many different codes distinct from the code protecting the bulk of the computation in different places (for example code switching or preparing resource states).
This feature will allow our statements to be agnostic to the choice of code used.
To define code types, we will need to first characterize physical errors on quantum states.

Our notion of a damaged state is partly inspired by \cite{kitaev1997quantum} and \cite{aharonov1997fault}.
\begin{definition}[Deviation]\label{def:deviation}
  For a set of qubits \(A\), a family of bad sets \(\baderrors\subseteq \power(A)\), and two states \(\sigma\), \(\tilde{\sigma}\) on \(A\),
  \(\tilde{\sigma}\) is said to be \textbf{\(\baderrors\)-deviated} from \(\sigma\) if there exists a \(\baderrors\)-avoiding set \(S \subseteq A\) and a (not necessarily CPTP) superoperator \(\error_S\) supported only on \(S\) such that \(\tilde{\sigma} = \left(\mathcal{I}_{A \setminus S}\otimes \error_S\right)(\sigma)\). 
  We denote this relationship by 
  \begin{align}
    \sigma \dev{\baderrors} \tilde{\sigma}.      
  \end{align}
  For a subspace \(\qcode\) of the Hilbert space, a state \(\tilde{\sigma}\) is said to be \(\baderrors\)-deviated from \(\qcode\) if there exists a (possibly mixed) state \(\sigma\) in \(\qcode\) such that \(\tilde{\sigma}\) is \(\baderrors\)-deviated from \(\sigma\).
  We use the less general term ``Pauli \(\baderrors\)-deviated'' (``diagonal Pauli \(\baderrors\)-deviated'') when the superoperator \(\error_S\) is a Pauli superoperator (diagonal Pauli superoperator).
\end{definition}
\nomenclature[D, 01]{\(\sigma \dev{\baderrors} \tilde{\sigma}\)}{Notation to represent that for two states \(\sigma\) and \(\tilde{\sigma}\), \(\tilde{\sigma}\) is \(\baderrors\)-deviated from \(\sigma\), see~\cref{def:deviation}.}

In general, our circuits may have data encoded in different quantum codes.
Formally, we attach this data to groups of qubits using \emph{code types}.
The code types will track the code (including a choice of logical basis) and any demands that will be made on damaged code states.

\begin{definition}[Code Type]\label{def:code-type}
    For a \(n\) qubit Hilbert space, let $r = n-k$, a quantum code \(\qcode\) with \(k\) logical qubits is specified by a \textbf{decoding unitary} \(\decU: \hilbert^n\rightarrow \hilbert^k\otimes \hilbert^{r}\) and a family of \textbf{bad error supports} \(\baderrors \subseteq \power([n])\) such that
    \begin{itemize}
        \item The codespace \(\qcode\) is specified by \(\encU\), the \textbf{encoding unitary}, acting on the \(k\) qubits logical information and the trivial syndrome,
        \begin{align}
            \qcode = \{\encU \ket{\psi}\ket{0^{r}} | \ket{\psi}\in \hilbert^k \}.            
        \end{align}
        We define the \textbf{reversible encoding channel} as \(\Rencoder(\rho) = \encU \rho \decU\) and the \textbf{reversible decoding channel} as its inverse \(\Rdecoder(\sigma) = \decU \sigma \encU\).

        \item For any logical state \(\rho\in \mixedstateArb{\hilbert^k}\), for any state  \(\tilde{\sigma}\) that is \(\baderrors\)-deviated from an encoded logical state \(\sigma = \encU (\rho\otimes \ketbra{0^{r}}{0^{r}}) \decU\), the decoded state\footnote{{The term ``state'' is an abuse of terminology, since in general, the result is not completely positive or trace-1.}} is proportional to the logical state.\footnote{In particular, this implies that errors on \(\baderrors\)-avoiding sets are correctable and that the proportionality constant may only depend on the error.}
        \begin{align}\label{eq:decoder-recovers}
            \tr_{\hilbert^{r}}(\Rdecoder(\tilde{\sigma})) \propto \rho
        \end{align}
        We define the \textbf{irreversible encoding channel} as \(\encoder(\rho) = \Rencoder(\rho\otimes \ketbra{0^{r}}{0^{r}})\), and the \textbf{irreversible decoding channel} as \(\decoder(\sigma) = \tr_{\hilbert^{r}}(\Rdecoder({\sigma}))\). 
    \end{itemize}
    A \textbf{code type} of \(\qcode\) is the tuple of data \(\ctype[\qcode] = (\decU, \baderrors)\). 
    The \textbf{logical type} of \(\ctype[\qcode]\) is the type corresponding to an unencoded register \((\mathcal{H}^k, \cqvar)\).
    Likewise the \textbf{syndrome type} of \(\ctype[\qcode]\) is the type corresponding to the syndrome needed to make the encoding map unitary \((\mathcal{H}^r, \cqvar)\).
    
    For convenience, we define a \textbf{trivial code type} \(\ctype[I] = (I, \{ \{1\}, \dots, \{n\} \})\) to indicate that the qubits (or bits) are unencoded and unprotected. 
    Similarly, we can trivially extend a code type to an arbitrary number of unencoded qubits (or bits) by extending the unitary \(U\) with identity operators.
\end{definition}
\nomenclature[D, 02]{\(\qcode, (U,\baderrors)\)}{A quantum code \(\qcode\) with code type \(\ctype[\qcode] = (U,\baderrors)\), where \(U\) is the decoding unitary and \(\baderrors\) is a collection of bad error supports, see~\cref{def:code-type}. We sometimes write \(\ctype[\qcode]\)-deviated and \(\dev{\ctype[\qcode]}\) in place of \(\baderrors\)-deviated and \(\dev{\baderrors}\).}
\nomenclature[D, 03]{\(\Rencoder,\Rdecoder\)}{Reversible encoders and decoders of a code.}
\nomenclature[D, 04]{\(\encoder,\decoder\)}{Irreversible encoders and decoders of a code.}
\begin{remark}\label{rmk:decoder-unitary}
    Typically, in quantum error correction, we think of an encoding unitary as Clifford.
    However, in this case, the constraint that \(\Rdecoder\) can successfully recover states that are \(\baderrors\)-deviated from the codespace means that \(\Rdecoder\) and its inverse, \(\Rencoder\) is quite complicated in general.\footnote{Think of it as the purification of a channel that measures the syndrome and corrects the state.}
    Since \(\Rdecoder\) is a proof tool, this does not present any complications.
\end{remark}

For correctable states, the application of the decoder decouples the logical subsystem from the syndrome subsystem.
We will later use this property for showing that correctly operating gadgets decouple the computation from the action of the fault.
\begin{proposition}[Decoupling of Errors]\label{cor:decoupling-errors}
Let \(\qcode\) be a code on \(n\) qubits with code type \((U, \baderrors)\). For every superoperator \(\error\) supported on a subset of qubits \(B\subseteq [n]\) that is \(\baderrors\)-avoiding, there exists an operator \(\synd_\error\in \unphysicalstate{\mathcal{H}^{r}}\) such that
for all \(\rho\in \mixedstateArb{\hilbert^k}\), the decoded logical subsystem is unentangled from the syndrome subsystem
\begin{align}
    \Rdecoder\circ \error\circ \encoder(\rho) = \rho \otimes \synd_\error~. 
\end{align}
\end{proposition}
\begin{proof}
    Let \(E\) be a Pauli operator whose support is \(\baderrors\)-avoiding.
    Consider a logical state \(\ket{\psi}\in \hilbert^k\) and the decoded state \(\ket{\psi_E} = UEU^\dagger(\ket{\psi}\ket{0^r})\).
    From definition (\cref{def:code-type}) decoder map satisfies (\cref{eq:decoder-recovers}) that \(\tr_{\hilbert^{r}}(\ketbra{\psi_E}) \propto \ketbra{\psi}\).\footnote{In fact, we have equality since all superoperators in the expression are CPTP.}
    Since the result of the partial trace is a pure state, this implies that \(\ket{\psi_E}\) is unentangled between the logical and syndrome subsystems.
    That is, we can write \(\ket{\psi_E} = \ket{\psi} \ket{s_{\psi,E}}\) for some pure state \(\ket{s_{\psi,E}} \in \hilbert^r\) of the syndrome subsystem. 

    We now use linearity to show that \(\ket{s_{\psi,E}}\) is indepedent of \(\psi\).
    Consider a superposition of any two logical states \(\ket{\psi},\ket{\phi}\in \hilbert^k\) written as \(\ket{\varphi} = \alpha \ket{\psi} + \beta \ket{\phi}\).
    Applying the previous argument to \(\ket{\varphi}\), \(\ket{\phi}\) and \(\ket{\psi}\), we have that
    \begin{align}
        UEU^\dagger \left( \ket{\varphi} \ket{0^r} \right)
        &= \alpha \ket{\psi} \ket{s_{\psi,E}} + \beta \ket{\phi} \ket{s_{\phi,E}} \\
        UEU^\dagger \left( \ket{\varphi} \ket{0^r} \right) &= \left( \alpha \ket{\psi} + \beta \ket{\phi}\right) \ket{s_{\varphi,E}}~.
    \end{align}
    It must therefore be the case that \(\ket{s_{\psi,E}} = \ket{s_{\phi,E}} = \ket{s_{\varphi, E}}\). In other words, the syndrome state \(\ket{s_E}\) is independent of the logical state and only depends on the error \(E\). 

    We now consider the case of a more general error superoperator \(\error\) (not necessarily CPTP).
    \(\error\) can be written as a linear combination of Pauli superoperators by expanding the Kraus operators in a basis of Pauli matrices. 
    Therefore, we can write 
    \(\error(\sigma) = \sum_{\mu,\nu}\alpha_{\mu,\nu} K_\mu \sigma K_\nu' \)
    for Pauli operators \(\{K_\mu\}_\mu\) supported on \(B\) and complex coefficients \(\{\alpha_{\mu,\nu}\}_{\mu,\nu}\).
    Then, we can apply the previous argument to an arbitrary pure state \(\psi\), 
    \begin{align}
        \Rdecoder\circ \error\circ \encoder(\psi)
        &=  \sum_{\mu,\nu}\alpha_{\mu,\nu} U K_\mu U^\dagger (\psi\otimes \ketbra{0^r}) U K_\nu' U^\dagger \\
        &= \sum_{\mu,\nu}\alpha_{\mu,\nu} \psi\otimes \ketbra{s_\mu}{s_\nu} \\
        &= \psi\otimes \left( \sum_{\mu,\nu}\alpha_{\mu,\nu}\ketbra{s_\mu}{s_\nu} \right)
    \end{align}
    for syndrome states \(\ket{s_\mu},\ket{s_\nu}\) independent of \(\rho\). 
    The same conclusion holds for mixed states \(\rho\) by considering a purification.
\end{proof}
\nomenclature[D, 05]{\(\error, K_\mu, K_{\nu}'\)}{We use \(\error\) to denote an error superoperator acting on states. Such a superoperator can be written as a linear combination of Pauli superoperators, thereby admitting a decomposition \(\error(\sigma) = \sum_{\mu,\nu}\alpha_{\mu,\nu} K_\mu \sigma K_{\nu}'\)
for Pauli operators \(\{K_\mu\}_\mu, 
\{K_{\nu}'\}_\nu\) and complex coefficients \(\{\alpha_{\mu,\nu}\}_{\mu,\nu}\).}

\begin{figure}
    \centering
    \includegraphics[width=0.65\linewidth]{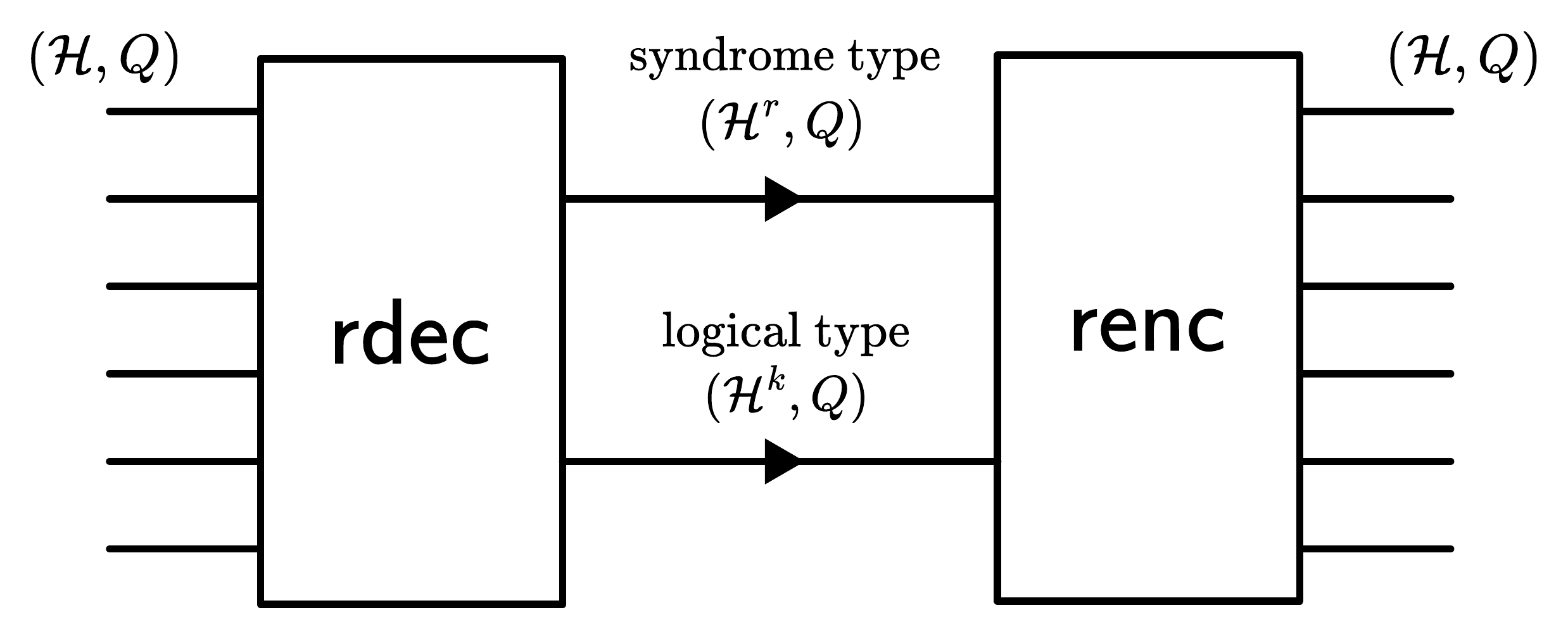}
    \caption{Reversible decoder and encoders of a quantum code.}
    \label{fig:decoder}
\end{figure}

The code types will need to be assigned to sets of edges in a circuit that share the same source and destination.
To do this, we introduce a notion of a ``bundle'' of edges.
\begin{definition}[Bundled Circuit]
\label{def:bundled-circuit}
    For a classical-quantum circuit \(C = (V,E,\edgeConn, \edgeType, \nodeGate)\), if there are multiple edges with the same source and destination vertices (which is common in contracted circuits), we may \textbf{bundle} these edges together, and label them with code specifications.
    Specifically, a bundling is a surjective function \(\bundle: E\rightarrow \bundleE\) from the edges of \(C\) to a set \(\bundleE\) called the \emph{bundles of C} such that all edges in a bundle \(b \in \bundleE\) share the same source, destination, and type. That is, for \(e_1, e_2\in \beta^{-1}(b)\)
    \begin{align}
          \edgeSrc(e_1) = \edgeSrc(e_2),\quad \edgeDest(e_1) = \edgeDest(e_2),\quad \edgeType(e_1) = \edgeType(e_2)
    \end{align}
\end{definition}

We can now assign code types to the bundles of a circuit.
\begin{definition}[Code Specification]\label{def:code-spec}
    For a circuit \(C = (V,E,\edgeConn, \edgeType, \nodeGate)\) and a bundling \(\bundle: E\rightarrow \bundleE\), we define a \textbf{code specification} to be a labeling \(\spec\) which assigns a code type to every edge bundle \(b \in \bundleE\) that is compatible in the sense that the edges of the edge bundle are in bijection with the qubits of the code type. 
    To reduce notation, this bijection will be implicit and clear from context.
\end{definition}
\nomenclature[D, 06]{\(\bundle, \spec\)}{\(\beta\) is a bundling of edges for a circuit \(C\). \(\spec\) is a code specification which assigns code types to bundled edges, see~\cref{def:bundled-circuit},~\cref{def:code-spec}.}

In order to concisely talk about input and output behavior of circuits with respect to inputs and outputs satisfying the code types, let us now introduce notation for the various data induced by the code specification.
\begin{definition}[Circuit Input/Output]\label{def:circuit-io}
    Let \(C = (V,E,\edgeConn, \edgeType, \nodeGate)\) be a classical-quantum circuit, and fix a bundling \(\bundle: E\rightarrow \bundleE\) and specification \(\spec\).
    In order to state various input and output properties of the circuit, we will introduce the following notations which are induced from the specification data in the obvious way.
    By input (output) code types, we mean the input (output) code types associated with each input (output) edge bundle.
    Likewise, we introduce the terms input/output logical (syndrome) types associated with each edge bundle.
    
    Suppose the circuit $C$ has \(\ell\) input and \(m\) output edge bundles, respectively.
    Let \((\baderrors^{\text{in}}_1, \dots, \baderrors^{\text{in}}_\ell)\) correspond to the bad error supports of the input code types.
    We define the \textbf{input bad error supports} \(\baderrors^{\text{in}}\) of \(C\) as the sum of the bad error supports of the input code type.\footnote{We leave the bijection between the input edges of \(C\) and the qubits of each bundle and code type implicit.}
    We define the \textbf{output bad error supports} \(\baderrors^{\text{out}}\) analogously with the bad error supports of the output code types  \((\baderrors^{\text{out}}_1, \dots, \baderrors^{\text{out}}_m)\).
    \begin{align}
        \baderrors^{\text{in}} = \baderrors^{\text{in}}_1 \boxplus \dots \boxplus \baderrors^{\text{in}}_\ell \\
        \baderrors^{\text{out}} = \baderrors^{\text{out}}_1 \boxplus \dots \boxplus \baderrors^{\text{out}}_m
    \end{align}
    Similarly, let \((\Rencoder^{\text{in}}_1, \dots, \Rencoder^{\text{in}}_\ell)\) correspond to the reversible encoders of the input code types.
    This induces an \textbf{input reversible encoder} \(\Rencoder^{\text{in}}\) for the whole circuit given by applying each reversible encoder to each edge bundle.
    Again, we define an analogous output reversible encoder \(\Rencoder^{\text{out}}\) using the reversible encoder of the output code types \((\Rencoder^{\text{out}}_1, \dots, \Rencoder^{\text{out}}_m)\).
    \begin{align}
        \Rencoder^{\text{in}} &= \Rencoder^{\text{in}}_1 \otimes \dots \otimes \Rencoder^{\text{in}}_\ell \\
        \Rencoder^{\text{out}} &= \Rencoder^{\text{out}}_1 \otimes \dots \otimes \Rencoder^{\text{out}}_m
    \end{align}
    Decoders and their irreversible counterparts are similarly defined in the obvious way.
\end{definition}

\subsection{Gadgets}
We are now ready to state our definition of a fault-tolerant gadget.
Under some circumstances, most commonly concatenating distance-3 codes \cite{aliferis2005quantum}, a circuit may satisfy multiple gadget definition for various combinations of parameters.
E.g. stronger conditions on the fault may imply stronger conditions on the output error or weaker conditions on the input error.

Our definition of gadget most closely aligns with \cite{kitaev1997quantum} and \cite{christandl2024fault} although we desire that our gadgets \emph{exactly} simulate the operation when the state is good. %
\begin{definition}[Fault-tolerant Gadget]\label{def:gadget}

    Let \(C = (V,E,\edgeConn, \edgeType, \nodeGate)\) be a classical-quantum circuit with input and output types \(\intype[\bullet], \outtype[\bullet]\). 
    Fix a bundling \(\bundle: E\rightarrow \bundleE\) and associated code specification \(\spec\).
    Let \(\badfaults\subseteq \power(V)\) be a collection of bad fault paths. 
    
    Let \(\Cwithenv\) be the computational circuit \(C\) with an arbitrary environment circuit \(\environC\) (recall~\cref{def:circuit-environment}) with input and output types \(\intype[\environ]\), \(\outtype[\environ]\).
    We use the circuit input/output definitions (input/output bad error supports \(\baderrors^{\text{in}}\) and \(\baderrors^{\text{out}}\), etc.) from \cref{def:circuit-io} extended to the environment circuit in the trivial way by assigning trivial code types to all edges of the environment circuit. 
    
    Let \(\g\) be a gate with input and output types \(\intype[\g], \outtype[\g]\). We say that \(C\) is a \textbf{fault-tolerant gadget} which simulates \(\g\) with respect to the data \((\bundle, \spec, \badfaults)\) if the input and output logical types of \(C\) (under \(\spec\)) match the input/output types of \(\g\) and the following conditions hold.
    For any \(\badfaults\)-avoiding (non-diagonal) \emph{Pauli} fault \(\fault\) supported only on the computational circuit,\footnote{Note that this assumption is reasonable since we will apply \cref{lemma:fault-decomposition} to the entire circuit being analyzed, which turns an arbitrary fault into Pauli faults.} there exists superoperators \(\tilde{\g}_{\fault}\in \superOpType[\typeVar[\bullet]^{\text{out}}]{\typeVar[\bullet]^{\text{in}}}\) and \(\recover_{\fault}\in \superOpType[\typeVar[\bullet]^{\text{in}}]{\typeVar[\bullet]^{\text{in}}}\) such that the circuit map can be decomposed as\footnote{\(I_{\environ}\) is the identity acting on the environment subsystem. }
    \begin{align}\label{eq:gadget}
        \circuitMap{\Cwithenv[\fault]} = \left(\tilde{\g}_{\fault}\otimes \circuitMap{\environC}\right)\circ \left(\recover_{\fault}\otimes I_{\environ}\right).
    \end{align}

    These superoperators satisfy the following properties. 
    \begin{itemize}[leftmargin=*, align = left]
        \item[\textbf{Friendly}] For an \emph{arbitrary} state \(\sigma\in \unphysicalstate{\intype[\bullet] \otimes \intype[\environ]} \), there exists a logical state \(\rho \in \mixedstate{{\intype[L]} \otimes {\intype[\environ]} }\) such that
        \begin{align}
            \encoder^{\text{in}}(\rho) \dev{\baderrors^{\text{in}}} (\recover_{\fault}\otimes I_{\environ}) (\sigma).
        \end{align}
        Moreover, \(\recover_{\fault}\) acts as logical identity on ``good'' input states that are \(\baderrors^{\text{in}}\)-deviated from the codespace. That is, for some logical state \(\rho \in \mixedstate{{\intype[L]} \otimes {\intype[\environ]} }\),
        \begin{align}
            \encoder^{\text{in}}(\rho) \dev{\baderrors^{\text{in}}} \sigma \quad \Longrightarrow \quad (\recover_{\fault}\otimes I_{\environ} )(\sigma) \dev{\baderrors^{\text{in}}} \sigma~.
        \end{align}
        
        We refer to the superoperator \(\recover_{\fault}\) as a \textbf{filter}.\footnote{This serves the same role as filters in \cite{aliferis2005quantum}.}
        Furthermore, \(\recover_{\fault}\) factorizes into operators acting on each input edge bundle.
        \begin{align}
            \recover_{\fault} = \recover_1 \otimes \dots \otimes \recover_\ell. \label{eq:decoupling:factorization}
        \end{align}

        \item[\textbf{Simulation}] For a ``good'' input state \(\sigma\) with \(\encoder^{\text{in}}(\rho)\dev{\baderrors^{\text{in}}} \sigma\), the output is \(\baderrors^{\text{out}}\)-deviated from the encoding of the logical operation \(\g\) applied to the logical state \(\rho\).
        \begin{align}\label{eq:simulation}
            \encoder^{\text{out}}\circ \left(\g\otimes \circuitMap{\environC}\right) (\rho) \dev{\baderrors^{\text{out}}} \left(\tilde{\g}_{\fault}\otimes \circuitMap{\environC} \right)(\sigma).
        \end{align}
    \end{itemize}
\end{definition}
\nomenclature[D, 10]{\(\badfaults, \tilde{\g}_{\fault}, \recover_{\fault}\)}{A circuit \(C\) is a fault-tolerant gadget that simulates \(\g\) with respect to the data \((\bundle,\spec,\badfaults)\) if for every \(\badfaults\)-avoiding fault \(\fault\), there exists superoperators \(\tilde{\g}_{\fault}, \recover_{\fault}\) such that \cref{eq:gadget} and the properties defined in \cref{def:gadget} hold.}

In the above definition, the friendly property can be dropped\footnote{In which case one can set \(\recover_{\fault}\) to be identity.} in many contexts where recursive simulation is not required. 
Friendly gadgets corresponds roughly to a ``strong'' simulation of \cite{kitaev1997quantum}.

With this definition, the key technical work in proving fault-tolerance becomes proving that various circuits satisfy this gadget definition with the desired code types. 
Later in \cref{sec:qldpc-gadgets}, we will show that the standard error correction circuit for a qLDPC code, which repeatedly measure stabilizers for \(O(d)\) rounds, is a friendly and fault-tolerant gadget which simulates identity. 
We then compose this gadget with other circuits that implement logical operations, such as transversal gates. 

We prove the following composition result \cref{lemma:gadget-composition} which allows gadgets to be assembled into new gadgets.
This lemma will be used in many locations in \cref{sec:qldpc-threshold-theorems} without reference.

\begin{proposition}[Composition of gadgets]\label{lemma:gadget-composition}
    Suppose \(C_1, C_2\) are fault-tolerant gadgets for gates \(\g_1, \g_2\) with respect to the data \((\contract_1, \spec_1, \badfaults_1)\) and \((\contract_2, \spec_2, \badfaults_2)\).
    \begin{itemize}[leftmargin=*, align = left]
        \item[\textbf{Parallel Composition}] \(C_1\otimes C_2\) is a fault-tolerant gadget for \(\g_1\otimes \g_2\), with respect to the data \((\contract_1\sqcup \contract_2, \spec_1\sqcup \spec_2, \badfaults_1\boxplus \badfaults_2)\), where \(\contract_1\sqcup \contract_2\) is the disjoint union of the two bundling and \(\spec_1\sqcup \spec_2\) is the disjoint union of the two specifications. 
        Additionally, if \(C_1\) and \(C_2\) are both friendly, then \(C_1\otimes C_2\) is friendly.

        \item[\textbf{Sequential Composition}] Suppose there is a bijection between the output bundled edges of \(C_1\) and the input bundled edges of \(C_2\) such that \(\spec_1,\spec_2\) assigns the same code type to edges paired in bijection. 
        Define \(C_2\circ C_1\) and \(\g_2\circ \g_1\) by joining the bundled edges according to the bijection.
        Then \(C_2\circ C_1\) is a fault-tolerant gadget for \(\g_2\circ \g_1\) with respect to the data \((\contract_1\cup \contract_2, \spec_1\cup \spec_2, \badfaults_1\boxplus \badfaults_2)\), where the bundling \(\contract_1\cup \contract_2\) and specification \(\spec_1\cup \spec_2\) are the (overlapping) union of the bundlings and specifications induced by the bijection.
        Additionally, if \(C_1\) is friendly, then \(C_2\circ C_1\) is friendly.
    \end{itemize}
\end{proposition}
\begin{proof}
    For parallel composition, recall that in the definition of a gadget the properties hold with respect to an arbitrary environment circuit. 
    For a Pauli fault \(\fault\) supported on the locations of \(C_1\otimes C_2\), it can be decomposed as\footnote{In particular, there is a fault \(f'\) that does not contract any locations that is equivalent to \(\fault\) in the sense that \(\circuitMap{C_1\otimes C_2[\fault]} = \circuitMap{C_1\otimes C_2[\fault']}\).} \(\fault = \fault_1\otimes \fault_2\) where \(\fault_1\) is supported on the locations of \(C_1\) and \(\fault_2\) is supported on the locations of \(C_2\). Then 
    \begin{align}
        \circuitMap{C_1\otimes C_2[\fault]} &= \circuitMap{C_1[\fault_1]}\otimes \circuitMap{C_2[\fault_2]}
    \end{align}
    and our claim follows by applying \cref{eq:gadget} independently to both circuits, since \(\fault_1\) is \(\badfaults_1\)-avoiding and \(\fault_2\) is \(\badfaults_2\)-avoiding.

    For sequential composition, again consider a Pauli fault \(\fault\) and its decomposition \(\fault = \fault_1\otimes \fault_2\). 
    \begin{align}
        \circuitMap{C_2\circ C_1[\fault]} &= \circuitMap{C_2[\fault_2]}\circ \circuitMap{C_1[\fault_1]}
    \end{align}
    Let \(\Cwithenv\) be circuit with computational circuit \(C_2\circ C_1\) and an arbitrary environment circuit \(\environC\). 
    We can write 
    \begin{align}
        \circuitMap{\Cwithenv} &= \circuitMap{C_2\circ C_1} \otimes \circuitMap{\environC} \\
        &= \circuitMap{C_2\otimes I_{\environ}}\circ \circuitMap{C_1\otimes \environC}. \\
        \circuitMap{\Cwithenv[\fault]} 
        &= \circuitMap{C_2\otimes I_{\environ}[\fault_2]}\circ \circuitMap{C_1\otimes \environC[\fault_1]}.
    \end{align}
    Apply \cref{eq:gadget} to \(C_2\otimes I_{\environ}\) and \(C_1\otimes \environC\), we see that 
    \begin{align}
        \circuitMap{C_2\otimes I_{\environ}[\fault_2]}
        &= (\tilde{\g}_{2,\fault_2}\otimes I_{\environ})\circ (\recover_{\fault_2}\otimes I_{\environ}), \\
        \circuitMap{C_1\otimes \environC[\fault_1]}
        &= (\tilde{\g}_{1,\fault_1}\otimes \circuitMap{\environC})\circ (\recover_{\fault_1}\otimes I_{\environ}).
    \end{align}
    Here \(\recover_{\fault_1}, \recover_{\fault_2}\) will be identities if their respective gadgets are not friendly. Putting things together, we see that 
    \begin{align}
        \circuitMap{\Cwithenv[\fault]} 
        &= (\tilde{\g}_{2,\fault_2}\otimes I_{\environ})\circ (\recover_{\fault_2}\otimes I_{\environ}) \circ (\tilde{\g}_{1,\fault_1}\otimes \circuitMap{\environC})\circ (\recover_{\fault_1}\otimes I_{\environ}) \\
        &= ((\tilde{\g}_{2,\fault_2} \circ \recover_{\fault_2}\circ \tilde{\g}_{1,\fault_1})\otimes \circuitMap{\environC})\circ (\recover_{\fault_1}\otimes I_{\environ}).
    \end{align}
    Note that the superoperator \(\tilde{\g}_{2,\fault_2} \circ \recover_{\fault_2}\circ \tilde{\g}_{1,\fault_1}\) satisfies the simulation property \cref{eq:simulation} for \(\g_2\circ \g_1\) and our conclusion follows.
\end{proof}

\subsection{Decoupling of faults}

When a gadget is faulty, anything can go wrong.
In particular, the logical output state may depend on the input error in a highly non-trivial way.
Informally, we think of noisily encoded states as being reversibly decoded to a state supported on two subsystems: the ``logical'' subsystem and the ``syndrome'' subsystem.
For non-faulty gadgets, the syndrome subsystem remains decoupled from the logical subsystem much like the environment circuit.
However, when a gadget fails, the fault is permitted to access the syndrome subsystem (since, on the encoded level, the resulting state is not independent of the syndrome).\footnote{Recall that faults may act on a contracted circuit.}

This decoupling will be critical to the proof of \cref{lemma:level-reduction}.
\begin{lemma}[Decoupling of Faults]\label{lemma:decoupling-faults}
    We use the variables defined in the definition of gadgets \cref{def:gadget}.
    To emphasize the irrelevance of the environment circuit \(C_\environ\), let \(\Rencoder^{\text{out}}_C\) and \(\Rdecoder^{\text{in}}_C\) refer to the reversible encoder and decoder of the computational circuit \(C\) (which is our fault-tolerant gadget).
    Let \(\syntype^{\text{in}}, \syntype^{\text{out}}\) denote the input and output syndrome type (see~\cref{def:code-type}) of \(C\).
    For any \(\badfaults\)-avoiding Pauli fault \(\fault\) on \(\Cwithenv\) supported only on the computational circuit, there exists a superoperator \(\noise_{\g, \fault}\in \linearOpType[\syntype^{\text{out}}]{\syntype^{\text{in}}}\) such that we can write the following decomposition. 
    \begin{align}
        \circuitMap{\Cwithenv[\fault]} = \left(\Rencoder^{\text{out}}_C \circ \left[\g\otimes \noise_{\g, \fault}\right] \circ \Rdecoder^{\text{in}}_C \circ \recover_{\fault}\right) \otimes \circuitMap{\environC}. \label{eq:decoupling:decomposition}
    \end{align}
    We have drawn the corresponding circuit in \cref{fig:decoupling}. 
    Additionally, there exists a superoperator \(\noise_{\recover, \fault}\in \linearOpType[\syntype^{\text{in}}]{\syntype^{\text{in}}}\) such that for all good input states \(\sigma\), we have
    \begin{align}\label{eq:decoupling-of-filter}
        \recover_{\fault} (\sigma) = \Rencoder^{\text{in}}_C \circ \left(I_{\intype[\bullet]}\otimes \noise_{\recover, \fault} \right) \circ \Rdecoder^{\text{in}}_C (\sigma).
    \end{align}
\end{lemma}
\nomenclature[D, 11]{\(\noise_{\g, \fault}, \noise_{\recover, \fault}\)}{Noise superoperators acting on the syndrome subsystem of an unencoded state, which came from \(\tilde{\g}_{\fault}\) and \(\recover_{\fault}\) of a gadget. See \cref{lemma:decoupling-faults}.}
\begin{proof}
From equation~\eqref{eq:gadget}, we have that \begin{align}
    \circuitMap{\Cwithenv[\fault]} = \left(\tilde{\g}_{\fault}\otimes \circuitMap{\environC}\right)\circ \left(\recover_{\fault}\otimes I_{\environ}\right).
\end{align}
Consider an arbitrary input state \(\sigma_0\in \unphysicalstate{{\intype[\bullet]} \otimes {\intype[\environ]}}\).
By the friendly property of the gadget \(C\), we have that \(\sigma_1 = \recover_{\fault}\otimes I_{\environ}(\sigma_0)\) is a good input state, which means there exists a logical state \(\rho \in \mixedstate{{\intype[L]} \otimes {\intype[\environ]} }\) such that
\begin{align}
    \encoder^{\text{in}}(\rho) \dev{\baderrors^{\text{in}}}
    \sigma_1 
\end{align}
It suffices for us to show that there exists a superoperator \(\noise_{\g, \fault}\) such that for any such \(\sigma_1\), 
\begin{align}
    \left(\tilde{\g}_{\fault}\otimes \circuitMap{\environC}\right)(\sigma_1) = 
    \left(
    \left(\Rencoder^{\text{out}}_C \circ \left[\g\otimes \noise_{\g, \fault}\right] \circ \Rdecoder^{\text{in}}_C\right) \otimes \circuitMap{\environC} \right)
    (\sigma_1).
\end{align}
Since the input and output code types on the environment circuit \(\environC\) is trivial, the encoders and decoders act trivially on the environmental inputs and outputs.
Therefore, the above equation is equivalent to
\begin{align}\label{eq:decoupling-eq}
    \left(
    \left(\Rdecoder_C^{\text{out}} \circ \tilde{\g}_{\fault}\circ \Rencoder_C^{\text{in}}\right)\otimes \circuitMap{\environC}
    \right)
    (\Rdecoder^{\text{in}}(\sigma_1))
    &= 
     \left(
    \left[\g\otimes \noise_{\g, \fault}\right]\otimes \circuitMap{\environC} 
    \right)
    (\Rdecoder^{\text{in}}(\sigma_1)).
\end{align}
From~\cref{cor:decoupling-errors}, we see that \(\Rdecoder^{\text{in}}(\sigma_1) = \rho\otimes \synd\) for some (not necessarily physical) state \(\synd\) satisfying the input syndrome type of \(C\).
Let \(M^{\text{in}}\) be the set of (not necessarily physical) syndrome states \(\synd\in \unphysicalstate{{\syntype^{\text{in}}}}\) such that 
\begin{align}
    \encoder^{\text{in}}(\rho) \dev{\baderrors^{\text{in}}} \Rencoder^{\text{in}}(\rho\otimes \synd).
\end{align}
Similarly define \(M^{\text{out}}\).
Note that \(M^{\text{in}}, M^{\text{out}}\) may not be linear subspaces of \(\unphysicalstate{{\syntype^{\text{in}}}}\) and \(\unphysicalstate{{\syntype^{\text{out}}}}\), respectively.
However, the compositions of maps in \cref{eq:decoupling-eq} are linear. 
Therefore, let \(M\) be a basis of \(\mathsf{span}(M^{\text{in}})\).
From the simulation property of the gadget \(C\), we see that for any \(\synd\in M\), there exists \(\synd_{\fault}\in M^{\text{out}}\) such that for all states \(\rho\) of the input logical type of \(C\), we have
\begin{align}
    \left(
    \left(\Rdecoder_C^{\text{out}} \circ \tilde{\g}_{\fault}\circ \Rencoder_C^{\text{in}}\right)\otimes \circuitMap{\environC}
    \right)
    (\rho\otimes \synd)
    &= 
    \left(\g\otimes \circuitMap{\environC} \right) (\rho)\otimes \synd_{\fault}.
\end{align}
Define \(\noise_{\g, \fault}: \unphysicalstate{{\syntype^{\text{in}}}}\rightarrow \unphysicalstate{{\syntype^{\text{out}}}}\) by \(\noise_{\g, \fault}(\synd) = \synd_{\fault}\).
Then \(\noise_{\g, \fault}\) is a valid linear map defined on \(\mathsf{span}(M^{\text{in}})\). 
Extend \(\noise_{\g, \fault}\) linearly and arbitrarily onto the rest of the domain \(\unphysicalstate{{\syntype^{\text{in}}}}\).
We have that for all good input states \(\sigma_1\),~\cref{eq:decoupling-eq} holds as desired. 
Applying the same arguments to \(\recover_{\fault}\), we can construct \(\noise_{\recover, \fault}\) for \cref{eq:decoupling-of-filter}.
\end{proof}

Later on, in the proof of level reduction, the syndrome subsystem will be moved to be part of the environment circuit.

\begin{figure}
    \centering
    \includegraphics[width=0.9\linewidth]{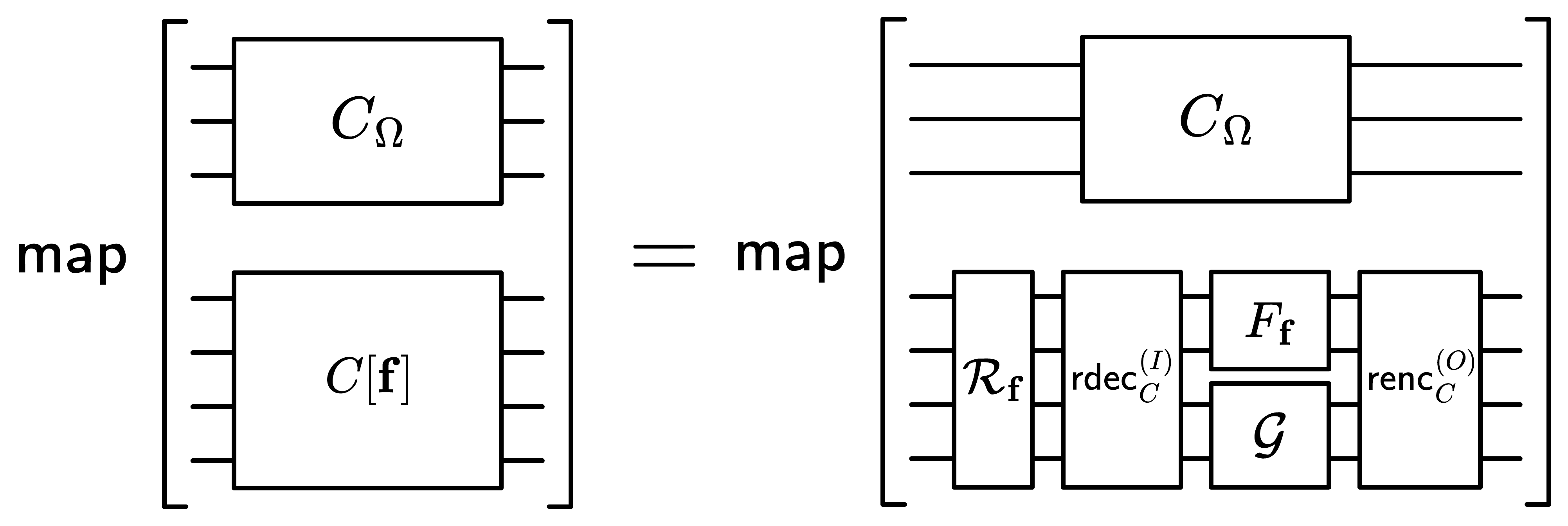}
    \caption{Pictoral representation of the decoupling lemma (\cref{lemma:decoupling-faults}).
    Whenever \(C\) accepts or outputs multiple wire bundles, \(\noise_{\g, \fault}\) acts on all syndrome subsystems and \(\g\) acts on all logical subsystems. }
    \label{fig:decoupling}
\end{figure}

\subsection{Level reduction}
We are now ready to prove an analog of ``level reduction'' from \cite{aliferis2005quantum} which represents the most challenging form of composition.
Roughly, this allows reasoning about the simulated circuit whenever gadgets in the fault-tolerant circuit fail.
The proof will crucially use the friendliness of the gadgets and the environment circuit.

\begin{theorem}[Level reduction]\label{lemma:level-reduction}
    Fix a circuit \(C = (V,E,\edgeConn, \edgeType, \nodeGate)\) (potentially with environment) and a contraction map \(\kappa \colon V \to V_\kappa\) such that every location in \(W = V_\kappa\) of the contracted circuit \(C_\kappa = (W,E_\kappa,\edgeConn_\kappa, \edgeType_\kappa, \nodeGate_\kappa)\) corresponds to a gadget \(w \in W\).
    Assign a bundling \(\bundle \colon E \to \bundleE\) to the edges of \(C_\kappa\) compatible with the bundling of each gadget and a code specification \(\spec\) such that each gadget \(w \in W\) simulates the gate \(\g_w\) with input and output code types\footnote{Actually, they need only be compatible with the code type in the sense that the code type assigned to the wire bundle may be a stronger set of bad error supports.} given by the global spec \(\spec\) and bad Pauli faults \(\mathcal{F}_w \subseteq P(\kappa^{-1}(w))\).
    We further require that: 1) The input and output types of the circuit (with respect to \(\spec\)) are trivial (see \cref{def:code-type}).
    2) All gadgets \(w\in W\) connected to the output have no bad fault paths \(\mathcal{F}_w = \{\}\) e.g. are classical decoders (classical locations do not have faults).\footnote{This is possible for circuits with classical input and output in the model of computation where classical circuits do not have faults. More generally, the corresponding statement requires encoders on the input and decoders on the output.}
        
    Consider also the logical circuit \(C_L = (W,E_L,\edgeConn_L, \edgeType_L, \nodeGate_L)\) where edge bundle \(e_\beta \in \bundleE\) of \(C_\kappa\) is replaced by an edge of the corresponding logical type (given by \(\spec\)) and, for every gadget \(w \in W\) of \(C_\kappa\), the gate \(\nodeGate_\kappa(w)\) is replaced by the gate that the gadget simulates \(\nodeGate_L(w) := \g_w\).
    
    Now fix a Pauli fault \(\fault\) supported on \(C\) and a family of bad fault paths \(\mathcal{W} \subseteq P(W)\).
    If \(\fault\) is \(\mathcal{W} \bullet \{\mathcal{F}_w\}_{w \in W}\)-avoiding, then, for some \(\mathcal{W}\)-avoiding fault\footnote{\(\fault'\) is not Pauli in general.} \(\fault'\) supported on \(C_{L,\Omega}\) (\(C_{L}\) \emph{with environment}) we have that
    \begin{align}
        \circuitMap{C[\fault]} = \circuitMap{C_{L,\Omega}[\fault']}
    \end{align}

\end{theorem}

\begin{proof}
    Let \(U \subseteq W\) be the set of ``good'' gadgets where, for each gadget \(w \in U\), the fault \(\fault\) is \(\mathcal{F}_w\)-avoiding on the locations of the gadget \(\kappa^{-1}(w)\).
    We will construct \(C_{L,\environ}[\fault']\) by a sequence of circuit rewrites of \(C_\kappa[\fault]\).
    
    The first step will be to rewrite the gates of \(C_\kappa[\fault]\) into a new gate map \(\nodeGate^{(1)}_\kappa\).
    First note that, because \(\fault\) is Pauli, we can take the vertex contraction to be trivial.
    Thus, in what follows, \(C\) may have an environment circuit, but we do not need to explicitly work with it.
    Let \(\nodeFaulty_\kappa(w)\) be the gate map of \(C_\kappa[\fault]\).
    
    For each gadget \(w \in W\), if the gadget is good \(w \in U\), then the support of the fault \(\fault\) is \(\mathcal{F}_w\)-avoiding and can we apply the decomposition (\cref{eq:decoupling:decomposition} and \cref{eq:decoupling:factorization}) of the decoupling lemma \cref{lemma:decoupling-faults}.
    For input edge bundles \((b_1, \dots, b_{\ell_w})\) of \(w\), the faulty gadget map has the decomposition
    \begin{align}\label{eq:level-reduction:faulty-gate-map-decomp}
        \nodeFaulty_\kappa(w) = \Rencoder^{\text{out}}_w \circ \left[\g_w\otimes \noise_{\g,\fault,w}\right] \circ \Rdecoder^{\text{in}}_w \circ \left(\recover_1^{\text{in}}\otimes \dots \otimes \recover_{\ell_w}^{\text{in}}\right).
    \end{align}
    We will ``push'' the superoperators \(\recover_i\) to the preceding gadget.
    More precisely, consider a gadget \(w \in W\) with \(m_w\) output bundles \((b_1, \dots b_{m_w})\).
    For \(i \in [{m_w}]\), let \(u_i \in W\) be the gadget associated with \(\edgeDest(b_i)\).
    If \(u_i\) is good \(u_i \in U\), then define \(\mathcal{T}_{w,i}^{\text{out}}\) to be the corresponding input bundle filter in the decomposition \cref{eq:level-reduction:faulty-gate-map-decomp} at \(u_i\).
    Otherwise, we take \(\mathcal{T}_{w,i}^{\text{out}}\) to be the identity.
    We can now write the modified gate map \(\nodeGate_\kappa^{(1)}\).
    \begin{align}
        \nodeGate_\kappa^{(1)}(w) = \begin{cases}
            \left(\mathcal{T}_{w,1}^{\text{out}} \otimes \dots \otimes \mathcal{T}_{w,{m_w}}^{\text{out}}\right) \circ \nodeFaulty_\kappa(w) & w \notin U \\
            \left(\mathcal{T}_{w,1}^{\text{out}} \otimes \dots \otimes \mathcal{T}_{w,{m_w}}^{\text{out}}\right) \circ \Rencoder^{\text{out}}_w \circ \left[\g_w\otimes \noise_{\g,\fault,w}\right] \circ \Rdecoder^{\text{in}}_w  & w \in U
        \end{cases}
    \end{align}
    By construction, the contracted circuit \(C_\kappa^{(1)} = (W,E_\kappa,\edgeConn_\kappa, \edgeType_\kappa, \nodeGate_\kappa^{(1)})\) with the modified gate map is equivalent to \(C_\kappa[\fault]\):
    \begin{align}
        \circuitMap{C_\kappa^{(1)}} = \circuitMap{C_\kappa[\fault]}.
    \end{align}

    In the next step, we will strip off the filter superoperators where we are able to.
    For each gadget \(w \in W\), let \(\baderrors^{\text{in}}_w\) and \(\baderrors^{\text{out}}_w\) be the input and output bad error supports of the gadget \(w\) (induced by \(\spec\)).
    Consider the action of \(\circuitMap{C_\kappa}\) on an arbitrary input state \(\rho_0\).
    Let \(\rho_t\), \(t \in [|W|]\) be the state after applying the gate maps of gadgets in \((1, \dots, t)\) (in topological order) to \(\rho_0\).

    For each \(t \in [|W|]\), the direct ancestors \(u\) of \(w_t\) have the filter superoperators \(\left(\mathcal{T}_{u,1}^{\text{out}} \otimes \dots \otimes \mathcal{T}_{u,{m_u}}^{\text{out}}\right)\), so there is a logical state \(\sigma_{t-1}\) such that \(\rho_{t-1}\) is \(\baderrors^{\text{in}}_{w_t}\)-deviated from \(\encoder[w_t]^{\text{in}} (\sigma_{t-1})\)\footnote{Here the input encoder of \(w_t\) is defined by \(\spec\) extended by identity for the subsystems not in the support of \(\nodeGate_\kappa(w_t)\)}.
    Thus, if \(w_t \in U\) is good, using the simulation definition (\cref{def:gadget}), it follows that \(\rho_t\) is \(\baderrors^{\text{out}}_{w_t}\)-deviated from \(\encoder[w_t]^{\text{out}} \circ \left(\g_{w_t} \otimes \noise_{\g,\fault,w_t}\right)(\sigma_{t-1})\).
    From \cref{eq:decoupling-of-filter} of \cref{lemma:decoupling-faults}, we see that on such a good logical state, the filter can be decomposed as 
    \begin{align}
        \forall i\in [m_{w_t}], \mathcal{T}_{w_t,i}^{\text{out}} &= \Rencoder^{\text{out}}_{w_t} \circ \left[I\otimes \noise_{\mathcal{T},\fault,w_t,i}\right] \circ \Rdecoder^{\text{out}}_{w_t}.
    \end{align}
    In other words, the filter superoperators at \(w_t\) (if present) act as identity on the logical state.
    We can therefore absorb the filter superoperators into the fault superoperator supported on the syndrome system at each good gadget. 
    Denote
    \begin{align}
        \noise_{\mathcal{T},\fault,w_t}
        &:= \bigotimes_{i=1}^{m_{w_t}} \noise_{\mathcal{T},\fault,w_t,i} \qquad 
        \noise_{\fault, w_t}:= \noise_{\mathcal{T},\fault,w_t}\circ \noise_{\g,\fault,w_t} 
    \end{align}
    we can define another modified gate map
    \begin{align}
        \nodeGate_\kappa^{(2)}(w) = \begin{cases}
            \left(\mathcal{T}_{w,1}^{\text{out}} \otimes \dots \otimes \mathcal{T}_{w,{m_w}}^{\text{out}}\right) \circ \nodeFaulty_\kappa(w) & w \notin U \\
            \Rencoder^{\text{out}}_w \circ \left[\g_w\otimes \noise_{\fault, w}\right] \circ \Rdecoder^{\text{in}}_w  & w \in U
        \end{cases}
    \end{align}
    By induction on the circuit locations in reverse topological order, and noting that the input code type is trivial so the base case is satisfied, we see that this rewrite of gate map results in an equivalent circuit \(C_\kappa^{(2)} = (W,E_\kappa,\edgeConn_\kappa, \edgeType_\kappa, \nodeGate_\kappa^{(2)})\).
    \begin{align}
        \circuitMap{C_\kappa^{(2)}} = \circuitMap{C_\kappa^{(1)}}
    \end{align}
    We will now unencode every edge bundle of \(C_\kappa^{(2)}\) by inserting a reversible encoder/decoder pair for each edge bundle \(b\).
    This replaces the edges of the edge bundle by two edges, one with the logical type of \(b\) and one with the syndrome type of \(b\).
    The reversible encoders and decoders cancel the reversible decoders and encoders of the good gadgets.
    \begin{align}
        \nodeGate_\kappa^{(3)}(w) = \begin{cases}
            \Rdecoder^{\text{out}}_w \circ \left(\mathcal{T}_{w,1}^{\text{out}} \otimes \dots \otimes \mathcal{T}_{w,{m_w}}^{\text{out}}\right) \circ \nodeFaulty_\kappa(w) \circ \Rencoder^{\text{in}}_w & w \notin U \\
            \g_w\otimes \noise_{\fault, w}  & w \in U
        \end{cases}
    \end{align}
    We have only inserted pairs of unitaries and their inverses.
    \begin{align}
        \circuitMap{C_\kappa^{(3)}} = \circuitMap{C_\kappa^{(2)}}
    \end{align}
    Finally, we can construct \(C_{L,\error}[\fault']\) by ``splitting'' all good gadgets of \(C_\kappa^{(3)}\):
    Every good gadget \(w\) has a gate \(\nodeGate_\kappa^{(3)}(w)=\g_w\otimes \noise_{\fault, w}\) that factorizes into a superoperator supported only on the logical subsystem \(\g_w\) and a superoperator supported only on the syndrome subsystem \(\noise_{\fault, w}\).
    
    We now equip \(C_L\) with an environment circuit \(C_{\environ}\) with network matching that of \(C_L\).
    That is, for every vertex \(w\) of \(C_L\), we add a vertex \(w^{(\mathrm{s})}\) to the environment circuit.
    The edge types are the syndrome types of each bundle of \(C_\kappa\) (which are in bijection with edges of \(C_L\)) and the gate map is arbitrary.
    The fault \(\fault'\) contracts \(w^{(\mathrm{s})}\) with \(w\) whenever \(w \notin U\) is not good.
    More precisely, let \(\nodeFaulty_L\) be the gate map of \(C_{L,\error}[\fault']\).
    \begin{align}
        w \in U & &\nodeFaulty_L(w) =& \g_w \\
                & &\nodeFaulty_L(w^{\mathrm{(s)}}) =& \noise_{\fault, w} \\
        w \notin U & &  \nodeFaulty_L(\{w,w^{\mathrm{(s)}}\}) =& \Rdecoder^{\text{out}}_w \circ \left(\mathcal{T}_{w,1}^{\text{out}} \otimes \dots \otimes \mathcal{T}_{w,{m_w}}^{\text{out}}\right) \circ \nodeFaulty_\kappa(w) \circ \Rencoder^{\text{in}}_w
    \end{align}
    Where we use the abuse of notation \(\{w,w^{\mathrm{(s)}}\}\) to refer to the vertex under the vertex contraction of the fault.
    This is a rearrangement of \(C_\kappa^{(3)}\), so
    \begin{align}
        \circuitMap{C_{L,\error}[\fault']} &= \circuitMap{C_\kappa^{(3)}} = \circuitMap{C_\kappa^{(2)}} = \circuitMap{C_\kappa^{(1)}} = \circuitMap{C_\kappa[\fault]}~. 
        \qedhere
    \end{align}
\end{proof}

\begin{figure}
    \centering
    \includegraphics[width=0.7\linewidth]{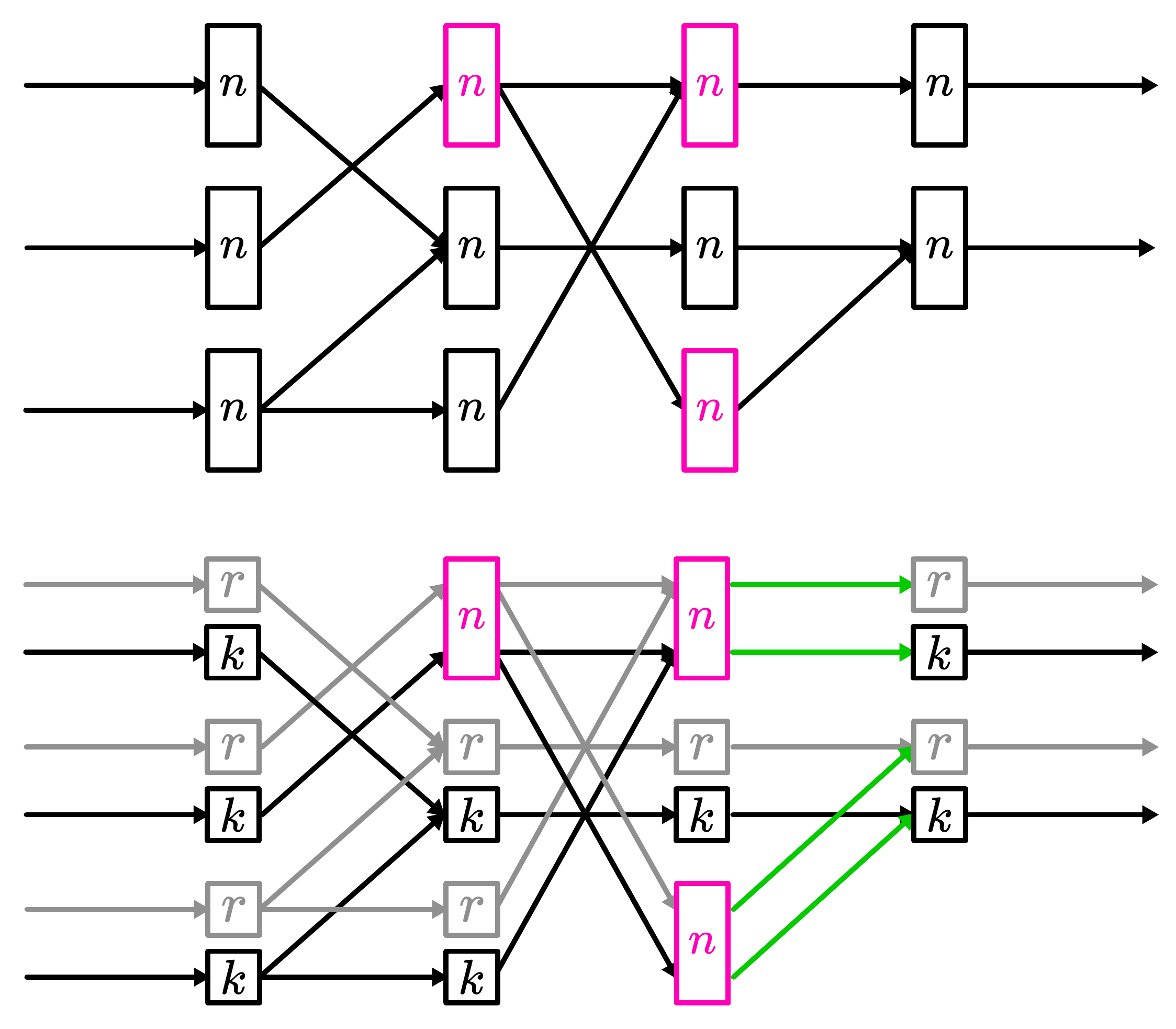}
    \caption{Pictoral transformation from \(C_\kappa[\fault]\) to \(C_{L,\environ}[\fault']\). The set of faulty gadgets \(U\) is highlighted magenta. Every vertex (not in \(U\)) in the transformed circuit (lower) is paired with an environment vertex. For clarity, connections between environment and computational vertex pairs are drawn as doubled arrows. Green highlighted connections correspond to ``pushing'' a filter map \(\recover_i\) to part of the fault on the previous vertex (see item 1 of case 1).}
    \label{fig:level-reduction}
\end{figure}

\subsection{Miscellaneous tools}
We now include two miscellaneous tools that are generally useful.
The first is the decomposition of faults into Pauli faults.
\begin{lemma}[Decomposition of faults]\label{lemma:fault-decomposition}
    For any circuit \(C\) and fault \(\fault\), there exists a new circuit with environment \(\Cwithenv\) and set of Pauli faults \(\{\fault'_i\}_i\) such that:
    \begin{itemize}[topsep = 0pt]
        \item \(\Cwithenv\) has computational circuit \(C\) and a closed environment circuit \(\environC\) (see \cref{def:circuit-environment}).
        \item The map of \(C[\fault]\) is a linear combination of maps of \(\Cwithenv\) subject to the faults \(\fault_i'\).
        \begin{align}\label{eq:decomposition-faults}
            \circuitMap{C[\fault]} = \sum_i \circuitMap{\Cwithenv[\fault'_i]} 
        \end{align}
        Note that since \(\environC\) is closed, every \(\circuitMap{\Cwithenv[\fault'_i]}\) is a superoperator with the same input and output type as \(C\) and \(C[\fault]\). 
        \item Furthermore, for each \(i\), the support of each fault \(\fault'_i\) on the computational circuit is a subset of the support of the original fault \(\fault\).
    \end{itemize}
\end{lemma}
\begin{proof}
    Note that general fault could replace gates in a circuit, while Pauli faults (\cref{def:pauli-superoperator}) are inserted before and after the intended gates. Therefore, the main argument for this proposition is to construct an environment circuit to help rewrite general faults on the computational circuit into Pauli faults. 

    Denote the circuit as \(C = (V,E,\edgeConn, \edgeType, \nodeGate)\), and let the fault be \(\fault = (\contract, \nodeGate_\contract)\), where \(\contract:V\rightarrow \contractV\) is a contraction map and \(\nodeGate_\contract(v)\) for \(v\in \contractV\) are faulty gates.
    We construct \(\Cwithenv\) as follows. 
    Create two independent copies of \(C\), namely \(C_1, C_2\) with corresponding data. Apply the fault \(\fault\) to \(C_2\), by taking \(\contract_2:V_2\rightarrow V_{2,\contract}\) and \(\nodeGate_{2,\contract}\) to be a labeling on \(V_{2,\contract}\).
    Then \(\Cwithenv\) is our circuit with environment, where \(C_1\) is the computational circuit and \(C_{2,\contract} = (V_{2,\contract}, E_{2,\contract}, \edgeConn_{2,\contract}, \edgeType_{2,\contract}, \nodeGate_{2,\contract})\) is the environment circuit.

    To satisfy~\cref{eq:decomposition-faults}, we add, as a new fault \(\fault'\), a collection of \(\SWAP\) gates between \(C_1\) and \(C_{2,\contract}\), and later decompose these \(\SWAP\) gates into Pauli faults.  
    We define the new fault \(\fault' = (\tau, \nodeGate_\tau)\) as follows.
    For a vertex \(v\in V\), denote its copies in \(V_1, V_2\) as \(v[1], v[2]\), respectively.
    Precisely, for every \(v\in V\) such that \(\nodeGate_{2,\contract}(\contract(v[2]))\ne \nodeGate_2(v[2])\), the set of vertices contracted to \(\contract(v[2])\) is \(\contract^{-1}(\contract(v[2]))\), and the corresponding vertices in \(V_1\) is 
    \begin{align}
        T_v := \{u[1] : u[2]\in \contract^{-1}(\contract(v[2]))\}.
    \end{align}
    Define 
    \begin{align}
        \nodeIn(T_v) = \bigcup_{u[1]\in T_v} \nodeIn(u[1]),
    \end{align}
    and similarly define \(\nodeOut(T_v)\).
    There is a natural bijection \(\phi_{\text{in}}:  \nodeIn(T_v)\rightarrow \nodeIn(\contract(v[2]))\), induced by the bijection between \(T_v\) and the vertices contracted to \(\contract(v[2])\). 
    Similarly there is a natural bijection
    \(\phi_{\text{out}}:  \nodeOut(T_v)\rightarrow \nodeOut(\contract(v[2]))\). 
    We define \(\tau\) to contract \(\contract(v[2])\) and \(T_v\), and define the gate on this contracted vertex to be
    \begin{align}
        &\nodeGate_{\tau}(\tau\left(
        T_v \cup \{\contract(v[2])\}
        \right)) =\\
        &\left(\bigotimes_{e\in \nodeOut(T_v)} \SWAP[e, \phi_{\text{out}}(e)]\right)
        \circ 
        \left(\nodeGate_{2,\contract}(\contract(v[2])) \otimes \bigotimes_{u[1]\in T_v} \nodeGate_1(u[1]) \right)
        \circ 
        \left(\bigotimes_{e\in \nodeIn(T_v)} \SWAP[e, \phi_{\text{in}}(e)]\right)
    \end{align}
    This gate is conceptually simple: we \(\SWAP\) all input edges of vertices in \(T_v\) with input edges of \(\contract(v[2])\), perform the same gates as in \(\Cwithenv\) (which performs a faulty gate on \(\contract(v[2])\) as specified by \(\fault\), and performs the ideal circuit gates as in \(C\)), and then \(\SWAP\) the output edges again. 
    This completes our description of \(\fault' = (\tau, \nodeGate_\tau)\).

    Now observe that the computational circuit of \(\Cwithenv\) is a copy of \(C\), while the environment circuit \(\environC\) is a copy of \(C[\fault]\).
    We finish the construction of \(\Cwithenv\) by initializing\footnote{This choice of initialization is completely arbitrary.} all input edges of \(\environC\) to \(\ket{0}\) and tracing out all output edges of \(\environC\), thereby closing \(\environC\).
    By adding the new fault \(\fault'\), we have \(\SWAP\)ed the faulty gates from the environment circuit into our computational circuit. 
    This gives us 
    \begin{align}
        \circuitMap{C[\fault]} = \circuitMap{\Cwithenv[\fault']}. 
    \end{align}
    We observe that the fault path of \(\fault'\) on the computational circuit of \(\Cwithenv\) is in bijection with the fault path of \(\fault\) on \(C\).
    Moreover, the fault \(\fault'\) conjugates the existing gates in \(\Cwithenv\) by \(\SWAP\) gates, which can be decomposed into non-diagonal Pauli faults. 
    This proves our claim.
\end{proof}

We will also use the following lemma, which is a statement of the standard fact that stabilizer measurements decoheres general error into Pauli errors. 
A proof was given in~\cite{nguyen2024quantum}, we include it here for completeness. 
\begin{definition}[Recoverable Set]
    For a stabilizer code \(\qcode\) on \(n\) physical qubits, a subset of qubits \(S\subseteq [n]\) is \textit{recoverable} if there are no logical operator of \(\qcode\) which is fully supported in \(S\). 
\end{definition}
\begin{lemma}[Decoherence of errors~\cite{nguyen2024quantum}]\label{lemma:deterministic-errors}
  For a \(n\)-qubit stabilizer code \(\qcode\), let \(\mathcal{M}\) be the channel that measures the \(r\) stabilizer checks of \(\qcode\) and outputs the measurement outcomes.
  For a superoperator \(\error_B\) supported on a recoverable subset of qubits
  \(B \subseteq [n]\) and a codestate \(\rho\) of \(\qcode\), the application of the noise superoperator followed by measurement of the checks can be written as
  \begin{align}
    \mathcal{M}\circ \error_B(\rho) = \sum_{s \in \F^r} \alpha_s \ketbra{s}{s}\otimes E_s \rho E_s,
  \end{align}
  where each \(E_s\) is a Pauli operator supported on \(B\) and the \(\alpha_s\) are complex coefficients.
  When \(\error_B\) is a physical noise channel, they satisfy \(\sum_s \alpha_s = 1\).
  In other words, measurement of the checks collapses the error into a single Pauli error \emph{as long as we remember the measurement outcome}.
\end{lemma}
\begin{proof}
Let \(\{S_i\}_{i\in [r]}\) be a basis of the stabilizer group of \(\qcode\). For a syndrome \(s\in \F^r\), we define the projector to the syndrome space corresponding to \(s\) as \(\Pi_s = \prod_r(\frac{1}{2}I + (-1)^{s_i}S_i)\). Then the measurement channel \(\mathcal{M}\) can be expressed as 
\begin{align}
    \mathcal{M}(\rho) &= \sum_{s\in \F^r}\ketbra{s}\otimes \Pi_s\rho\Pi_s.
\end{align}
Consider the Kraus decomposition of \(\error_B\) in terms of Pauli operators \(\{K_\mu\}_\mu, 
\{K_{\nu}'\}_\nu\) supported on \(B\) and complex coefficients \(\{\alpha_{\mu,\nu}\}_{\mu,\nu}\). 
\begin{align}
    \mathcal{M}\circ \error_B(\rho)
    &= \sum_{s\in \F^r}\ketbra{s}\otimes \Pi_s \left( 
        \sum_{\mu,\nu}\alpha_{\mu,\nu} K_{\mu} \rho K_{\nu}'
     \right)\Pi_s \\
    &= \sum_{s\in \F^r}\ketbra{s}\otimes \left( 
        \sum_{\mu,\nu}\alpha_{\mu,\nu} \Pi_s K_{\mu} \rho K_{\nu}'\Pi_s
        \right)\\
    &= \sum_{s\in \F^r}\ketbra{s}\otimes \left( 
        \sum_{\mu,\nu}\alpha_{\mu,\nu} \Pi_s K_{\mu}\Pi_0 \rho \Pi_0 K_{\nu}'\Pi_s \right)
\end{align}
We see that the projectors \(\Pi_s\) annihilate all terms except those where \(K_{\mu}, K_{\nu}'\) both have syndrome \(s\). Moreover, sicne \(B\) is a recoverable set, any two Pauli oeprators supported on \(B\) which has the same syndrome are equivalent up to stabilizers.
Therefore, for every \(s\in \F^r\), there exists Pauli operator \(E_s\) supported on \(B\) which is equivalent up to stabilizers to all \(K_\mu, K_{\nu}'\) with syndrome \(s\). 
In other words, \(\Pi_s E_s\Pi_0 = \Pi_s K_\mu \Pi_0\) for such \(\mu,\nu\). 
Let \(\alpha_s\) be the sum of all \(\alpha_{\mu,\nu}\) where \(K_\mu, K_{\nu}'\) have syndrome \(s\). 
We can write 
\begin{align}
    \mathcal{M}\circ \error_B(\rho)
    &= 
    \sum_{s\in \F^r}\ketbra{s}\otimes \left( 
    \sum_{\mu,\nu}\alpha_{\mu,\nu} \Pi_s K_{\mu}\Pi_0 \rho \Pi_0 K_{\nu}'\Pi_s \right)\\
    &= \sum_{s\in \F^r}\alpha_{s} \ketbra{s}\otimes 
    \Pi_s E_s\Pi_0 \rho \Pi_0 E_s\Pi_s \\
    &= \sum_{s\in \F^r}\alpha_{s} \ketbra{s}\otimes 
    E_s \rho E_s. \qedhere
\end{align}
\end{proof}

\section{Circuit correctness}\label{sec:circuit-correctness}
Here we prove that under various noise models, the total variation distance between the output distribution of the simulated circuit and that of the fault-tolerant circuit can be bounded in terms of the evaluation of the weight enumerator for various noise models.

In what follows, let \(C\) be a classical-quantum circuit with classical input and classical output.
Note that \(\circuitMap{C}\) is a map from bitstrings to distributions.

Let \(\cft =  (V,E,\edgeConn, \edgeType, \nodeGate)\) be a gadget for \(C\) with trivial input and output types and bad fault paths \(\mathcal{F}\) (e.g. obtained from \cref{lemma:level-reduction} with \(\mathcal{W} = \{\{v\} \mid v \in V\}\)).
In what follows, in order to obtain standard threshold results, think of \(\cft\) as being an element of a family indexed by some element \(n\) such that below some critical value \(x \in [0,\epsilon_*)\), the evaluation of the weight enumerator can be upper bounded e.g. \(\weightenum{\mathcal{F}}{x} \le e^{-f(n)}\) for some positive function \(f\).
Thus, a ``universal'' threshold theorem simply provides such a family of mappings from circuits to gadgets with upper bounds on the weight enumerators.

\subsection{Adversarial noise}
Our first noise model is simply an adversarial one.
In this model, an arbitrarily powerful adversary is permitted an arbitrarily large quantum memory and access to qubits supported on a set of locations of the circuit \(S\subseteq V\) that is \(\mathcal{F}\)-avoiding.
The adversary is permitted to replace each fault location with an arbitrary channel of the same type.
This is far stronger than the capability to insert Pauli operators: They may, for example, ``steal'' qubits and insert them into the circuit at a later time.

Let \(C_{\mathsf{FT},\environ}[\fault]\) be the execution of the circuit \(\cft\) in the presence of the adversary.
Their interaction is captured by the introduction of an environment circuit and an \(\mathcal{F}\)-avoiding \(\fault\).

\begin{proposition}\label{prop:adversarial-noise}
    Let \(C_{\mathsf{FT},\environ}[\fault]\) be the execution of \(\cft\) in the presence of an arbitrary, closed environment circuit \(\environC\) and a \(\mathcal{F}\)-avoiding physical fault \(\fault\), such that $\circuitMap{C_{\mathsf{FT},\environ}[\fault]}$ is a physical CPTP channel. 
    It holds that
    \begin{align}
        \circuitMap{C_{\mathsf{FT},\environ}[\fault]} = \circuitMap{C}.
    \end{align}
\end{proposition}
\begin{proof}
    We first apply \cref{lemma:fault-decomposition} to write the map of the original circuit in terms of a sum over executions of circuit \(\cft\) with a modified environment subject to Pauli faults \(\{\fault_i'\}_i\).
    \begin{align}
        \circuitMap{C_{\mathsf{FT},\environ}[\fault]} = \sum_i\circuitMap{C_{\mathsf{FT},\environ'}[\fault_i']}.
    \end{align}
    The Pauli faults are also \(\mathcal{F}\)-avoiding, so we can apply the definition of a fault-tolerant gadget (simulation property,~\cref{def:gadget}).
    For some complex coefficients \(\{c_i\}_i\) we can write the following.
    \begin{align}
        \sum_i\circuitMap{C_{\mathsf{FT},\environ'}[\fault_i']} = \sum_i c_i \circuitMap{C} = \circuitMap{C}.
    \end{align}
    Where the last equality holds because \(\circuitMap{C_{\mathsf{FT},\environ}[\fault]}\) is a physical (CPTP) channel (which is the case, for example, when each (faulty) gate in \( C_{\mathsf{FT},\environ}[\fault] \) is physical).
\end{proof}

\subsection{Local stochastic noise}
Our next noise model is related to the adversarial stochastic noise mode of \cite{aliferis2005quantum}.
It is the adversarial noise model of the previous section, but the support of the adversary is permitted to be a random variable that obeys a ``locally stochastic'' property. 

\begin{proposition}\label{prop:local-stochastic-noise}
    Let \(C_{\mathsf{FT},\environ}[\fault]\) be the execution of \(\cft\) in the presence of an arbitrary, closed environment circuit and a physical fault \(\fault\) distributed according to some distribution \(M\) such that, for all subsets \(S\subseteq V\) of locations of the computational circuit, \(\Pr_M(S \subseteq \supp \fault) \le \epsilon^{|S|}\). 
    Then, for any input bitstring \(x\), the total variation distance of the output distributions (equivalently, trace distance of the corresponding density operators) is at most \(\epsilon_L = \weightenum{\mathcal{F}}{\epsilon}\).
    That is, denoting the probability density of \(\fault\) by \(P_M\),
    \begin{align}
        \frac{1}{2}\tr \left| \circuitMap{C}\left(\ketbra{x}{x}\right) - \left(\sum_{\fault} P_M(\fault) \circuitMap{C_{\mathsf{FT},\environ}[\fault]}(\ketbra{x}{x})\right)\right| \le \weightenum{\mathcal{F}}{\epsilon}.
    \end{align}
\end{proposition}
\begin{proof}
    We first compute the probability that the support of \(\fault\) is not \(\mathcal{F}\)-avoiding.
    \begin{align}
        \Pr_M(\supp \fault \text{ is not \(\mathcal{F}\)-avoiding}) \le \sum_{S \in \mathcal{F}} \Pr_M(S \subseteq \supp \fault) \le \sum_{S \in \mathcal{F}} \epsilon^{|S|} = \weightenum{\mathcal{F}}{\epsilon}.
    \end{align}
    Let \(\bar{E}\) be this subset of events.
    Whenever the fault is not in the set, the two density matrices are equal by \cref{prop:adversarial-noise}.
    Otherwise, the difference may be at most \(2\) by normalization.
    \begin{align}
        &\frac{1}{2}\tr \left| \circuitMap{C}\left(\ketbra{x}{x}\right) - \left(\sum_{\fault} P_M(\fault) \circuitMap{C_{\mathsf{FT},\environ}[\fault]}(\ketbra{x}{x})\right)\right|\\
        &\le \frac{1}{2}\sum_{\fault} P_M(\fault) \tr \left| \circuitMap{C}\left(\ketbra{x}{x}\right) - \circuitMap{C_{\mathsf{FT},\environ}[\fault]}(\ketbra{x}{x})\right|\\
        &= \frac{1}{2}\sum_{\fault \in \bar{E}} P_M(\fault) \tr \left| \circuitMap{C}\left(\ketbra{x}{x}\right) - \circuitMap{C_{\mathsf{FT},\environ}[\fault]}(\ketbra{x}{x})\right| \\
        &\le \sum_{\fault \in \bar{E}} P_M(\fault) \le \weightenum{\mathcal{F}}{\epsilon}.\qedhere
    \end{align}
\end{proof}

In the above proposition, the statement holds for any input string $x$ because we have assumed classical gates in $\cft$ to be noiseless. This assumption can be dropped if we incorporate classical fault-tolerance constructions. 

\subsection{Coherent noise}

We can also bound against coherent noise in much the same way using a standard technique \cite{aharonov1997fault,aliferis2005quantum}. However, a more careful bound on the contribution from locations outside of the fault support is required.
\begin{proposition}\label{prop:coherent-noise}
    Let \(\cft[\fault]\) be the execution of \(\cft\) subject to a fault \(\fault\) that does not contract any vertices and applies the correct gate of \(\cft\) followed by an operation with diamond distance \( \leq \epsilon\) from the identity.
    Additionally, suppose that 
    \begin{enumerate}[topsep = 0pt]
        \item \(\cft\) is a composition of gadgets,
        \begin{align}
            \circuitMap{\cft} &= \circuitMap{C_m}\circ \circuitMap{C_{m-1}} \circ \dots \circ \circuitMap{C_1},
        \end{align}
        and each gadget \(C_i\) has bad fault paths \(\badfaults_i\) such that \(\badfaults = \boxplus_{i=1}^m \badfaults_i\);
        \item Let \(V_i\) denote the set of locations of \(C_i\).
        There exists \(\ratio\) such that for each \(i\in [m]\), for every \(F_i\in \badfaults_i\), it holds that \(|V_i|\le \ratio |F_i|\).
    \end{enumerate}    
    Set \(\eta = \max_i\weightenum{\mathcal{F}_i}{2(1+\epsilon)^{\ratio-1}\epsilon}\).
    Then, for any input bitstring \(x\), the trace distance of the corresponding density operators is
    \begin{align}
        \frac{1}{2}\tr \left| \circuitMap{C}\left(\ketbra{x}{x}\right) - \circuitMap{\cft[\fault]}(\ketbra{x}{x})\right| \le m(1+\eta)^m\eta.
    \end{align}
\end{proposition}
\begin{proof}

    Since the faults do not contract any locations and act on each gate independently, we can decompose the fault \(\fault\) into \(\fault_1, \dots \fault_m\) supported on the locations in \(C_1, \dots, C_m\).
    \begin{align}
        \circuitMap{\cft[\fault]} &= \circuitMap{C_m[\fault_m]} \circ \dots \circ \circuitMap{C_1[\fault_1]}
    \end{align}
    We first analyze each individual gadget \(C_i[\fault_i]\).
    Let \(\nodeFaulty_i\) be the gate map of \(\fault_i\). For each $v \in V_i$, we can decompose 
    \begin{align}
        \nodeFaulty_i(v) = (1-\epsilon) \nodeGate_i(v) + 2\epsilon \noise_i(v), \qquad \text{where } \|\noise_i(v)\|_{\diamond} \leq 1.
    \end{align}
    For a subset \(S \subseteq V_i\), let \(\nodeFaulty_{i, S}\) denote the gate map that takes the difference of the noisy and perfect gate whenever \(v \in S\).
    \begin{align}
        v \in V_i \qquad \nodeFaulty_{i, S}(v) = \begin{cases}
            2\epsilon \noise_i(v) & v \in S \\
            (1-\epsilon)\nodeGate_i(v) & v \notin S
        \end{cases}.
    \end{align}
    Let \(\tilde{\fault}_{i, S}\) denote a fault with this gate map.
    Using the sub-multiplicativity of the diamond norm, we note that 
    \begin{align}
        \|\circuitMap{C_i[\tilde{\fault}_{i, S}]}\|_{\diamond} \le (2\epsilon)^{|S|}(1-\epsilon)^{|V_i|-|S|}
    \end{align}
    since \(\circuitMap{C_i[\tilde{\fault}_{i, S}]}\) contains \(|S|\) operators each with diamond norm at most \(2\epsilon\) and the remaining operators each have diamond norm \(1-\epsilon\).
    For simplicity of notation, let us call a set \(S\) good if it is \(\badfaults_i\)-avoiding and bad otherwise.
    We can rewrite the execution of the noisy circuit in terms of these faults:
    \begin{align}
        \circuitMap{C_i[\fault_i]} 
        &= \sum_{S \subseteq V_i} \circuitMap{C_i[\tilde{\fault}_{i, S}]} \\
        &= \sum_{\text{good } S \subseteq V_i} \circuitMap{C_i[\tilde{\fault}_{i, S}]}
        + \sum_{\text{bad } S \subseteq V_i} \circuitMap{C_i[\tilde{\fault}_{i, S}]}.
    \end{align}
    Let us denote
    \begin{align}
        \cig &= \sum_{\text{good } S \subseteq V_i} \circuitMap{C_i[\tilde{\fault}_{i, S}]}, \qquad 
        \cib = \sum_{\text{bad } S \subseteq V_i} \circuitMap{C_i[\tilde{\fault}_{i, S}]}.
    \end{align}
    We can bound the total contribution of the `bad' fault paths within a single gadget in terms of diamond norm.
    \begin{align}
        \left \| \cib \right \|_\diamond 
        &= 
        \left \| \sum_{\substack{\text{bad } S \subseteq V_i}} \circuitMap{C_i[\tilde{\fault}_{i, S}]} \right \|_\diamond \\
        &\leq \sum_{F \in \mathcal{F}_i} \sum_{F \subseteq S \subseteq V_i} (2\epsilon)^{|F|} (2\epsilon)^{|S|-|F|}(1-\epsilon)^{|V_i|-|S|} \\
        & = \sum_{F \in \mathcal{F}_i} (2 \epsilon)^{|F|} (1+\epsilon)^{|V_i|-|F|}  \\
        &\le \sum_{F \in \mathcal{F}_i} (2(1+\epsilon)^{\ratio-1}\epsilon)^{|F|}\\
        &= \weightenum{\mathcal{F}_i}{2(1+\epsilon)^{\ratio-1}\epsilon}.
    \end{align}
    Set \(\eta = \max_i\weightenum{\mathcal{F}_i}{2(1+\epsilon)^{\ratio-1}\epsilon}\).
    By the triangular inequality of diamond norm, we see that \(\left \| \cig \right \|_\diamond  \leq 1+\eta\) for all \(i\).
    Using these upper bounds, we can bound the total contribution of the bad fault paths on the whole circuit.
    \begin{align}
        \circuitMap{\cft[\fault]} 
        &= \circuitMap{C_m[\fault_m]} \circ \dots \circ \circuitMap{C_1[\fault_1]}  \\
        &= (\C_{m,+} + \C_{m, -})\circ \dots \circ (\C_{1,+} + \C_{1, -}).
    \end{align}
    Denote 
    \begin{align}
        \C_+ &= \C_{m,+} \circ \C_{m-1,+} \circ \dots \C_{1,+}, \qquad \C_- = \circuitMap{\cft[\fault]} - C_+.
    \end{align}
    \(\C_-\) has a decomposition based on the last gadget that failed, as follows.
    \begin{align}
        \C_- &= \sum_{i=1}^m \C_{m,+}\circ \C_{m-1,+}\circ \dots \circ \C_{i,-} \circ (\C_{i-1,+} + \C_{i-1, -})\circ \dots \circ (\C_{1,+} + \C_{1, -}).
    \end{align}
    This decomposition enables us to bound the diamond norm of \(\C_-\). By the sub-multiplicativity of the diamond norm, and the fact that \(\C_{i,+} + \C_{i, -}\) has diamond norm \(1\) for all \(i\),
    \begin{align}
        \left \| \C_-  \right \|_{\diamond}
        &\leq \sum_{i=1}^m \left \|
        \C_{m,+}\circ \C_{m-1,+}\circ \dots \circ \C_{i,-} \circ (\C_{i-1,+} + \C_{i-1, -})\circ \dots \circ (\C_{1,+} + \C_{1, -})
        \right \|_{\diamond} \\
        &\leq \sum_{i=1}^m (1+\eta)^{m-i}\eta \\
        &\leq m(1+\eta)^m\eta.
    \end{align}
    On the other hand, since \(\C_+\) is the sum over noisy executions of \(\cft\) with \(\badfaults\)-avoiding faults, it holds by the definition of gadget that
    \begin{align}
        \C_+ &= \alpha\cdot \circuitMap{C}.
    \end{align}
    for some \(\alpha\in \mathbb{C}\). Now fix an input state \(\ketbra{x}\), we write
    \begin{align}
        &\tr \left| \circuitMap{C}\left(\ketbra{x}{x}\right) - \circuitMap{\cft[\fault]}(\ketbra{x}{x})\right| \\
        &= \tr \left| \circuitMap{C}\left(\ketbra{x}{x}\right) - (\C_+ + \C_-)(\ketbra{x}{x})\right| \\
        &\le \tr \left| \circuitMap{C}\left(\ketbra{x}{x}\right) - \alpha \cdot\circuitMap{C}\left(\ketbra{x}{x}\right)\right| + \tr \left| \C_-\ketbra{x}{x}\right| \\
        &\le |1-\alpha| + \|C_-\|_{\diamond}\\
        &\le |1-\alpha| + m(1+\eta)^m\eta,
    \end{align}
    where we used here that the trace norm is bounded by the diamond norm. 
    To bound \(|1-\alpha|\), note that 
    \begin{align}
        1 = \tr(\circuitMap{\cft[\fault]}(\ketbra{x}{x}))
        &= \tr(C_+\ketbra{x}{x}) + \tr(C_-\ketbra{x}{x}) = \alpha + \tr(C_-\ketbra{x}{x}). 
    \end{align}
    Therefore we can bound
    \begin{align}
        |1-\alpha| &= |\tr(C_-\ketbra{x}{x})| \le \tr\|C_-\ketbra{x}{x}\| \le \|C_-\ketbra{x}{x}\|_{\diamond} = m(1+\eta)^m\eta.
    \end{align}
    We conclude that 
    \begin{align}
        \frac{1}{2}\tr \left| \circuitMap{C}\left(\ketbra{x}{x}\right) - \circuitMap{\cft[\fault]}(\ketbra{x}{x})\right|
        &= m(1+\eta)^m\eta. \qedhere
    \end{align}
\end{proof}

\section{Application: qLDPC codes and transversal gates}\label{sec:qldpc-gadgets}
We show that a standard construction of error correction gadgets (originating from \cite{dennis2002topological,gottesman2014fault}) for quantum-low density parity check (qLDPC) codes satisfies our gadget definition.

\label{sec:surface:ec-gadget}

\subsection{Correctable sets}
A CSS code defined by the check matrices \((H_X, H_Z)\) is said to be \(\Delta\)-qLDPC if every row and column of \(H_X\) and \(H_Z\) has weight at most \(\Delta\).
\nomenclature[E, 01]{\(H_X, H_Z, \Delta\)}{A CSS code has stabilizer check matrices \(H_X, H_Z\). We say that the code is \(\Delta\)-qLDPC if every row and column of \(H_X\) and \(H_Z\) has weight at most \(\Delta\).}

For an error rate \(p\) (IID) on \(n\) qubits, on average, errors have weight around \(np\).
Thus, naively one might conclude that linear distance scaling is a necessary condition for a code family to achieve a threshold.
This turns out to be too strong: It suffices to be able to correct ``most'' errors of weight \(np\).
It turns out that the qLDPC property imposes strong constraints\footnote{More specifically, the qLDPC property induces a sort of ``geometry'' on the space of qubits, and logical operators must be the union of large connected components.} on the low weight logical operators of the code, so that it is rare for a random error to have large overlap with a logical operator.
\cite{kovalev2013fault} showed that all qLDPC code families with polynomial distance scaling posses a threshold -- below a critical error rate, recovery from depolarizing noise fails with exponentially small (in the code distance) probability when decoding with a minimum weight decoder.

We begin by defining a minimum-weight decoder for a classical linear code as well as a minimum-weight decoder in the \(X/Z\) basis which we refer to as the CSS minimum-weight decoder.
\begin{definition}[Minimum-weight decoder]\label{def:min-weight-decoder}
    For a classical code defined by a \(m\times n\) parity check matrix \(H\), given a syndrome \(s\in \F^m\) in the column space of \(H\), a minimum weight decoder outputs a minimum Hamming weight vector in the set of valid corrections, \(\{x\in \F^m: Hx = s\}\). 
    For a CSS code defined by \((H_X, H_Z)\), a CSS minimum weight decoder \(D\colon \F^{n-k} \mapsto \{I,X,Y,Z\}^n\) consists of two minimum weight decoders, \(D_X, D_Z\), for matices \(H_X, H_Z\). 
    For a syndrome \(s = (s_X,s_Z)\), where \(s_X, s_Z\) are the \(X\) and \(Z\) syndromes respectively,
    \begin{align}
        D(s) = X^{D_Z(s_Z)}\cdot Z^{D_X(s_X)}.
    \end{align}
    Here for a vector \(v\in \F^n\), we use \(X^v\) to denote the Pauli operator \(X^{v_1}\otimes \dots X^{v_n}\), where \(X^1 = X\) and \(X^0 = I\). \(Z^v\) is defined similarly.
\end{definition}

We now need a parameterized notion of bad sets (from \cite{gottesman2014fault}) of the code.
These are the sets that will trick our decoder into applying the wrong correction.
\begin{definition}[Adjacency graph]\label{def:adjacency-graph}
    For a check matrix \(H \in \F^{r \times n}\), the adjacency graph of \(H\), denoted \(\adjgraph{H}\) is a graph on the vertices \([n]\), where two vertices \(i \ne j \in [n]\) are adjacent if there exists a row \(k\) which is non-zero on columns \(i\) and \(j\). I.e. if bit vertices \(i\) and \(j\) are connected by a constraint vertex \(k\) in the Tanner graph associated with \(H\).
\end{definition}
\nomenclature[E, 02]{\(\adjgraph{H}\)}{The adjacency graph of a parity check matrix \(H\), see \cref{def:adjacency-graph}.}

We will now associate length-\(n\) bitstrings with subsets of \([n]\).
\begin{proposition}[\cite{kovalev2013fault}]
    For a pair of check matrices \((H_X,H_Z)\) defining a CSS code of distance \(d\), any non-trivial \(X\) logical operator with support \(x\) must induce a subgraph of \(\adjgraph{H_Z}\) with at least one connected component of size at least \(d\).
\end{proposition}

\begin{definition}[Clustering sets]\label{def:clustering-sets}
    For a graph \(G=(V,E)\) the \((d,t)\)-clustering sets \(\clusterset_{(d,t)}\subseteq P(V)\) of \(G\) are the subsets \(W \subseteq U \subseteq V\) of \(V\) of size \(|W| = t\) contained in a subset \(U\) of size \(|U| = d\) that induces a connected subgraph of \(G\).
    \begin{align}
        \clusterset_{(d,d)} &= \left\{U \subseteq V \mid |U| = d \text{ and \(U\) induces a connected subgraph of \(G\)}\right\} \\
        t<d \quad \clusterset_{(d,t)} &= \bigcup_{U \in \clusterset_{(d,d)}} \left\{W \subseteq U \mid |W| = t\right\}
    \end{align}
\end{definition}
\nomenclature[E, 03]{\(\clusterset_{(d,t)}\)}{The \((d,t)\)-clustering sets of a graph \(G = (V, E)\) which are all size \(t\) subsets of vertices of  \(d\)-vertex connected subgraphs of \(G\). See \cref{def:clustering-sets}.}

The \(t\) parameter of these clusters is the appropriate analog of the weight of an error for qLDPC code families with sublinear distance: The error cannot be too high weight in any one region.
As we will see, this parameter will behave much like Hamming weight in that we can recovery a sort of sub-additivity, see \cref{lemma:union-correctable-sets}.

We give upper bounds on the weight enumerator of the clustering sets by bounding the number of bad sets analogously to \cite{gottesman2014fault}.
We first need the following lemma from \cite{aliferis2007accuracy}.\footnote{We actually require the stronger variant proven (but not stated) in \cite{aliferis2007accuracy}. See Lemma 2 of \cite{gottesman2014fault}.}
\begin{lemma}[Cluster counting \cite{aliferis2007accuracy} Lemma 5]\label{lemma:cluster-counting}
    Let \(G = (V,E)\) be a graph with maximum degree \(\Delta\).
    Fix a subset \(S \subseteq V\) of vertices.
    Then, consider subsets \(C \subseteq V\) that is a union of connected sets \(C = \cup_i C_i\) in \(G\) where each set \(C_i\) contains at least one element of \(S\), \(C_i \cap S \ne \varnothing\).
    We say that \(C\) is a \textbf{cluster cover} of \(S\), and denote the collection of cluster covers of size \(m\) to be \(\cc(m, S)\).
    The number of such subsets of size \(m\), \(M(m,S)\) satisfies the upper bound
    \begin{align}
        M(m,S) \le e^{|S| - 1} (\Delta e)^{m-|S|}
    \end{align}
\end{lemma}
\nomenclature[E, 04]{\(\cc(m, S)\)}{For a set of vertices \(S\) of graph \(G\), \(\cc(m, S)\) is the collection of size \(m\) cluster covers of \(S\), see \cref{lemma:cluster-counting}.}

This allows us to upper bound the weight enumerator of this family.
\begin{lemma}\label{lemma:weight-enume-cluster-sets}
    For a maximum degree-\(\Delta\) graph \(G=(V,E)\) with the \((d,t)\)-clustering sets \(\clusterset_{(d,t)}\subseteq P(V)\), we have that (\(x \in [0,1]\))
    \begin{align}
        \weightenum{\clusterset_{(d,t)}}{x} \le \left(\binom{d}{t - 1} (\Delta e)^{d-1} |V|\right) x^{t}
    \end{align}
\end{lemma}
\begin{proof}
    We compute the number of elements of \(\clusterset_{(d,t)}\) by taking a union bound over the following decomposition:
    \begin{align}
        \clusterset_{(d,t)} \subseteq \bigcup_{v \in V} \bigcup_{\substack{U \subseteq V \\ |U| = d \\ \text{\(U\) is connected in \(G\)}}} \left\{W \subseteq U \mid |W| = t \text{ and } v \subseteq W\right\}
    \end{align}

    \cref{lemma:cluster-counting} upper bounds the number of connected (in \(G\)) sets \(U \subseteq V\) of size \(|U|=d\) containing a particular vertex \(v \in V\) as \((\Delta e)^{d-1}\).
    There are \(\binom{d}{t-1}\) subsets of \(U\) that contain \(v\).
    By construction, each element of \(\clusterset_{(d,t)}\) has weight \(t\).
    The result follows.
\end{proof}

We are now ready to define our code types.
From now on, fix a \(\Delta\)-qLDPC code \(\qcode\) with parameters \(\dsl n, k, d \dsr\) defined by check matrices \((H_X,H_Z)\) and an arbitrary choice of irreversible encoding map implementable by a Clifford circuit \(M\) acting on a logical state \(\mixedstate{\hilbert^k}\) and an \(n-k\) ancilla register initialized to \(\ket{0}^{\otimes (n-k)}\).
Call this the \textbf{computational code}.
\(M\) satisfies the property that for some choice of \(n-k\) generators of the stabilizer \(\{s_i\}_{i=1}^{n-k} \subseteq \{I,X,Y,Z\}^{n-k}\), any eigenstate in eigenspace \(s \in \F^{n-k}\) of the stabilizer generators (the \textbf{syndrome space} \(s\)) can be written as \(M\left(\ket{\psi} \otimes \ket{s}\right)\) for some state \(\ket{\psi} \in \hilbert^k\).

\begin{definition}[Purified decoder]\label{def:purified-decoder}
    Let \(\sigma \colon \{I,X,Y,Z\}^n \mapsto \F^{n-k}\) be the syndrome map with respect to the generating set of the stabilizer \(\{s_i\}_{i=1}^{n-k} \subseteq \{I,X,Y,Z\}^n\) of \(\qcode\) and \(D\colon \F^{n-k} \mapsto \{I,X,Y,Z\}^n\) be a decoding algorithm outputting a correction such that \(\sigma \circ D (s) = s\).
    We can define the purification \(U_{D}\) of \(D\) as the unitary \(U\) that measures the syndrome \(s\), applies the corresponding correction \(D(s)\), and places the value of the syndrome in the ancilla register after applying the unencoding circuit \(M^{\dagger}\).
    Define the syndrome space projector \(\Pi_s = M\left(I\otimes \ketbra{s}{s}\right)M^{\dagger}\).
    Then \(U_{D}\) can be written in terms of a sum over syndrome spaces and the corresponding corrections.
    \begin{align}\label{eq:ldpc:unitary_decoder}
        U_{D} = \sum_{s \in \F^{n-k}} (I\otimes \ketbra{s}{0}) M^{\dagger} D(s) \Pi_s
    \end{align}
\end{definition}

\begin{definition}[Code types for computational code]\label{def:qldpc-codetype}
    Let \(\clusterset_{(m,t)}^{(X)}\) and \(\clusterset_{(m,t)}^{(Z)}\) be the \((m,t)\)-clustering sets in \(\adjgraph{H_X}\) and \(\adjgraph{H_Z}\), respectively.
    Let \(c \in [0,1]\).
    All the bad error supports will be stated in terms of \(c\).
    We define the bad error supports to be all sufficiently large subsets of connected subsets of size at least \(d\).
    \begin{align}
        \baderrors^{(X)}_c &= \boxplus^{n}_{m=d} \clusterset_{(m,\lceil c\cdot m/2\rceil)}^{(X)} \\
        \baderrors^{(Z)}_c &= \boxplus^{n}_{m=d} \clusterset_{(m,\lceil c\cdot m/2\rceil)}^{(Z)}
    \end{align}
    We define \(\baderrors_c \equiv \baderrors^{(X)}_c \boxplus \baderrors^{(Z)}_c\).

    The reversible decoding unitary \(U_{\qcode}\) is the purification of a CSS minimum-weight decoder (see \cref{def:min-weight-decoder}) for \(\qcode\), as defined in \cref{eq:ldpc:unitary_decoder}. 
    The code type \(\compCodeType[c]\) is defined to be the pair of bad sets and the purified minimum weight decoder \(\left(U_{\qcode}, \baderrors_c \equiv \baderrors^{(X)}_c \boxplus \baderrors^{(Z)}_c\right)\).
\end{definition}
\nomenclature[E, 10]{\((U_{\qcode}, \baderrors_c)\)}{\(\compCodeType[c] = (U_{\qcode}, \baderrors_c)\) is the code type of a computational qLDPC code \(\qcode\), parametrized by a constant \(c\in [0,1]\). \(U_{\qcode}\) is the purification of a decoder. See \cref{def:purified-decoder} and \cref{def:qldpc-codetype}.}

Note that \(U_{\qcode}\) inverts \(M\) on the trivial syndrome space and can be easily confirmed to be unitary.
\begin{align}
    \forall \ket{\psi} \in \hilbert^k \quad \ket{\psi} = (U_{\qcode} M) \ket{\psi} \ket{0}^{\otimes n-k}
\end{align}

We now show that \(\compCodeType[c]\) is a valid code type \cref{def:code-type}.
\begin{proposition}[Correctable errors~\cite{kovalev2013fault,gottesman2014fault}]
    For \(c \in [0,1]\), let \(\compCodeType[c] = (U_{\qcode}, \baderrors_c)\).
    Fix a logical state \(\rho \in \mixedstate{\hilbert^k}\) and a state \(\sigma \in \unphysicalstate{\hilbert^n}\) such that \(\sigma\) is \(\baderrors_c\)-deviated from the encoding of \(\rho\): \(\encoder(\rho) \dev{\baderrors_c} \sigma\).

    Then, the decoding of \(\sigma\) is proportional to \(\rho\).
    \begin{align}
        \decoder(\sigma) \propto \rho
    \end{align}
\end{proposition}
\begin{proof}
    Let \(\error\) be a superoperator whose support \(S\subseteq [n]\) is \(\baderrors_c\)-avoiding. 
    First, we write the Kraus decomposition of the superoperator in terms of Pauli operators \(\{K_\mu\}_\mu, 
    \{K_{\nu}'\}_\nu\) supported on \(S\subseteq [n]\) and complex coefficients \(\{\alpha_{\mu,\nu}\}_{\mu,\nu}\)
    \begin{align}
        \error\circ \encoder(\rho) &= \sum_{\mu,\nu}\alpha_{\mu,\nu} K_\mu M (\rho\otimes \ketbra{0}{0})M^{\dagger} K_{\nu}'
    \end{align}
    Using this decomposition and the definition of \(U_{\qcode}\) (\cref{eq:ldpc:unitary_decoder}), we consider the action of \(\decoder\) on \(\sigma\).
    \begin{align}
        &\decoder(\sigma) = \tr_{\hilbert^{r}} \left( U_{\qcode} (\error\circ \encoder (\rho) )U_{\qcode}^\dagger \right) \\
        &= \tr_{\hilbert^{r}} \left( \sum_{s,t \in \F^{n-k}} (I\otimes \ketbra{s}{0}) M^{\dagger} D(s) \Pi_s \left( \sum_{\mu,\nu}\alpha_{\mu,\nu} K_\mu M (\rho\otimes \ketbra{0}{0})M^{\dagger} K_{\nu}' \right) \Pi_t D(t) M(I\otimes \ketbra{0}{t})\right) \\
        &= \tr_{\hilbert^{r}} \left( \sum_{s,t \in \F^{n-k}} \left( \sum_{\mu,\nu}\alpha_{\mu,\nu}
             (I\otimes \ketbra{s}{0}) M^{\dagger} D(s) \Pi_s K_\mu M (\rho\otimes \ketbra{0}{0})M^{\dagger} K_{\nu}' \Pi_t D(t) M(I\otimes \ketbra{0}{t})\right) \right)
    \end{align}
    Let \(s_\mu, s_\nu\) denote the syndrome of \(K_\mu, K_{\nu}'\). 
    The projectors \(\Pi_s\) and \(\Pi_t\) annihilate all terms except for those where \(s_\mu = s\) and \(s_\nu = t\).
    Suppose \(D(s)K_\mu\) and \(D(t)K_{\nu}'\) are both stabilizers, we can write
    \begin{align}
        D(s) \Pi_s K_\mu M (\rho\otimes \ketbra{0}{0})M^{\dagger} K_{\nu}' \Pi_t D(t)
        &= M (\rho\otimes \ketbra{0}{0})M^{\dagger}. 
    \end{align}
    We can then simplify the equation for \(\decoder(\sigma)\).
    \begin{align}
        \decoder(\sigma)
        &= \tr_{\hilbert^{r}} \left( \sum_{s,t \in \F^{n-k}}  \left( \sum_{\mu,\nu: s_\mu = s, s_\nu = t}\alpha_{\mu,\nu}
        (I\otimes \ketbra{s}{0}) M^{\dagger} M (\rho\otimes \ketbra{0}{0})M^{\dagger} M(I\otimes \ketbra{0}{t})\right) \right) \\
        &= \tr_{\hilbert^{r}} \left( \sum_{s,t \in \F^{n-k}} \left( \sum_{\mu,\nu: s_\mu = s, s_\nu = t}\alpha_{\mu,\nu} \rho\otimes \ketbra{s}{t} \right) \right).
    \end{align}
    The partial trace annihilates all terms where \(s\ne r\). Let \(\alpha_s\) be the sum of all \(\alpha_{\mu,\nu}\) where \(s_\mu = s_\nu = s\). Then the above simplifies to
    \begin{align}
        \decoder(\sigma) &= \sum_{s\in \F^{n-k}}\alpha_s \rho.
    \end{align}
    It suffices for us to show that for 
    an arbitrary Pauli error \(E\) with \(\baderrors^{(X)}_c \boxplus \baderrors^{(Z)}_c\)-avoiding support, the minimum weight decoder will correct \(E\) up to a stabilizer. 
    Decompose the Pauli error \(E = p E_X E_Z\) into \(X\) and \(Z\) parts \(E_X \in \{X,I\}^{n}\) and \(E_Z \in \{Z,I\}^{n}\) up to a scalar \(p \in \{\pm 1, \pm i\}\).
    Consider the correction of \(E_X\), the case of \(E_Z\) is identical with the roles of \(Z\) and \(X\) interchanged.
    
    Since \(E\) has \(\baderrors_c\)-avoiding support, \(E_X\) must have \(\baderrors^{(Z)}_c = \boxplus^{n}_{m=d} \clusterset_{(m,\lceil c\cdot m/2\rceil)}^{(Z)}\)-avoiding support.
    The minimum weight decoder applies a minimum weight correction \(C_X \in \{I,X\}^n\) such that \(C_X E_X\) has trivial syndrome.
    We will show that \(C_X E_X\) is in the stabilizer of the code, following \cite{kovalev2013fault,gottesman2014fault}.

    We first establish some properties of the corrected error \(C_X E_X\).
    \(\supp (C_X E_X)\) induces a subgraph of the adjacency graph \(\adjgraph{H_Z}\) that has connected cluster \(L_1, \dots, L_\ell\).
    From the definition of the adjacency graph, each connected cluster must have trivial syndrome (commute with the stabilizer of the code) if \(C_X E_X\) does.
    Thus, the restriction of \(C_X E_X\) to any connected cluster \(L'\) of size strictly less than the code distance \(d\) must be in the stabilizer of the code.
    It also must be the case that on \(L'\), the error is supported on at least half of the set, \(|L' \cap \supp E_X|\ge |L'|/2\).
    Otherwise, \(E_X|_{L'}\) would be a lower weight correction than \(C_X|_{L'}\) for \(E_X|_{L'}\).
    
    Now suppose, for a contradiction, that \(C_X E_X\) is not a stabilizer of the code.
    Then, at least one connected cluster \(L\) of \(\supp (C_X E_X)\) in \(\adjgraph{H_Z}\) must have size \(m \ge d\).
    Since \(L\) satisfies \(|L \cap \supp E_X|\ge |L|/2\), \(\supp E_X\) contains a subset of \(L\) of size at least \(m/2\).
    So this subset this is a superset of an element of \(\clusterset_{(m,\lceil c\cdot m/2\rceil)}^{(Z)}\), contradicting the assumption that \(E\) was \(\baderrors_c\)-avoiding.

    An identical argument holds for \(C_Z E_Z\), so \(C_Z E_Z C_X E_X\) is in the stabilizer of the code and \(E\) is corrected by the CSS minimum-weight decoder.
    This completes our proof.
    \end{proof}

We will also need to be able to analyze the structure of unions of sets that avoid the clustering sets.

\begin{definition}[Separable Bad Sets]\label{def:separable-bad-sets}
    For \(\baderrors,\baderrors_1, \baderrors_2\subseteq \power(\Omega)\), we say that \(\baderrors\) is \textbf{separable} into \(\baderrors_1, \baderrors_2\) if for all \(\baderrors_1\)-avoiding set \(X\) and \(\baderrors_2\)-avoiding set \(Y\), \(X\cup Y\) is \(\baderrors\)-avoiding. 
\end{definition}

\begin{proposition}[Union of avoiding sets]\label{lemma:union-correctable-sets}
    Let \(G=(V,E)\) be a graph, and \(\clusterset_{(d,t)}\) be the \((d,t)\)-clustering sets of \(G\).
    For \(s,t \in \mathbb{N}\) such that \(s+t \le d\), \(\clusterset_{(d,{s+t})}\) is separable into \(\clusterset_{(d,{s})}\) and \(\clusterset_{(d,{t})}\).
\end{proposition}
\begin{proof}
    Consider \(X,Y\subseteq [V]\) such that \(X\) is \(\clusterset_{(d,s)}\)-avoiding and \(Y\) is \(\clusterset_{(d,t)}\)-avoiding, let \(Z = X \cup Y\).
    Suppose, for a contradiction, that \(Z\) is not \(\clusterset_{(d,{s+t})}\)-avoiding.
    Then, there exists a subset of vertices \(U \subseteq V\) of size \(|U| = d\) such that \(U\) induces a connected subgraph of \(G\) and \(|U \cap Z| \ge s+t\).

    However, since \(X\) and \(Y\) are \(\clusterset_{(d,s)}\)-avoiding and \(\clusterset_{(d,t)}\)-avoiding, respectively, we have that their intersections with \(U\) must be small: \(|U \cap X| < s\) and \(|U \cap Y| < t\).
    The contradiction follows.
    \begin{align}
        s+t \le |U \cap Z| &\le |U\cap X| + |U\cap Y| < s + t. \qedhere
    \end{align}
\end{proof}

It is a straightforward corollary that if we combine errors (e.g. by applying a two-qubit gate transversally) then the resulting error is well controlled.
\begin{corollary}\label{lemma:separable-bad-errors}
    For \(c+c' \le 1\), \(\baderrors_{c+c'}\) is separable into \(\baderrors_{c}\) and \(\baderrors_{c'}\). 
\end{corollary}
\begin{proof}
    Apply \cref{lemma:union-correctable-sets} to the clustering sets in the definition of \(\baderrors_c\).
\end{proof}

Concluding our definition of the computational code type, we bound the weight enumerator of \(\baderrors_c\).
\begin{lemma}\label{lemma:weight-enum-baderrors-c}
    \begin{align}
        \weightenum{\baderrors_c}{x} 
        &\le n\cdot [O_{\Delta,c}(x)]^{\frac{cd}{2}}.
    \end{align}
\end{lemma}
\begin{proof}
    Since \(H_X, H_Z\) has row and column weights bounded by \(\Delta\), the graphs \(\adjgraph{H_X}, \adjgraph{H_Z}\) has degree bounded by \(\Delta^2\). 
    We can bound the weight enumerator of \(\baderrors_c\) with \cref{lemma:weight-enume-cluster-sets}. 
    \begin{align}
        \weightenum{\baderrors^{(X)}_c}{x} 
        &\le \sum_{m\ge d} \weightenum{\clusterset^{(X)}_{(m,\lceil c\cdot m/2\rceil)}}{x} \\
        &\le \sum_{m\ge d} n\cdot (2e\Delta^2)^m x^{\frac{cm}{2}} \\
        &\le n\cdot [O_{\Delta,c}(x)]^{\frac{cd}{2}}
    \end{align}
    Since \(\baderrors_c = \baderrors^{(X)}_c\boxplus \baderrors^{(Z)}_c\), we multiply the above upper bound by a factor of 2, which is absorbed into the constants in \(O_{\Delta,c}(x)\). 
\end{proof}

\subsection{Spacetime code}\label{sec:spacetime-code}
Before we are ready to construct and prove properties of the error correction gadget, we first need to construct the following classical code commonly referred to as a ``spacetime code.''
In what follows, we will use the CSS property to focus on correction of \(X\) errors.
The case of \(Z\) errors will follow with \(X\) and \(Z\) interchanged.
\begin{definition}[Spacetime code]\label{def:spacetime-code}
    Let \((H_X,H_Z)\) be a CSS code.
    The spacetime code of \(H_Z \in \F^{r_Z \times n}\) is a modified check matrix that corresponds to the spacetime history of repeated syndrome extraction.\footnote{Syndrome bits of this code are sometimes called ``detectors.''}

    The spacetime code check matrix \(H \in \F^{(T r_Z) \times T(n + r_Z)}\) will have \(T\) copies of \(H_Z\) along the diagonal and \(T\) blocks of \(r_Z\) extra columns corresponding to syndrome measurement errors.
    Each measurement error will be able to flip the syndrome of a check in two consecutive rounds.
    
    We describe this more explicitly.
    In what follows, for two intervals \(I\) and \(J\), we will use \(H[I,J]\) to denote the sub-matrix with rows in \(I\) and columns in \(J\).
    For \(t \in [T]\), we define the data error intervals and syndrome error intervals (see \cref{fig:spacetime-check-matrix}).
    \begin{align}
        I_t &= [(t-1)\cdot r_Z + 1, ~t\cdot r_Z - 1]\\
        J_t^{\mathrm{d}} &= [(t-1)\cdot n + 1, ~t\cdot n - 1]\\
        J_t^{\mathrm{s}} &= [(t-1)\cdot r_Z + T n, ~ t\cdot r_Z + T n - 1]
    \end{align}
    We can now specify the entries of \(H\) (all unspecified entries are 0).
    \begin{align}
        t \in [T]& & H[I_t, J_t^{\mathrm{d}}] &= H_Z\\
        t \in [T]& & H[I_t, J_t^{\mathrm{s}}] &= \mathsf{Id}_{r_Z \times r_Z} \\
        t \in [T-1]& & H[I_t, J_{t+1}^{\mathrm{s}}] &= \mathsf{Id}_{r_Z \times r_Z}
    \end{align}
    We will also require an ``interpretation'' of the results of decoding the spacetime code.
    Define the \textbf{flattening map} \(\phi \colon \F^{T n + Tr_z} \to \F^{n}\) given by dropping the syndrome error blocks and summing together the data
    error blocks \(x \mapsto \sum_{t \in [T]} x[J_t^{\mathrm{d}}]\).
\end{definition}
\nomenclature[E, 11]{\(H, T, I_t, J_t, \phi\)}{\(H\) is the spacetime check matrix (commonly referred to as the detector check matrix) of \(T\) rounds of measurements of \(H_Z\). For \(t\in [T]\), \(I_t\) denotes the detectors at time \(t\), \(\Jd_t, \Js_t\) denote the data bits and syndrome bits. \(\phi\) is the flattening map. See \cref{def:spacetime-code}.}
\begin{figure}
    \centering
    \includegraphics[width=0.5\linewidth]{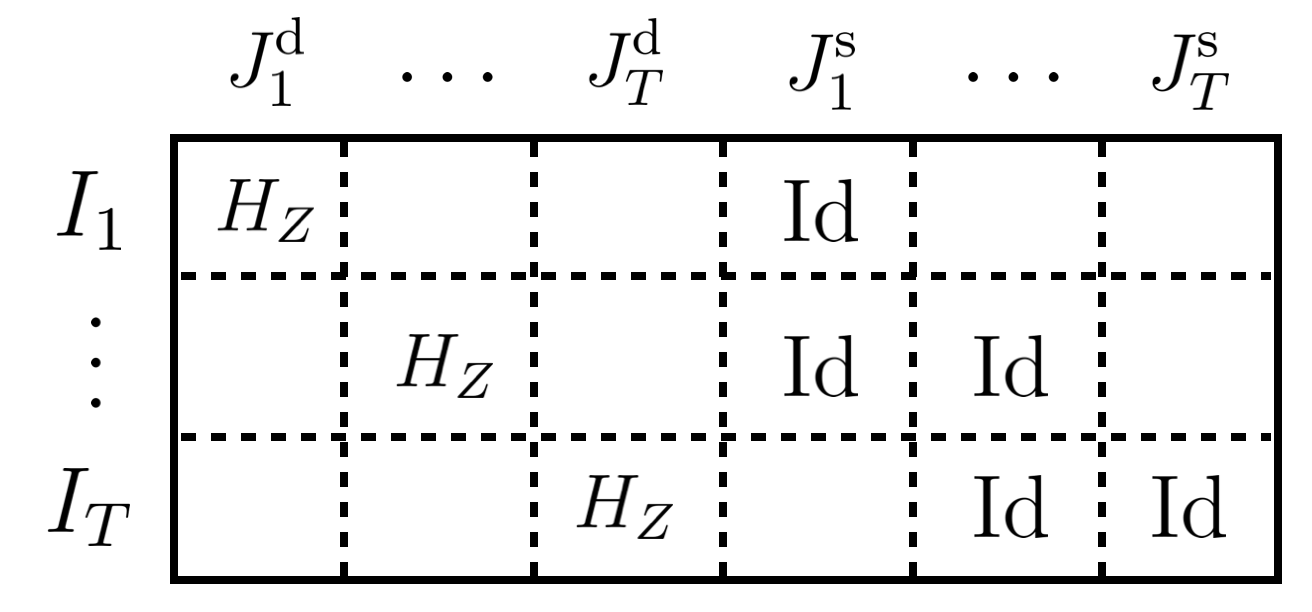}
    \caption{Spacetime code check matrix}
    \label{fig:spacetime-check-matrix}
\end{figure}
The rows \(I_t\) correspond to the change in syndrome measurements in the \(t\)-th round.
The columns \(J_t^{\mathrm{s}}\) correspond to syndrome measurement errors affecting the syndrome at rounds \(t-1\) and \(t\) while the columns \(J_t^{\mathrm{d}}\) correspond to new data errors accumulated in round \(t\).

For our error correction gadget, we will consider the \textbf{input error}, which corresponds to bits of \(J_{1}^{\mathrm{d}}\), separately from the remaining error. 
For this purpose, we define the \textbf{bulk} of the spacetime code with bits \(\Jbulk := [T(n+r_Z)]\setminus \Jd_1\) and checks \(\Ibulk = I_2\cup \dots \cup I_T\) namely, all data and syndrome bits except for the input error bits and all rows except for the first.
We denote this \textbf{bulk check matrix} \(\Hbulk\).
Similarly, we define the \textbf{bulk flatting map}
\(\phi_\bulk \colon \F^{(T-1) n + Tr_z} \to \F^{n}\) as \(\phi_\bulk(x) = \sum_{t=2}^{T} x[J_t^{\mathrm{d}}]\).
We call the last \(r_Z\) bits \(J_{T}^{\mathrm{s}}\) of a bitstring \(\F^{T n + Tr_z}\) the \textbf{final round syndrome error} bits, and the first \(r_Z\) bits \(J_{1}^{\mathrm{s}}\) the \textbf{initial round syndrome error}.
\nomenclature[E, 12]{\(\Hbulk, \phi_\bulk\)}{The bulk check matrix and flattening map, which is defined with \(\Jd_1\) and \(I_1\) removed.}

We will now prove some properties about the spacetime code: Let \(d\) be the distance of the original CSS code.

\begin{lemma}[Final round syndrome]\label{prop:final-round-syndrome}
    For a bitstring \(x \in \F^{T n + Tr_Z}\) with trivial spacetime syndrome \(Hx = 0\), the syndrome of the flattened bitstring \(\phi(x)\) is equal to the final round syndrome error.
    \begin{align}
        H_Z \phi(x) = x[J_{T}^{\mathrm{s}}]
    \end{align}
\end{lemma}
\begin{proof}
    \(x\) has zero syndrome with respect to the spacetime code \(Hx = 0\) which implies (by definition of \(H\)) the following system of equations.
    \begin{align}
        t \in [T-1]\quad H_Z x[J_{t+1}^{\mathrm{d}}] + x[J_{t}^{\mathrm{s}}] + x[J_{t+1}^{\mathrm{s}}] &= 0 \\
        H_Z x[J_{1}^{\mathrm{d}}] + x[J_{1}^{\mathrm{s}}] &= 0
    \end{align}
    Summing over this linear system yields the result.
    \begin{align}
        \sum_{t \in [t]} H_Z x[J_{t}^{\mathrm{d}}] + x[J_{T}^{\mathrm{s}}] &= H_Z \phi(x) + x[J_{T}^{\mathrm{s}}] = 0. 
        \qedhere
    \end{align}
    \qedhere
\end{proof}

\begin{lemma}\label{cor:bulk-syndrome}
    Similarly, for a bitstring \(x_\bulk \in \F^{(T-1) n + Tr_Z}\cong \F^{\Jbulk}\) with no bulk spacetime syndrome \(\Hbulk x = 0\), its flattened syndrome is the difference between the initial and final round syndrome errors.
    \begin{align}
        H_Z \phi_\bulk(x_\bulk) = x_\bulk[J_{T}^{\mathrm{s}}] + x_\bulk[J_{1}^{\mathrm{s}}]
    \end{align}
\end{lemma}
\begin{proof}
    As in the proof of \cref{prop:final-round-syndrome}, the zero syndrome condition implies the following equation from which the result follows.
    \begin{align}
        H_Z \phi_\bulk(x_\bulk) + x_\bulk[J_{T}^{\mathrm{s}}] + x_\bulk[J_{1}^{\mathrm{s}}] &= 0 \qedhere
    \end{align}
\end{proof}

\begin{lemma}\label{prop:successful_correction}
    For a bitstring \(x \in \F^{T n + Tr_Z}\) with trivial spacetime syndrome \(Hx = 0\) and trivial final round syndrome error \(x[J_{T}^{\mathrm{s}}]= 0\), if \(x\) is the union of connected components in the adjacency graph of the spacetime code \(\adjgraph{H}\) of size strictly less than the distance \(<d\), then \(x\) flattens to a trivial logical operator of the CSS code, \(\phi(x) \in \rowspan(H_X)\).   
\end{lemma}
\begin{proof}
    Consider a single connected component \(x'\) of \(x\) in \(\adjgraph{H}\).
    Since no other row of \(H\) with support on \(x'\) has support on \(x - x'\), it must be that \(Hx' = 0\) which implies \(H_Z \phi(x') = 0\) by \cref{prop:final-round-syndrome} and the assumption on the final round syndrome bits \(x[J_{T}^{\mathrm{s}}] = 0\).
    Using the assumption on the size of each connected component, the Hamming weight of the flattened bitstring is less than the distance \(|\phi(x')| \le |x'| < d\).
    
    \(\phi(x')\) is in the kernel of \(H_Z\) and has Hamming weight less than \(d\), so by the distance of the CSS code \((H_X,H_Z)\), it must be in the rowspace of \(H_X\), \(\phi(x') \in \rowspan(H_X)\).
    It then follows that every connected component of \(x\) flattens to an element of the row space of \(H_X\).
    \(\phi(x)\) is the sum over the flattening of every connected component so it is also in the row space of \(H_X\).
\end{proof}

\begin{definition}[Extended spacetime decoding graph]\label{def:extended-decoding-graph}
To bound the residual error on the code after correction, we introduce a slight modification to \(\adjgraph{H}\). 
We expand \(H\) by \(n\) columns and \(r_Z\) rows, which we denote \(J_{T+1}^{\mathrm{d}}\) and \(I_{T+1}\).
We denote this new matrix \(\Hext\) with the additional entries 
\begin{align}
    \Hext[I_{T+1}, J_{T+1}^{\mathrm{d}}] &= H_Z, \qquad 
    \Hext[I_{T+1}, J_{T}^{\mathrm{s}}] = \mathsf{Id}_{r_Z \times r_Z}~.
\end{align}
The graph \(\adjgraph{\Hext}\), compared to \(\adjgraph{H}\), has one more layer of \(n\) vertices labeled by \(J_{T+1}^{\mathrm{d}}\) which are connected to vertices labeled by \(J_{T}^{\mathrm{s}}\). 
We say that \(\adjgraph{\Hext}\) is the \textbf{extended spacetime decoding graph}.
\end{definition}
\nomenclature[E, 13]{\(\Hext\)}{The extended spacetime check matrix with residue error \(\Jd_{T+1}\), see \cref{def:extended-decoding-graph}.}

\begin{lemma}[Residual Error~\cite{gottesman2014fault}]\label{prop:residue-error-support}
    For a bitstring \(x \in \F^{T n + Tr_Z}\) with trivial spacetime syndrome \(Hx = 0\), 
    let \(y\in \F^{n}\) be a minimum weight vector (supported on \(J_{T+1}^{\mathrm{d}}\)) such that \(H_Z\phi(x) = H_Zy\). 
    Let \(S\subseteq J_{T+1}^{\mathrm{d}}\) be a set. 
    If \(y\) contains \(S\), then there must exist a cluster cover \(R\) of \(S\) (see~\cref{lemma:cluster-counting}), \(R\subseteq x\sqcup y\), such that \(|R|\ge 2|S|\) and for every cluster \(C\) of \(R\), we have \(|x\cap C|\ge |C|/2\).
\end{lemma}
\begin{proof}
    Consider the set of vertices \(x \sqcup y\).
    Let \(x' \sqcup y'\) denote the restriction of \(x \sqcup y\) to maximal connected clusters in \(\adjgraph{\Hext}\) containing vertices in \(S\), with \(x'\subseteq x, y'\subseteq y\). Note that \(S\subseteq y' \subseteq y\) and \(R = x' \sqcup y'\) is a cluster cover of \(S\) in \(\adjgraph{\Hext}\). 
    By~\cref{prop:final-round-syndrome}, 
    for every cluster \(C\) in \(R\),
    we have \(H_Z(y'\vert_C) = H_Z\phi(x'\vert_C)\). 
    Moreover, since \(y\) is the minimum weight vector with \(H_Zy = H_Z\phi(x)\), we have that \((y'\vert_C)\) is the minimum weight vector with \(H_Z(y'\vert_C) = H_Z\phi(x'\vert_C)\).
    Therefore 
    \begin{align}
        |x\cap C| = |x'\vert_C|\ge |\phi(x'\vert_C)|\ge |y'\vert_C| = |C\cap S|,
    \end{align}
    which implies that \(2|x\cap C| \ge |x'\vert_C| + |y'\vert_C| = |C|\). Our claim follows.
\end{proof}

\begin{lemma}\label{prop:residue-logical-error}
    For a bitstring \(x \in \F^{T n + Tr_Z}\) with trivial spacetime syndrome \(Hx = 0\), 
    let \(y\in \F^{n}\) be a minimum weight vector (supported on \(J_{T+1}^{\mathrm{d}}\)) such that \(H_Z\phi(x) = H_Zy\). 
    If \(\phi(x) - y \notin \rowspan (H_X)\), then there must exists a connected cluster of vertices \(S\) of size at least \(d\) in \(\adjgraph{\Hext}\) such that \(|x\cap S|\ge |S|/2\).
\end{lemma}
\begin{proof}
    Since  \(\phi(x) - y \notin \rowspan (H_X)\), it must be the case that \(\phi(x) - y\) contains a connected cluster of size at least \(d\) in \(\adjgraph{H_Z}\), the adjacency matrix of the base code \(H_Z\). 
    This implies that \(x\sqcup y\) contains a connected cluster \(S\) of size at least \(d\) in \(\adjgraph{\Hext}\). 
    Restricting to this cluster \(S\), we see that \(H_Zy\vert_{S} = H_Z\phi(x\vert_{S})\), and \(y\vert_S\) is the minimum weight vector satisfying this equality. Therefore we have 
    \begin{align}
        |x\cap S| = |x\vert_{S}|\ge |\phi(x\vert_{S})|\ge |y\vert_{S}|.
    \end{align}
    This implies that \(|x\cap S|\ge |S|/2\), as desired. 
\end{proof}

We are now ready to prove when a correction procedure for the spacetime code succeeds.

\begin{definition}[Decoding]\label{def:min-weight-decoding}
    Consider the extended spacetime decoding graph \(\adjgraph{\Hext}\).
    Let \(\ein\in \F^{|\Jd_1|}\) be an input error, and \(e\in \F^{|\Jbulk|}\) be an error on the bulk. 
    We denote \(\tilde{e} = \ein\sqcup e\).
    Let \(c\in \F^{|\Jbulk|}\) be the minimum weight vector such that \(H(e) = H(c)\), and let \(\cin \in \F^{n}\cong \F^{|\Jd_1|}\) be the minimum weight vector such that
    \begin{align}
        H_Z(\cin) = H_Z(\ein) + (e+c)\vert_{\Js_1}.
    \end{align}
    Denote \(\tilde{c} = \cin\sqcup c\), then \(\tilde{c}\) is the output correction of this minimum-weight decoding procedure, with \(H(\tilde{e}+\tilde{c}) = 0\).
    Let \(y\in \F^{|\Jd_{T+1}|}\) be the minimum weight vector that satisfies 
    \begin{align}
        H_Z\phi(\tilde{e} + \tilde{c}) = H_Zy,
    \end{align}
    we say that \(y\) is the \textbf{residue error} on the code block after decoding.
\end{definition}

\begin{lemma}[Successful Decoding]\label{lemma:spacetime-code-faults}
    Recall the code type and bad error supports defined in~\cref{def:qldpc-codetype}.
    Following the definitions and notations in \cref{def:min-weight-decoding}, let \(\reserror\subseteq \power(\Jd_{T+1})\) be a collection of bad error supports for the residue error.
    There exists a family of bad fault paths 
c    \(\Fst\subseteq \power(\Jbulk)\), which depends on \(\reserror\), such that the following conditions hold.
    \begin{itemize}[topsep = 0pt]
        \item \textbf{Bounded residue error. }
        If \(e\) is \(\Fst\)-avoiding, then \(y\) is \(\reserror\)-avoiding.
    
        \item \textbf{No logical error on good input. } If additionally \(\ein\) is \(\baderrors^{(Z)}_{1/2}\)-avoiding, then \(\phi(\tilde{e} + \tilde{c}) + y \in \rowspan (H_X)\).
    \end{itemize}
    Additionally, \(\mathcal{F}_{\mathsf{spacetime}}\) has weight enumerator
    \begin{align}
        \weightenum{\mathcal{F}_{\mathsf{spacetime}}}{x} \le |J|\cdot [O_{\Delta}(x)]^{d/40} + \weightenum{\reserror}{O_{\Delta}(x^{1/2})}~.
    \end{align}
\end{lemma}
\begin{proof}
We construct \(\Fst\) out of three collection of bad fault paths. 
Let \(\adjgraph{\Hbulk}\) denote the induced subgraph of \(\adjgraph{\Hext}\) with vertices \(\Jbulk\).
The first collection, \(\badfaults_1\), are error supports that have significant overlap with any connected cluster of size at least \(d\) in \(\adjgraph{\Hbulk}\).
Recall that \(\clusterset_{(m,t)}^{(\Hbulk)}\) are the \((m,t)\) clustering sets in \(\adjgraph{\Hbulk}\). Define 
\begin{align}
    \badfaults_1 = \boxplus_{m \ge d} \clusterset_{(m,\lceil m/10\rceil)}^{(\Hbulk)}.
\end{align}
Now suppose \(e\) is a \(\badfaults_1\)-avoiding fault, and let us consider the connected components of \(\tilde{e}+\tilde{c}\). 
There are three types of components:
\begin{enumerate}[leftmargin=*, align = left, topsep = 0pt]
    \item[\textbf{Type 1.}]
    \textbf{Input components} that contain vertices in \(\Jd_1\) and do not contain vertices in \(\Js_T\).
    
    \item[\textbf{Type 2.}]
    \textbf{Internal components} that  do not contain vertices in \(\Jd_1\) and \(\Js_T\).

    \item[\textbf{Type 3.}]
    \textbf{Residue components} that do not contain vertices in \(\Jd_1\) and contain vertices in \(\Js_T\).
\end{enumerate}
In particular, there are no components that contain vertices in both \(\Jd_1\) and \(\Js_T\).
To see this, let \(c\) be the restriction of \(\tilde{c}\) onto \(\Jbulk\). If \(\tilde{e}+\tilde{c}\) contains a connected component that traverses from \(\Jd_1\) to \(\Js_T\), then \(e + c\) must contain a connected component \(C\) that traverses from \(\Js_1\) to \(\Js_T\), which means \(|C|\ge d\).
Since \(c\) is the minimum weight correction to \(e\), we must have \(|e\cap C|\ge |C|/2\ge d/2\), which is ruled out by \(\badfaults_1\).

Observe that the same argument implies that all internal components must have size less than \(d\). 
By~\cref{prop:successful_correction}, these components are flattened to stabilizers. They induces no logical error and no residue syndrome. We can safely ignore them for the rest of this proof.

To bound the residue error \(y\) and therefore show the first condition, we can safely ignore input components and consider residue components. 
Let \(B\) be a set in \(\reserror\).
From~\cref{prop:residue-error-support}, if \(y\) contains \(B\), then there must be a cluster cover \(R\) of \(B\) in the extended graph \(\adjgraph{\Hext}\) such that 
\begin{enumerate}[topsep = 0pt]
    \item \(R\subseteq (e+c)\sqcup y\), and \(|R|\ge 2|B|\);
    \item For every cluster \(C\) of \(R\), we have \(|(e+c)\cap C|\ge |C|/2\). Note that \(C\) must contain some vertex of \(B\), which is in \(\Jd_{T+1}\).
\end{enumerate}
Taking the union over disjoint connected clusters \(C\), we have
\begin{align}
    (e+c)\cap R \ge |R|/2.
\end{align}
Due to minimality of \(c\), we have \(|e\cap R|\ge |R|/4\). 
For a set \(B\subseteq \Jd_{T+1}\), recall that \(\cc(r, B)\) denote the collection of cluster covers of \(B\) of size at least \(r\) in \(\adjgraph{\Hext}\).
Define
\begin{align}
    \badfaults_2 = \boxplus_{B\subseteq \Jd_{T+1},B\in \reserror} \boxplus_{r\ge 2|B|} \boxplus_{R\in \cc(r, B)} 
    \{T\subseteq R: |T| = |R|/4, \text{ and } T\subseteq \Jbulk\}.
\end{align}
We see that if \(y\) is not \(\reserror\)-avoiding, then \(e\) must contain a set from \(\badfaults_2\).
In other words, if \(e\) is \((\badfaults_1\boxplus \badfaults_2)\)-avoiding, then \(y\) is \(\reserror\)-avoiding, as desired.

Next, we want to rule out the case that \(\phi(\tilde{e}+\tilde{c}) + y \notin \rowspan(H_X)\). 
As discussed earlier, internal components do not induce logical errors when flattened. 
For residue components, by~\cref{prop:residue-logical-error}, 
there must exists a connected cluster \(C\subset (e+c)\sqcup y\) of size at least \(d\) in \(\adjgraph{\Hext}\) (\(C\) needs to contain vertices in \(\Jd_{T+1}\)) such that \(|(e+c)\cap C| \ge |C|/2\).
We have 
\begin{align}
   |(e+c)\cap C| \ge |C|/2 \Rightarrow |e\cap C|\ge |C|/4\ge d/4.
\end{align}
This is forbidden by the fact that \(e\) is \(\badfaults_1\)-avoiding.\footnote{This is not directly implied by the definition of \(\badfaults_1\), since it is defined over connected components in the bulk graph, while \(C\) is a connected component in the extended graph. 
To prove our claim, we construct a connected cluster \(C'\) in the bulk graph.
Let \(z\in \F^{n}\) be the restriction of \(C\) to \(\Jd_{T+1}\). 
Let \(z'\) denote the same set (equivalently vector) of vertices as \(z\), but supported over \(\Jd_{T}\). 
Let \(C' = C\vert_{\Jbulk}\cup z'\) (note that we used \(\cup\) instead of \(\sqcup\) here). 
Then \(C'\) is a connected component in the bulk graph of size at least \(d\), and \(|e\cap C'| = |e\cap C|\ge d/4\).
This cannot happen since \(e\) is \(\badfaults_1\)-avoiding. 
}
Therefore residue components cannot induce logical errors when flattened.

Finally, we consider input components.
Let \(\mathcal{T}\) be the union of all input components. 
Recall that in our decoding process (\cref{def:min-weight-decoding}), we first computed \(c\) as a correction to \(e\) and then computed \(\cin\) as a correction to \(\ein\).
Therefore \(\mathcal{T}\cap \Jbulk = \mathcal{T}\cap (e+c)\), which has no spacetime syndrome under \(\Hbulk\). 
By~\cref{cor:bulk-syndrome}, we have 
\begin{align}
    H_Z\phi_\bulk(\mathcal{T}\cap (e+c)) = (e+c)\vert_{\Js_1}.
\end{align}
Let \(\fin\in \F^n\) be the minimum weight vector such that \(H_Z\fin = (e+c)\vert_{\Js_1}\), then \(\cin\) is the minimum weight vector such that \(H_Z\cin = H_Z(\ein + \fin)\).
We can now repeat our arguments above, as \(\fin\) is canonical to \(y\). 
Since \(e\) is \(\badfaults_1\)-avoiding, we have that 
\begin{align}
    \phi_\bulk(\mathcal{T}\cap (e+c)) + \fin \in \rowspan(H_X).
\end{align}
Moreover, define
\begin{align}
    \badfaults_3 = \boxplus_{B\subseteq \Jd_{1},B\in \baderrors^{(Z)}_{1/10}} \boxplus_{r\ge 2|B|} \boxplus_{R\in \cc(r, B)} 
    \{T\subseteq R: |T| = |R|/4, \text{ and } T\subseteq \Jbulk\}.
\end{align}
If \(e\) is \(\badfaults_3\)-avoiding, we know that 
\(\fin\) is \(\baderrors^{(Z)}_{1/10}\)-avoiding.

Now suppose \(\ein\) is \(\baderrors^{(Z)}_{1/2}\)-avoiding. 
By~\cref{lemma:separable-bad-errors}, \(\ein+\fin\) is \(\baderrors^{(Z)}_{6/10}\)-avoiding, which implies that \(\ein+\fin+\cin \in \rowspan(H_X)\). 
Therefore, 
\begin{align}
    \phi(\mathcal{T}) = \ein + \cin + \phi_\bulk(\mathcal{T}\cap (e+c)) = (\ein+\cin+\fin) + (\fin + \phi_\bulk(\mathcal{T}\cap (e+c))) \in \rowspan(H_X).
\end{align}
In other words, input components cannot induce logical errors.

Putting things together, we define 
\begin{align}
    \Fst = \badfaults_1 \boxplus \badfaults_2 \boxplus \badfaults_3,
\end{align}
and we have argued that if \(e\) is \(\Fst\)-avoiding, then both conditions hold. 
It remains for us to bound the weight enumerator of \(\Fst\).
Let \(\Delta_H\) denote the maximum degree of \(\adjgraph{\Hext}\), we know that \(\Delta_H = O(\Delta)\).
By the same calculation as in \cref{lemma:weight-enum-baderrors-c}, we have 
\begin{align}\label{eq:}
    \weightenum{\badfaults_1}{x} \le |J|\cdot [O_{\Delta}(x)]^{d/10}.
\end{align}
To bound \(\badfaults_2\), we consider bad error supports in \(\reserror\). 
Invoking \cref{lemma:cluster-counting},
\begin{align}
    \weightenum{\badfaults_2}{x}
    &\le \sum_{B\in \reserror}\sum_{r\ge 2|B|}\sum_{R\in \cc(r, B)}2^{|R|}x^{|R|/4} \\
    &\le \sum_{B\in \reserror}\sum_{r\ge 2|B|}(\Delta_H e)^r 2^{r}x^{r/4} \\
    &\le \sum_{B\in \reserror}[O_{\Delta}(x)]^{|B|/2} \\
    &\le \weightenum{\reserror}{O_{\Delta}(x^{1/2})}~.
\end{align}

To bound \(\badfaults_3\), we need to consider \(\baderrors^{(Z)}_{1/10}\).
\begin{align}
    \weightenum{\badfaults_3}{x} 
    &\le 
    \sum_{B\in \baderrors^{(Z)}_{1/10}} \sum_{r\ge 2|B|}\sum_{R\in \cc(r, B)} 2^{r}x^{r/4} \\
    &\le \sum_{B\in \baderrors^{(Z)}_{1/10}}  \sum_{r\ge 2|B|} (\Delta_H e)^r 2^{r}x^{r/4} \\
    &\le \weightenum{\baderrors^{(Z)}_{1/10}}{O_{\Delta}(x^{1/2})} \\
    &\le |J|\cdot [O_{\Delta}(x)]^{d/40}.
\end{align}
Therefore, putting everything together we see that
\begin{align}
    \weightenum{\Fst}{x} 
    &\le |J|\cdot [O_{\Delta}(x)]^{d/40} + \weightenum{\reserror}{O_{\Delta}(x^{1/2})}. 
    \qedhere
\end{align}

\end{proof}

\subsection{Error correction gadget}

\begin{theorem}[Error correction gadget \cite{dennis2002topological,gottesman2014fault}]\label{lemma:qldpc:ec-gadget}
    We can construct a gadget $\gadget_{\mathsf{EC}} = (V,E,\edgeConn, \edgeType, \nodeGate)$ 
    of depth \(O(d)\) for the computational code with the following guarantee.
    \begin{enumerate}[topsep = 0pt]
        \item For input code type \(\compCodeType[1/2] = (U, \baderrors_{1/2})\) and any output code type \((U, \reserror)\), let \(\reserror_1, \reserror_2\) be two collections\footnote{Here we seperate the output error \(\reserror\) into two collections to account for two sources of residue errors: one from the decoder ouput as in \cref{lemma:spacetime-code-faults}, and another from the last round of syndrome extraction circuit.}
        such that \(\reserror\) is separable into \(\reserror_1, \reserror_2\) as defined in \cref{def:separable-bad-sets}. 
        There is a family of bad fault paths \(\mathcal{F}_{\mathsf{EC}}\) such that \(\gadget_{\mathsf{EC}}\) is a fault-tolerant gadget for identity with respect to the input/output code types and \(\mathcal{F}_{\mathsf{EC}}\).

        \item This family of bad fault paths \(\mathcal{F}_{\mathsf{EC}}\) has weight enumerator
        \begin{align}\label{eq:qldpc-ec-fault-weight-enumerator}
            \weightenum{\mathcal{F}_\mathsf{EC}}{x} &\le d \cdot n \left( O_{\Delta}(x)\right)^{\frac{d}{80\Delta}} + \weightenum{\reserror_1}{O_{\Delta}(x^{\frac{1}{4\Delta}})} + \weightenum{\reserror_2}{O_{\Delta}(x^{\frac{1}{2\Delta}})}~.
        \end{align}
    \end{enumerate}
\end{theorem}
\nomenclature[F, 00]{\(\gadget_{\mathsf{EC}}\)}{The error correction gadget defined by performing \(d\) rounds of syndrome extraction, see \cref{lemma:qldpc:ec-gadget}.}
\begin{construction}
    \textbf{Syndrome extraction circuit.}
    First, we construct a constant-depth syndrome extraction circuit for the code in the standard way.
    A \(\Delta\)-qLDPC code defined by the check matrices \((H_X,H_Z)\) has a depth \(\Delta+2\) syndrome extraction circuit each for the \(X\) syndrome and the \(Z\) syndrome.
    Let us consider the case of measuring the \(Z\) syndrome (\(X\) type operators associated with \(H_X\)).
    The case of the \(X\) syndrome is identical with \(\CNOT\) replaced by \(\CZ\) and \(H_X\) replaced by \(H_Z\).

    First, compute an edge coloring of the Tanner graph of \(H_Z \in \F^{r \times n}\):
    Define the bipartite graph \(G = ([r] \sqcup [n], E)\) where two vertices \((i,j) \in [r] \times [n]\) are adjacent if \(H[i,j] = 1\).
    \(G\) has maximum degree \(\Delta\) by the \(\Delta\)-qLDPC property.
    The edge coloring of \(G\) is a map \(\phi \colon E \mapsto [\Delta]\) such that no vertex has two edges of the same color incident to it.
    Equivalently, for each color \(c \in [\Delta]\), the subset of edges \(\phi^{-1}(m)\) is a matching in \(G\).
    \(G\) is bipartite so such a map exists and is efficiently computable~\cite{schrijver2003combinatorial}.

    A single round of syndrome extraction is as follows:
    \begin{enumerate}
        \item Prepare \(r\) ancilla qubits in \(\ket{+}\).
        \item For each color \(c \in [\Delta]\), perform the gate layer:
        \begin{itemize}
            \item For each edge \((i,j) \in \phi^{-1}(c)\), apply \(\CNOT(i,j)\) (controlled on the ancilla qubit \(i\))
        \end{itemize}
        \item Measure all \(r\) ancilla qubits in the \(X\) basis. The \(i\)-th measurement outcome is the eigenvalue of the \(X\) operator associated with the \(i\)-th row of the check matrix \(H_X\).
    \end{enumerate}
    Let us refer to a single round of syndrome \(Z\) and \(X\)-type syndrome extraction as the circuit \(M\).
    The full error correction gadget alternates layers of \(Z\) and \(X\)-type syndrome extraction rounds \(T\) (to be picked later) rounds.

    That is, the gadget applies \(M\) for \(T\) times where each application of \(M\) should be interpreted as producing an \(Z\) and \(X\) measurement vector \(m_1, \dots, m_T \in \F^{r_X + r_Z}\).
\end{construction}

\begin{proof}
    From these measurments, we will compute a correction and apply it.
    
    \textbf{Decoding.} 
    Let \(r_X\) be the number of rows of \(H_X\) and \(r_Z\) be the number of rows of \(H_Z\).  
    The syndrome extraction circuit produces a pair of measurements \(m_X \in \F^{r_X \times T}\) and \(m_Z \in \F^{r_Z \times T}\).
    We take pairwise differences to obtain syndromes \(s_X\) and \(s_Z\).
    \begin{align}\label{eq:spacetime-syndrome-diffs}
        s_{X/Z}[I_1] &= m_{X/Z}[I_1]  \\
        t \in [T-1], \quad s_{X/Z}[I_{t+1}] &= m_{X/Z}[I_t] + m_{X/Z}[I_{t+1}]
    \end{align}
    
    In what follows, to reduce notation, we will not distinguish between diagonal Pauli superoperators and Pauli operators.
    Additionally, for a bit string \(x \in \F^m\), we will also use \(x\) to denote the state \(\ketbra{x}{x}\) when clear from context.
    For now, suppose that the fault \(\fault\) is a diagonal Pauli fault and the input state \(\rho_{\mathrm{0}}\) is \(\baderrors_c\)-deviated from a state \(\encoder(\sigma)\) by a diagonal Pauli superoperator \(E_0\)
    \begin{align}
        \circuitMap{M^{\otimes t}[\fault]} (\rho_0) = \rho_t\otimes m_X[I_1] \otimes m_Z[I_1] \otimes \dots \otimes m_X[I_t] \otimes m_Z[I_t]
    \end{align}
    be the state (including measurement outputs) after the \(t\)-th round of \(Z\) and \(X\) syndrome extraction, and \(E_t\) be the diagonal Pauli superoperator supported on the code block that takes \(\rho_{t-1}\) to \(\rho_{t}\), \(E_t(\rho_{t-1}) = \rho_t\).
    Let us write \(\xi\) for the syndrome map \(\mathcal{P}^n \to \F^{r_X + r_Z}\).
    For \(t \in [T]\), define \(\tilde{m}_t = m_t - \xi(E_{t-1} \dots E_0)\) to be the difference between the measured syndrome and syndrome of the state input to the \(t\)-th application of \(M\), and decompose \(E_t = E_t^{(X)} E_t^{(Z)}\) into the product of Pauli \(X\) and Pauli \(Z\).
    Using that, in the absence of any further faults, the measurement output is the syndrome of the state and that \(\xi\) is a group homomorphism, \(\xi(AB) = \xi(A)+\xi(B)\), we can see that the syndrome differences (\cref{eq:spacetime-syndrome-diffs} give the syndrome of any new errors applied up to measurement errors.
    \begin{align}
        s_{X/Z}[I_1] &= \xi(E_0^{(Z/X)}) + \tilde{m}_{X/Z}[I_1] \\
        t \in [T-1], \quad s_{X/Z}[I_{t+1}] &= \xi(E_{t-1}^{(Z/X)}) + \tilde{m}_{X/Z}[I_t] + \tilde{m}_{X/Z}[I_{t+1}]
    \end{align}

    For convenience, we will decode these syndromes separately.
    We focus on decoding \(X\) errors with syndrome \(s_Z\). 
    The case of \(Z\) errors is identical.
    
    We now decode the syndrome \(s_Z\) with respect to the spacetime code \(H\) of \(H_Z\), as specified by the decoding procedure in \cref{def:min-weight-decoding}, to produce a correction \(c_X \in \F^{(T n + Tr_z)}\).
    We apply an \(X\) correction to the \(i\)-th qubit when the \(i\)-th bit of the flattening \(\phi\) of the correction \(c_X\) is \(1\), \(C_X =X^{\phi(c_X)}\).

    The overall gadget \(\gadget_\mathsf{EC}\) is then given in figure \cref{fig:ec-gadget}. (Runtime of the decoder ignored.)
    \begin{figure}[h]
        \centering
        \includegraphics[width=0.7\linewidth]{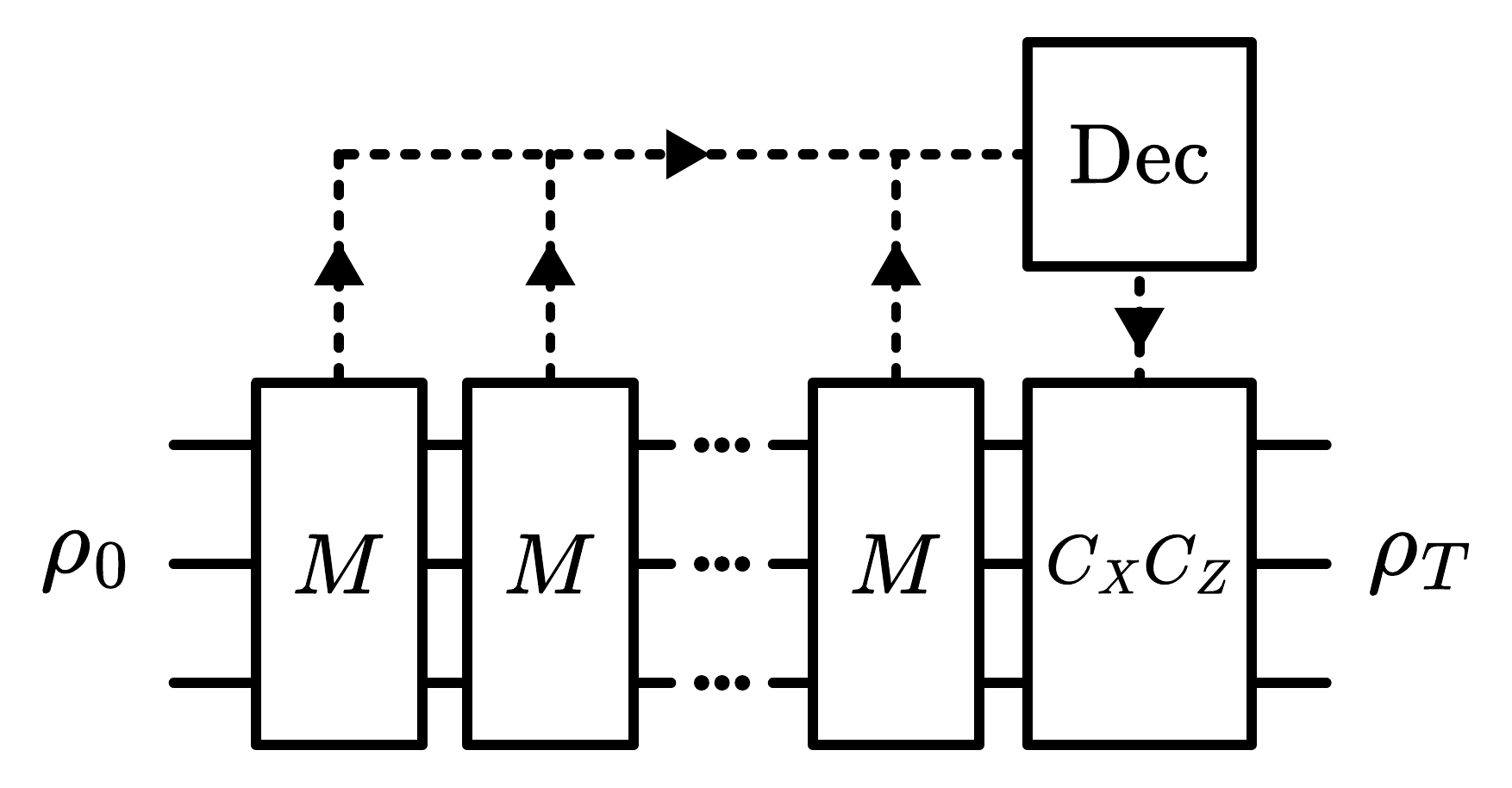}
        \caption{Error correction gadget \(\gadget_\mathsf{EC}\), the circuit \(M\) is repeated \(T\) times.}
        \label{fig:ec-gadget}
    \end{figure}

    \textbf{Fault analysis.} 
    We would now like to analyze the residual error.
    We will set \(T = d\).
    First note that there is an error \(e \in \F^{(T n + Tr_z)}\) of the spacetime code corresponding to the error that in the circuit that yields the measured syndrome differences.
    \begin{align}\label{eq:ec-gadget:spacetime-code-error}
        &t \in [T], \quad E^{(X)}_{t-1} = X^{e[J_t^{\mathrm{d}}]} \\
        &t \in [T], \quad \tilde{m}_{Z}[I_t] = e[J_t^{\mathrm{s}}]
    \end{align}
    With the exception of the final round error, the flattening also matches the residual error: \(X^{\phi(e)} =  E^{(X)}_{T-1} \dots E_0^{(X)}\).
    We now utilize the obvious bijection between the supports of \(E^{(X)}_{T-1}, \dots, E^{(X)}_{0}\), \(\tilde{m}_{Z}\), and bits of the spacetime code \([T n + Tr_z]\).
    For a bit \(i \in [T n + Tr_z]\), let \(\chi_i \subseteq V\) denote the set of locations of \(\gadget_\mathsf{EC}\) for which a fault may flip bit \(i\).
    A fault supported on a single location may flip at most \(s = 2 \Delta\) bits of the spacetime code: A two-qubit gate failure affects a data and ancilla qubit. and each of these has a gate shared with \(\Delta-1\) other qubits.\footnote{These are the well known ``hook errors.'' A more careful counting would should that the majority of the spread of the error does not intersection any of the sets of interest in so many places.}
    Likewise, a single bit of the spacetime code may be flipped by a fault at one of at most \(\Delta(2\Delta + 4)\) locations.
    
    Define a bipartite graph \(G = (V \sqcup [Tn + Tr], E)\) between the locations of the gadget and the spacetime code bits.
    Bit \(i\) is connected to spacetime location \(v\) if \(v \in \chi_i\).
    \(G\) has left degree at most \(2\Delta\) and right degree at most \(\Delta (2\Delta + 4)\).
    We can now define a family of bad fault in the gadget corresponding to the bad spacetime code errors.
    Let \(\mathcal{F}_{\mathsf{spacetime}}^{(X)}\) be the family of bad error supports in the statement of \cref{lemma:spacetime-code-faults}, with the guarantee that the residue error is \(\reserror_1\)-avoiding.
    For a subset of vertices \(S \subseteq [Tn + Tr]\), let \(N(S)\) denote the neighborhood of \(S\) in \(G\).
    If a location is capable of flipping a particular spacetime code bit then the two must be adjacent in \(G\).
    \begin{align}
        \mathcal{F}_{\mathsf{EC},1}^{(X)} = \boxplus_{f \in \mathcal{F}_{\mathsf{spacetime}}^{(X)}} \{f' \subseteq N(f) \mid \forall u \in f, \hspace{1ex} |N(u) \cap f'| > 0\}
    \end{align}
    There are at most \(2^{|f| \Delta(2\Delta+4)}\) subsets of \(N(f)\) and each subset must have weight at least \(|f|/(2\Delta)\).
    Thus for \(x \in [0,1]\), we can upper bound the weight enumerator as
    \begin{align}
        \weightenum{\mathcal{F}_{\mathsf{EC},1}^{(X)}}{x} \le \weightenum{\mathcal{F}_{\mathsf{spacetime}}^{(X)}}{2^{\Delta(2\Delta + 4)} x^{\frac{1}{2\Delta}}}
    \end{align}
    We repeat this construction with \(X\) and \(Z\) interchanged to obtain \(\mathcal{F}_{\mathsf{EC},1}^{(Z)}\).

    We also need to bound the error caused by the last application of \(M\) that propagates to \(E_T\).
    We can perform an identical argument with \(\mathcal{F}_{\mathsf{spacetime}}^{(X)}\) replaced by \(\reserror_2\) and bits of the spacetime code replaced by \(E_T\).
    \begin{align}
        \weightenum{\mathcal{F}_{\mathsf{EC},2}}{x} \le \weightenum{\reserror_2}{2^{\Delta(2\Delta + 4)} x^{\frac{1}{2\Delta}}}
    \end{align}

    Our bad fault paths are
    \begin{align}
        \mathcal{F}_{\mathsf{EC}} = \mathcal{F}_{\mathsf{EC},1}^{(X)} \boxplus \mathcal{F}_{\mathsf{EC},1}^{(Z)} \boxplus \mathcal{F}_{\mathsf{EC},2}~.
    \end{align}

    Now, let \(\fault\) be a \(\mathcal{F}_{\mathsf{EC}}\)-avoiding diagonal Pauli fault.
    In the execution of \(M[\fault]\) on a state \(\rho_0\), the spacetime code error \(e\) (\cref{eq:ec-gadget:spacetime-code-error}) is \(\mathcal{F}_{\mathsf{EC},1}^{(X)}\)-avoiding in the \(Z\)-type syndrome extraction and \(\mathcal{F}_{\mathsf{EC},1}^{(Z)}\)-avoiding in the \(X\)-type syndrome extraction. 
    In the case where \(E_0\) is \(\baderrors_{1/2}\)-avoiding, it follows by \cref{lemma:spacetime-code-faults} that \(C_X C_Z E_{T-1}\dots E_0\) is \(\reserror_1\)-avoiding.
    Otherwise, there exists a logical operator of the code \(L\) such that \(LC_X C_Z E_{T-1}\dots E_0\) is \(\reserror_1\)-avoiding (friendly).

    The final round of errors \(E_T\) has \(\reserror_2\)-avoiding support,
    by the separability of \(\reserror\) as \(\reserror_1\) and \(\reserror_2\), the output is \(\reserror\)-deviated from the codespace.
    In the case that the input is \(\baderrors_{1/2}\)-deviated from the encoding of a logical state \(\encoder(\sigma)\), then the output is \(\reserror\)-deviated from \(\encoder(\sigma)\).

    We now drop the assumption that the faults are diagonal: For a general input error and faults (non-diagonal Pauli) \cref{lemma:deterministic-errors} can be applied at each step, so that the errors at each step is either equivalent to a diagonal Pauli error (the left and right sides are stabilizer equivalent) up to some part in the lightcone of the fault in that round or the state is annihilated by the measurement projectors (so the deviation property trivially holds with coefficient \(0\)).

    We now plug in the bounds on the weight enumerators to obtain (recall \(T\) was picked to be \(d\) so there are at most \(\propto n d\) locations in the spacetime graph.)
    \begin{align}
        \weightenum{\mathcal{F}_\mathsf{EC}}{x} &\le d \cdot n [O_{\Delta}(x)]^{\frac{d}{80\Delta}} + \weightenum{\reserror_1}{O_{\Delta}(x^{\frac{1}{4\Delta}})} + \weightenum{\reserror_2}{O_{\Delta}(x^{\frac{1}{2\Delta}})}~. 
    \end{align}
    Following the definition of a friendly, fault-tolerant gadget \cref{def:gadget}, we see that this entire error correction gadget implements the filter superoperator \(\recover_{\fault}\) and an identity logical gate, \(\tilde{\g}_{\fault} = \IDENT\).
\end{proof}
\begin{remark}[Friendliness]\label{remark:ECgadget-friendliness}
    Through \cref{lemma:spacetime-code-faults}, we have argued that the error correction gadget is friendly. I.e., when the input error is bad, the residue error of the output state is still bounded, although the logical state is erroroneous. 
    By a closer inspection of the gadget construction and decoding, we note that the EC gadget is \textit{friendly in both the \(X\) and \(Z\) bases}. 
    More precisely, suppose the input error is a diagonal Pauli error \(P_XP_Z\). 
    If \(P_X\) has \(\baderrors_{1/5}\)-avoiding support while \(P_Z\) is arbitrarily supported, the logical error \(L\) will be a \(Z\) error. 
    Similarly, if \(P_Z\) has \(\baderrors_{1/5}\)-avoiding support while \(P_X\) is arbitrarily supported, the logical error \(L\) will be a \(X\) error. 
\end{remark}

In later constructions of fault-tolerance schemes, we will often construct a logical gate gadget by first applying this error correction gadget, followed by the corresponding logical operation. 
The constructed gadgets are therefore friendly and the error support on the input state to the logical operation can be easily bounded. 

\begin{remark}[Single-shot decoding]
    The error correction gadget constructed above repeatedly measure the sydnrome information for \(O(d)\) rounds, which protects the encoded information from measurement errors. 
    For most codes this repetition is necessary to control the residue error distribution.
    There are families of qLDPC codes which can be single-shot decoded~\cite{bombin2015gauge,Kubica2022single,bridgeman2024lifting,fawzi2020constant,gu2024single,nguyen2024quantum}. 
    More precisely, these codes have efficient decoders which only requires \(O(1)\) rounds of syndrome information to compute corrections with logical error rates $e^{-O(d)}$ and controlled residue errors.
    Using these codes as computation code could improves the depth overhead of fault-tolerant gadgets. 
    It is worth noting, however, that an error correction gadget constructed from a single-shot decoder may not be friendly. 
    For instance, the decoder of \cite{gu2024single}, which is a modified version of the decoder proposed by \cite{leverrier2022parallel}, assumes that the input qubit and measurement error weights are both upper bounded by \(cn\) for some constant \(c\) (the decoded code is asymptotically good). 
    For the bad input state where this assumption is not satisfied, the decoder makes no guarantee on the output state and residue error. 
    Conceptually this can be seen as a \textit{logical leakage error}, where not only is the logical information corrupted, the logical qubit itself is also lost and requires reinitialization (for instance, by a friendly gadget).
\end{remark}

\begin{remark}[Code switching]\label{remark:code-switching}
    The error correction gadget we constructed measures the same set of stabilizers repeatedly and implements a logical identity. 
    Alternatively, by measuring different groups of stabilizers corresponding to different codes, we can fault-tolerantly {switch} from one code (and code type) to another~\cite{bombin2009quantum}. 
    Such \textit{code-switching} protocols are highly versatile. 
    One motivating example is to switch between a two-dimensional topological code which admits transversal Clifford gates~\cite{calderbank1996good,steane1996multiple,Bombin2006topological} and a three-dimensional topological code which admits transversal non-Clifford gates~\cite{bombin2015gauge,kubica2015universal,Paetznick2013universal,Kubica2015unfolding}, see~\cite{bombin2016dimensional}. 
    This enables one to utilize two computational codes and implement a universal gate set transversally. 

    Another well-studied type of code-switching is \textit{code surgery}, which implements multi-qubit Pauli measurement on logical qubits. 
    Such multi-qubit Pauli measurement, when supplied with magic states, can be used to implement universal quantum computation in the framework of Pauli-based compuation (PBC)~\cite{bravyi2016trading}. 
    The prototypical example of code surgery is lattice surgery~\cite{Horsman2012LatticeSurgery,fowler2012surface} on surface codes, which can be applied to design fault-tolerant surface code architectures with two-dimensional nearnest neighbor connectivity~\cite{Fowler2018lattice,litinski2019game}. 
    These methods are later generalized to other qLDPC codes~\cite{hastings2016weight,hastings2021weight,cohen2022low} under the name \textit{qLDPC surgery}. 
    While the first proposals of qLDPC surgery~\cite{cohen2022low} incurs a heavy space overhead, recent works~\cite{cowtan2024ssip,cross2024improved,williamson2024low,ide2024fault,swaroop2024universal,zhang2024time} presented improved proposals with significantly lower theoretical and practical overheads.
    These methods are subsequently improved and applied in~\cite{cowtan2025parallel,he2025extractors,yoder2025tour}. 
    In particular, the extractor architectures proposed in~\cite{he2025extractors} enables fault-tolerant PBC on arbitrary qLDPC codes with low-overhead and fixed connectivity. 
    It would be interesting to prove the fault-tolerance of these surgery gadgets in our weight enumerator formalism.
    
    The aforementioned examples of code-switching can all be described using the langauge of subsystem codes and gauge-fixing~\cite{vuillot2019deformation}. 
    In contrast, the dynamical codes proposed in~\cite{Hastings2021dynamically} demonstrated that logical information can be stored and operated on in a continuous sequence of code-switching~\cite{davydova2024quantum}.
    It would be interesting to describe the fault-tolerance properties of these dynamically generated logical qubits in terms of weight enumerators as well. 
\end{remark}

\subsection{Constant depth gadgets}
Gate gadgets for transversal gates are a short corollary.
\begin{corollary}[Transversal gate gadget]\label{lemma:qldpc:transversal-gates}
    Let \(\g\) be a gate on either one or two computational code blocks that can be implemented transversally.
    Then, there exists a gadget for \(\g\) with input and output code types \(\compCodeType[1/5]\) and a family of bad fault paths \(\mathcal{F}_\g\) where
    \begin{align}
        \weightenum{\mathcal{F}_\g}{x} \le 2 \weightenum{\mathcal{F}_{\mathsf{EC}}}{x} +\weightenum{\baderrors_{1/10}}{x} \le d\cdot n\cdot [O_{\Delta}(x)]^{\frac{d}{80\Delta}}. 
    \end{align}
\end{corollary}
\begin{proof}
    We consider the more complicated case of two code blocks.
    
    We execute the error correction gadget \cref{lemma:qldpc:ec-gadget} on the two inputs block,
    setting \(\reserror = \baderrors_{1/5}\) separable into two copies of \(\baderrors_{1/10}\). 
    From \cref{lemma:weight-enum-baderrors-c}, we know
    \begin{align}
        \weightenum{\baderrors_{1/10}}{x}
        &\le n\cdot [O_{\Delta}(x)]^{\frac{d}{20}}.
    \end{align}
    Combined with \cref{eq:qldpc-ec-fault-weight-enumerator}, we have an upper bound on the weight enumerator of \(\mathcal{F}_{\mathsf{EC}}\). 
    \begin{align}
        \weightenum{\mathcal{F}_{\mathsf{EC}}}{x}\le d\cdot n\cdot [O_{\Delta}(x)]^{\frac{d}{80\Delta}}. 
    \end{align}
    We pick \(\mathcal{F}_\g\) to be the sum of \(\mathcal{F}_{\mathsf{EC}}\) on the two error correction gadgets and \(\baderrors_{1/10}\) on the transversal gates. 
    The upper bound in our claim follows.
    
    Assume the fault is \(\mathcal{F}_\g\)-avoiding. Then, if the input state \(\rho\) is \(\baderrors_{1/2}\boxplus \baderrors_{1/2}\)-deviated from a logical state \((\encoder\otimes\encoder)(\sigma)\), then the output of the error correction gadgets is  \(\baderrors_{1/5}\boxplus \baderrors_{1/5}\)-deviated from \((\encoder\otimes\encoder)(\sigma)\).
    The transversal gate gadget can at most create a \(\baderrors_{1/10}\)-avoiding error and spread a \(\baderrors_{1/5}\)-avoiding error to each block.
    It follows by \cref{lemma:separable-bad-errors} that the output state is \(\baderrors_{1/2}\boxplus \baderrors_{1/2}\)-deviated from \((\encoder\otimes\encoder)\circ \g(\sigma)\).

    For an arbitrary input state, the output of the error correction gadgets is still  \(\baderrors_{1/5}\boxplus \baderrors_{1/5}\)-deviated from the code space.
\end{proof}

Note that we precede all transversal gates with a \(O(d)\)-round error correction gadget, which incurs a notable time (equivalently depth) overhead in theoretical and practical fault-tolerance. 
Alternatively, we could interleave transversal gates with syndrome extraction rounds to amortize the time overhead of a full error correction gadgets across multiple logical operations. 
This approach forms the backbone of \textit{transversal algorithmic fault-tolerance} as proposed and studied in \cite{zhou2024algorithmic}. 
It is an interesting future direction to formulate this approach in terms of gadgets and weight enumerators. 

Beyond transversal gates, we further consider constant depth gates supported on one computational code block, and show the standard claim that constant depth gates are fault tolerant via a lightcone argument. 

\begin{lemma}[Constant depth gate gadget]\label{lemma:qldpc:constant-depth-gadget}
    Let \(\g\) be a gate on one computational code block that can be implemented by a circuit \(C_\g\) of depth \(\depth\). 
    Set \(c = \frac{1}{5(T+1)2^{T+1}}\).
    Then there exists a gadget for \(\g\) with input and output code type \(\compCodeType[1/5]\) and a family of bad fault paths \(\mathcal{F}_\g\) where
    \begin{align}
        \weightenum{\mathcal{F}_\g}{x} \le 
        dn\cdot [O_\Delta(x)]^{\frac{d}{80\Delta}} + n\cdot [O_{\Delta,T}(x)]^{\frac{cd}{16\Delta}} + Tn\cdot [O_{\Delta,T}(x)]^{\frac{cd}{2}}.
    \end{align}
\end{lemma}
\begin{proof}
    We will first apply the error correction gadget \cref{lemma:qldpc:ec-gadget} on the code block, with \(\reserror\) to be specified later. 
    The residue error being \(\reserror\)-avoiding and the fault path on \(C_\g\) being \(\badfaults_\g\)-avoiding will guarantee that the final output error is \(\baderrors_{1/5}\)-avoiding. 

    We use a lightcone argument to define \(\reserror\) and \(\badfaults_\g\). 
    Since \(C_\g= (V,E,\edgeConn, \edgeType, \nodeGate)\) acts on the code qubits labelled by \([n]\) and has depth \(T\), we partition the edges in \(C_\g\) by depth as \(E = \cup_{t=1}^{T+1}E_t\), where \(E_1, E_{T+1}\) are the input and output qubits respectively and each \(E_t\) is a copy of \([n]\).
    We denote the qubits in \(E_t\) as \(e_i^t\), for \(i\in [n]\), with the implicit identification that \(\{e_i^t\}_{t=1}^{T+1}\) denote the same physical qubit at different time.
    For a qubit \(e_i^t\in E_t\), the \textit{forward lightcone} of \(e_i^t\), denoted \(\lightcone(e_i^t)\), are all qubits \(e'\in E_{T+1}\) such that there is a directed path from \(e\) to \(e'\) in the directed network of \(C_\g\).
    Observe that for a qubit \(e_i\) and two time steps \(t_1 < t_2\), we have 
    \(\lightcone(e_i^{t_1}) \supseteq \lightcone(e_i^{t_1})\). 
    Note further that for all \(i\in [n]\), \(|\lightcone(e_i^1)|\le 2^T\) since we assumed that every gate in \(C_\g\) acts on at most two qubits.
    
    To account for error propagation through the circuit \(C_\g\), we define the \textit{lightcone adjacency graph}, which we denote \(G_\lightcone\). 
    \(G_\lightcone\) has vertices corresponding to qubits \([n]\) of the computational code. 
    Given the stabilizer check matrices \(H_X, H_Z\), we add an edge between any pair of two qubits which are acted on by at least one common stabilizer. 
    In other words, \(\adjgraph{H_X},\adjgraph{H_Z}\) are both subgraphs of \(G_\lightcone\). 
    Further, for every qubit \(e_i\), for every qubit \(e_j\) such that \(e_j^{T+1}\in \lightcone(e_i^1)\), we add an edge between \(e_i, e_j\). 
    This completes our definition of \(G_\lightcone\). 
    The maximum degree of \(G_\lightcone\), denoted \(\Delta_\lightcone\), satisfies \(\Delta_\lightcone = O_{\Delta, T}(1)\). 

    We now bound the output error support \(\outerror\subseteq E_{T+1}\) using the support of the input error \(B_1\subseteq E_1\) and the fault path \(F\subseteq V\). 
    Recall that the vertices of \(C_\g\) is partitioned into depth, \(V = \cup_{t=1}^T V_t\). 
    We similarly partition the fault path \(F = \cup_{t=1}^T F_t\), where \(F_t = F\cap V_t\). 
    For each \(F_t\), let \(B_{t+1}\subseteq E_{t+1}\) denote the output qubits of vertices in \(F_t\), and note that \(|B_{t+1}|\le 2|F_t|\). 
    We define the lightcone of a set of qubits simply as the union of lightcones of individual qubits.
    \begin{align}
        \lightcone(B):= \cup_{e\in B}\lightcone(e).
    \end{align}
    By the definition of lightcone, we note that faults on a set of qubits \(B\subseteq E\) can propagate to only qubits in \(\lightcone(B)\subseteq E_{T+1}\). 
    This implies that
    \begin{align}
        \outerror\subseteq \lightcone\left( \bigcup_{t=1}^{T+1}B_t\right) = \bigcup_{t=1}^{T+1}\lightcone(B_t).
    \end{align}

    It suffices for us to bound the error supports of \(\lightcone(B)\).
    We define bad error supports \(\badlcerrors_c\subseteq \power([n])\) according to the graph \(G_\lightcone\). 
    Let \(\clusterset_{m,m'}\) denote the \((m,m')\)-clustering sets in \(G_\lightcone\). 
    Similar to \cref{def:qldpc-codetype}, we define
    \begin{align}
        \badlcerrors_c = \boxplus_{m=d}^n \clusterset_{m, \lceil cm/2\rceil}.
    \end{align}
    From the definition of \(G_\lightcone\), we have the following claim.
    \begin{claim}
        For all \(t\le T+1\), if \(B\subseteq E_t\) is \(\badlcerrors_{c/2^{T+1}}\)-avoiding, then \(\lightcone(B)\subseteq E_{T+1}\) is \(\baderrors_{c}\)-avoiding.
    \end{claim}
    \begin{proofclaim}
    Suppose for the sake of contradiction \(\lightcone(B)\) contains a set \(S\subseteq E_{T+1}\cong [n]\) such that \(S\) is a subset of a connected set \(\bar{S}\) of size \(m\ge d\) and \(|S|\ge \lceil cm/2\rceil\).
    Since \(\bar{S}\) is a connected set in \(\adjgraph{H_X}\) or \(\adjgraph{H_X}\), which are both subgraphs of \(G_\lightcone\), \(\bar{S}\) must be connected in \(G_\lightcone\) as well. 
    Let \(N(\bar{S})\) denote the neighborhood of \(\bar{S}\) in \(G_\lightcone\). 
    Then 
    \begin{align}
        |B\cap (N(\bar{S})\cup \bar{S})| \ge \frac{|S|}{2^T} \ge \frac{c|\bar{S}|}{2^{T+1}},
    \end{align}
    since every lightcone has size bounded by \(2^T\). 
    For convenience let \(R:= B\cap (N(\bar{S})\cup S)\).     
    Then \(R\cup \bar{S}\) is a connected set of vertices in \(G_\lightcone\) of size at least \(d\), and 
    \begin{align}
        \frac{c|R\cup \bar{S}|}{2^{T+2}}
        &\le \frac{c|\bar{S}|}{2^{T+2}} + \frac{c|R|}{2^{T+2}}
        \le |R|.
    \end{align}
    This contradicts our assumption that \(B\) is \(\badlcerrors_{c/2^{T+1}}\)-avoiding, which proves our claim.
    \end{proofclaim}
    
    Set \(c = \frac{1}{5(T+1)2^{T+1}}\).
    Following \cref{lemma:separable-bad-errors}, we see that if \(B_t\) is \(\badlcerrors_{c}\)-avoiding for all \(t\le T+1\), then \(\outerror\) is \(\baderrors_{1/5}\)-avoiding. 
    Therefore, we set \(\reserror = \badlcerrors_{c}\), which is separable as two copies of \(\badlcerrors_{c/2}\).
    The EC gadget, \cref{lemma:qldpc:ec-gadget}, gives us \(\badfaults_\mathsf{EC}\) with bounded weight enumerator.
    \begin{align}
        \weightenum{\badfaults_\mathsf{EC}}{x}
        &\le dn\cdot [O_\Delta(x)]^{\frac{d}{80\Delta}} + 2\weightenum{\badlcerrors_{c/2}}{O_\Delta(x^{\frac{1}{4\Delta}})} \\
        &\le dn\cdot [O_\Delta(x)]^{\frac{d}{80\Delta}} + n\cdot [O_{\Delta,T}(x)]^{\frac{cd}{16\Delta}}.
    \end{align}   
    We further define the bad fault paths on the circuit \(C_\g\), \(\badfaults_C\), as follows. 
    For every \(t\le T\), for \(e\in E_{t+1}\), let \(s(e)\subseteq V_{t}\) denote the locations at time \(t\) that has \(e\) as an output qubit.
    We define
    \begin{align}
        \badfaults_{C, t} = \boxplus_{B\in \badlcerrors_c}\{F_t\subseteq V_t\mid \forall e\in B, |s(e)\cap F_t| > 0\}.
    \end{align}
    There are at most \(2|B|\) locations in \(V_t\) with output qubits overlapping with \(B\), so there are at most \(2^{2|B|}\) such sets \(F_t\). Each set must have weight at least \(|B|/2\). Therefore,
    \begin{align}
        \weightenum{\badfaults_{C, t}}{x} \le \weightenum{\badlcerrors_c}{4x^{\frac{1}{2}}}.
    \end{align}
    We define \(\badfaults_C\) to be the product of \(\badfaults_{C, t}\).
    \begin{align}
        \badfaults_C 
        &= \boxplus_{t=1}^T \badfaults_{C, t}
    \end{align}
    Invoking \cref{lemma:weight-enum-baderrors-c},
    the weight enumerator can be bounded as 
    \begin{align}
        \weightenum{\badfaults_C}{x} 
        &\le T\weightenum{\badlcerrors_c}{4x^{\frac{1}{2}}} \\
        &\le Tn\cdot [O_{\Delta,T}(x)]^{\frac{cd}{2}}.
    \end{align}
    Define \(\badfaults_\g = \badfaults_C\boxplus \badfaults_{\mathsf{EC}}\), and the bound on its weight enumerator follows.
\end{proof}

\begin{lemma}[CSS Code Initialization]\label{lemma:qldpc:initialization}
    Let \(\INITZ\) be the gate that initializes \(k
    \) qubits to the \(\ket{0}\) state. There exists a gadget for \(\INITZ\) with trivial input code type and output code type \(\compCodeType[1/5]\) and a family of bad fault paths \(\mathcal{F}\) where
    \begin{align}
        \weightenum{\mathcal{F}}{x} \le d\cdot n\cdot [O_{\Delta}(x)]^{\frac{d}{80\Delta}}. 
    \end{align}
\end{lemma}
\nomenclature[F, 01]{\(\INITZ, \INITX\)}{Gadgets which initialize a CSS code state in the all \(\ket{0}\) or \(\ket{+}\) basis, see \cref{lemma:qldpc:initialization}.}
\begin{proof}
    We use \(\INITZ^{\otimes n}\) to initialize \(n\) qubits to the \(\ket{0}\) state, and then apply the error correction gadget \cref{lemma:qldpc:ec-gadget} with \(\reserror = \baderrors_{1/5}\) separable into two copies of \(\baderrors_{1/10}\). 
    Define the bad fault paths \(\badfaults_{\INIT}\) on the locations of the transversal initialization to be \(\baderrors_{1/10}\)
    We define the overall bad fault paths \(\mathcal{F}\) to be \(\badfaults_{\INIT}\boxplus \badfaults_{\ecgadget}\). 
    The upper bound on the weight enumerator then follows from the same analysis as in \cref{lemma:qldpc:transversal-gates}.

    It remains for us to show that this gadget behaves correctly. We make use of the following fact.
    \begin{fact}
        For a CSS code with \(\PAULIX\) parity check matrix \(H_X\), let \(m\) be the dimension of \(\rowspan(H_X)\), i.e., \(|\rowspan(H_X)| = 2^m\).
        Let \(\ket{0_L}^{\otimes k}\) denote the logical \(\ket{0^k}\) state of the code.
        It holds that
        \begin{align}
            \ket{0}^{\otimes n} &= \frac{2^{m/2}}{2^{n}}\sum_{v\in \F^n} \PAULIZ^v \ket{0_L}^{\otimes k}.
        \end{align}
        As a corollary, it holds that
        \begin{align}
            \ketbra{0}^{\otimes n} &= \frac{1}{2^{2n-m}}\sum_{u,v\in \F^n} \PAULIZ^u \encoder(\ketbra{0}^{\otimes k})\PAULIZ^v.
        \end{align}
    \end{fact}
    \begin{proofclaim}
        Recall that for a CSS code, 
        \begin{align}
            \ket{0_L}^{\otimes k} &= \frac{1}{2^{m/2}}\sum_{r\in \rowspan(H_X)}\ket{r}.
        \end{align}
        We can thereby compute
        \begin{align}
            \frac{2^{m/2}}{2^{n}}\sum_{v\in \F^n} \PAULIZ^v \ket{0_L}^{\otimes k}
            &= \frac{1}{2^n}\sum_{r\in \rowspan(H_X)} \left(\sum_{v\in \F^n} \PAULIZ^v \ket{r} \right) \\
            &= \frac{1}{2^n}\sum_{r\in \rowspan(H_X)} \left (\sum_{v\in \F^n} (-1)^{v\cdot r} \ket{r} \right )
        \end{align}
        For \(r\ne 0^n\), the sum \(\sum_{v\in \F^n} (-1)^{v\cdot r} \ket{r}\) evaluates to 0. The only term that remains is \(\ket{0}^{\otimes n}\), which proves this claim.
    \end{proofclaim}

    Let \(\fault\) be a \(\badfaults\)-avoiding Pauli fault, which can be divided into \(\fault_1\) supported on locations of \(\INITZ^{\otimes n}\) and \(\fault_2\) supported on locations of the error correction gadget. 
    Suppose the faulty circuit \(\INITZ^{\otimes n}[\fault_1]\) introduces an error \(\error\), whose support is \(\baderrors_{1/5}\)-avoiding, onto the state \(\ketbra{0}^{\otimes n}\). 
    In the error correction gadget, we repeatedly measure the stabilizers of the CSS code, which enables us to invoke a variant of the decoherence of errors lemma (\cref{lemma:deterministic-errors}) on the input state. Specifically, consider the Kraus decomposition of \(\error\) in terms of Pauli operators \(\{K_\mu\}_\mu, \{K_{\nu}'\}_\nu\) and complex coefficients \(\{\alpha_{\mu,\nu}\}_{\mu,\nu}\). 
    The measurement channel can be written as 
    \begin{align}
        \mathcal{M}(\rho) &= \sum_{\text{syndrome } s}\ketbra{s}\otimes \Pi_s\rho\Pi_s,
    \end{align}
    where \(\Pi_s\) is the projector onto the syndrome eigenspace corresponding to \(s\). 
    \begin{align}
        \Pi_s \error(\ketbra{0}^{\otimes n}) \Pi_s
        &= \frac{1}{2^{2n-m}}\sum_{\mu,\nu} \alpha_{\mu,\nu} \Pi_s K_\mu \left( 
            \sum_{u,v\in \F^n} \PAULIZ^u \encoder(\ketbra{0}^{\otimes k})\PAULIZ^v
         \right) K_{\nu}' \Pi_s \\
         &= \frac{1}{2^{2n-m}}\sum_{\mu,\nu,u,v} \alpha_{\mu,\nu} \Pi_s K_\mu \PAULIZ^u \encoder(\ketbra{0}^{\otimes k})\PAULIZ^v K_{\nu}' \Pi_s.
    \end{align}
    As in the proof of \cref{lemma:deterministic-errors}, only terms where both \(K_\mu \PAULIZ^u\) and \(K_{\nu}' \PAULIZ^v\) has the same syndrome \(s\) passes the projectors while the other terms are annihilated. 
    Moreover, since we have a CSS code, the \(X\) and \(Z\) errors have independent syndromes. 
    All \(X\) errors must come from the operators \(K_\mu, K_{\nu}'\), whose support is \(\baderrors_{1/5}\)-avoiding and therefore recoverable. As in the proof of \cref{lemma:deterministic-errors}, the \(X\) terms in \(K_\mu \PAULIZ^u\) and \(K_{\nu}' \PAULIZ^v\) can be chosen to be the same up to stabilizers. 
    Denote this \(X\) error \(P_{X, s}\). 
    Similarly, the \(Z\) terms can be chosen to be the same up to stabilizers since \(\encoder(\ketbra{0}^{\otimes k})\) is stabilized by all \(Z\) logical operators. 
    Denote this \(Z\) error \(P_{Z, s}\). 
    We see that 
    \begin{align}
        \Pi_s \error(\ketbra{0}^{\otimes n}) \Pi_s
         &\propto P_{X, s}P_{Z, s} \encoder(\ketbra{0}^{\otimes k}) P_{X, s}P_{Z, s}.
    \end{align}
    Therefore, the measured state is a linear combination of \(\baderrors_{1/5}\)-avoiding \(X\) errors and arbitrarily supported \(Z\) errors acting on the ideal state \(\encoder(\ketbra{0}^{\otimes k})\).
    \begin{align}
        \mathcal{M}\circ \error (\ketbra{0}^{\otimes n})
        &= \sum_{\text{syndrome } s} \ketbra{s}\otimes \left( P_{X, s}P_{Z, s} \encoder(\ketbra{0}^{\otimes k}) P_{X, s}P_{Z, s} \right)
    \end{align}
    From \cref{remark:ECgadget-friendliness}, we know that the error correction gadget is friendly in both bases. Since \(X\) errors have bounded support, the logical error resulting from the EC gadget is a \(Z\) error, which has no impact on the state \(\encoder(\ketbra{0}^{\otimes k})\). 
    We conclude that our initialization gadget has the correct output state and output code type. 
\end{proof}

\section{qLDPC Threshold Theorems}\label{sec:qldpc-threshold-theorems}
We first begin by proving a threshold theorem for the surface code roughly following the outline in \cite{dennis2002topological}.
The scheme we construct utilizes transversal gates and magic state distillation, and dose not limit the connectivity.
To prove a threshold under connectivity constraints (such as a 2D lattice), one could utilize physical \(\SWAP\) gates to route far-apart surface code patches near to each other (such as in~\cite{pattison2023hierarchical}) for transversal gates, or implement logical gates with lattice surgery (see \cref{remark:code-switching}). 
Both of these approaches can be captured by our formalism, so it would be useful to add them.

\subsection{Surface codes}
We now begin a threshold proof using surface codes.
Here, we will introduce surface codes as a computational code as in \cref{sec:surface:ec-gadget}.

\begin{definition}[Surface code]\label{def:surface-code}
    For \(d\in\mathbb{N}\), we use \(\surfacecode{d}\) to refer to a surface code of distance \(d\) encoding one logical qubit into \(O(d^2)\) physical qubits.
    Concretely, we use the planar surface code (\cite{bravyi1998quantum,freedman2001projective}) as defined in \cite{dennis2002topological} with parameters \(\dsl d^2+(d-1)^2, 1, d \dsr\).
    For convenience, \(\surfacecode{1}\) refers to the trivial quantum code on one qubit with identity encoding map.
    We follow \cref{def:qldpc-codetype} and let \(\sctype{d}{c} = (U, \baderrors_c)\) denote the code types of a surface code of distance \(d\). 
\end{definition}
\nomenclature[F, 10]{\(\surfacecode{d}, \sctype{d}{c}\)}{Surface code of distance \(d\) and its code type, parametrized by constant \(c\).}

\begin{proposition}
    A classical linear code defined by a check matrix \(H \in \mathbb{F}^{r \times n}\) with column weight at most 2 is efficiently minimum-weight decodable by the min-weight perfect matching decoder~\cite{dennis2002topological}.
\end{proposition}
\begin{proof}[Proof sketch]
    The proof is the standard construction of a matching decoder for the surface code.
    We give a brief sketch here.
    We introduce a graph \(G_{\mathrm{decoding}} = ([r]\cup \{b\}, E)\) with \(r+1\) vertices: One vertex for every row of \(H\) and then an additional ``boundary vertex'' \(b\).
    For every column of \(H\) introduce an edge in \(G_{\mathrm{decoding}}\): If the column is weight two, connect the corresponding vertices \((r_1, r_2)\). If the column is weight one, connect the vertex to the boundary \((r_1, b)\).
    
    Given a syndrome \(s \in \im H \subseteq \F^r \simeq \sigma \subseteq [r]\), add \(b\) to the set if \(\sigma\) is odd.
    Then, we compute a minimum-weight perfect matching \(M \subseteq \sigma \times \sigma\) on the weighted fully connected graph \(G_{\mathrm{matching}} = (\sigma, \sigma \times \sigma, w)\) where \(w(v_1,v_2)\) is the weight of the shortest path connecting \(v_1,v_2\) in \(G_{\mathrm{decoding}}\).
    The correction is the \(\F\)-sum of the shortest paths corresponding to each element of the matching \(M\).
\end{proof}

\begin{corollary}[Surface code error correction gadget~\cite{dennis2002topological,kitaev2003fault,bravyi1998quantum,kitaev1997quantum}]
    There is an error correction gadget \(\ecgadget[\surfacecode{d}]\) (implementing identity) which has input code type \(\sctype{d}{1/2}\) and output code type \(\sctype{d}{1/5}\). 
    Its family of bad fault paths has weight enumerator bounded as in \cref{eq:qldpc-ec-fault-weight-enumerator}. 
\end{corollary}
\begin{proof}
    Consider the spacetime code studied in \cref{sec:spacetime-code}. 
    Its parity check matrix, illustrated in \cref{fig:spacetime-check-matrix}, has column weight bounded by 2 due to the structure of stabilizers in the surface code. 
    Therefore, minimum weight decoding of the spacetime code is equivalent to finding a min-weight perfect matching on a graph, which can be efficiently solved. 
    The rest of this lemma follows from \cref{lemma:qldpc:ec-gadget}. 
\end{proof}

\begin{lemma}[Transversal Clifford gates]\label{lemma:surface:transversal-gates}
    There exists a constant \(\epsilon_{*,\mathsf{TRANSVERSAL}} \in (0,1)\) and friendly gate gadgets for the set of operations \(\mathsf{TRANSVERSAL} = \{\IDENT, \PAULIX, \PAULIZ, \CNOT, \HAD, \PHAS, \MEASX, \MEASZ\}\) (and their classically controlled analogs, if unitary) such that 
    \begin{itemize}[topsep = 0pt]
        \item The gadgets for \(\IDENT, \PAULIX, \PAULIZ, \CNOT, \HAD, \PHAS\) and their classically controlled analogs has input and output code types \(\sctype{d}{1/5}\). The gadgets for \(\MEASX, \MEASZ\) has input code type \(\sctype{d}{1/5}\) and trivial output code type.
        \item Each gadget has width \(O(d^2)\) and (identical) depth \(O(d)\).
        \item For each operation, \(\qGate \in \mathsf{TRANSVERSAL}\), the family of bad fault paths of the corresponding gadget \(\badfaults_{\qGate}\) has weight enumerator bounded on \(x\in [0,\epsilon_{*,\mathsf{TRANSVERSAL}}]\) as
        \begin{align}
            \weightenum{\badfaults_{\qGate}}{x} \le e^{-\beta d} \poly(d)
        \end{align}
        for some constant \(\beta > 0\).
    \end{itemize}
\end{lemma}
\begin{construction}
    The logical operations \(\IDENT, \PAULIX, \PAULIZ\) and \(\CNOT\) can all be implemented transversally on any CSS code. 
    A depth-2 implementation of \(\HAD\) applying \(H^{\otimes n}\) and one layer of \(\SWAP\) is given in~\cite{moussa2016transversal,breuckmann2024foldtransversal}, with the idea originating from \cite{dennis2002topological}.
    A depth-1 implementation of \(\PHAS\) applying a combination of \(\PHAS\) and \(\CZ\) is given in \cite{moussa2016transversal,breuckmann2024foldtransversal}.
    The gate gadgets for these unitary gates are their low-depth implementation preceded by the EC gadget \(\ecgadget[\surfacecode{d}]\).\footnote{Technically, if the unitary implementation is depth-1, we add a layer of identity gates afterwards to make them depth-2 in order to satisfy the identical depth condition in lemma statement.} 
    Gadgets for classically controlled analogs are constructed similarly.

    For measurement \(\MEASX\), we first apply the EC gadget \(\ecgadget[\surfacecode{d}]\), then transversally measure all physical qubits of \(\surfacecode{d}\) in \(\PAULIX\) basis to obtain a classical bitstring \(b_X\in \F^n\). 
    Let \(H_X\) denote the \(\PAULIX\) stabilizer check matrix of \(\surfacecode{d}\), compute \(s_X = H_Xb_X\), and apply a minimum weight decoder \(D_X\) to compute a correction \(c_X\in \F^n\) where \(H_Xc_X = s_X\). 
    Let \(L_X\) be a \(\PAULIX\) logical operator of \(\surfacecode{d}\), and let \(\ell_X\in \F^n\) be the indicator vector of its support. Output \(\ell_X\cdot (b_X+c_X)\). The gadget for \(\MEASZ\) can be constructed analogously.\footnote{We pad these gadgets by identity gates (on classical registers) so that they have the same depth as the unitary gate gadgets.}
\end{construction}
\begin{proof}
    The unitary gate gadgets for \(\{\IDENT, \PAULIX, \PAULIZ, \CNOT, \HAD, \PHAS\}\) have bad fault paths and weight enumerator bounds given by \cref{lemma:qldpc:constant-depth-gadget}.

    For \(\MEASZ\) (and similarly \(\MEASX\)), denote the proposed gadget \(\mz\). 
    We first argue that in the absence of input error and fault, the perfect gadget \(\mz\) will correctly measure the logical qubit in \(\PAULIZ\) basis destructively. 
    To see that, since surface code is a CSS code, its logical state has a standard basis. 
    \begin{align}
        \forall v\in \ker(H_Z), \ket{v + H_X} &\propto \sum_{r\in \rowspan(H_X)}\ket{v+r}.
    \end{align}
    Transversal measurement in the \(\PAULIZ\) basis will therefore measure to \(v + r\) for some \(v\in \ker(H_Z)\) and \(r\in \rowspan(H_X)\). 
    We see that \(H_Z(v+r) = 0\), which means the decoder acts trivially. 
    For a \(\PAULIZ\) logical operator with indicator vector \(\ell_Z\in \ker(H_X)\), we see that \(\ell_Z\cdot (v+r) = \ell_Z\cdot v\). 
    This is the correct logical measurement result due to the commutation rules of the logical operators of CSS codes. 
    
    In the presence of input error \(\errorin\) and Pauli fault \(\fault\), 
    the fault \(\fault\) is separable as \(\fault_{\ecgadget}\) supported on the locations of the EC gadget and \(\fault_{\PAULIZ}\) supported on the locations of the transversal measurement. 
    Let us consider the output state of the EC gadget,
    \begin{align}
        \error \circ \encoder(\rho') 
        &= \circuitMap{\ecgadget[\surfacecode{d}][\fault_{\ecgadget}]} (\errorin \circ \encoder(\rho)).
    \end{align}
    Assuming \(\errorin\) and \(\fault\) is well-behaving as in \cref{lemma:qldpc:ec-gadget}, \(\rho' = \rho\) and \(\error\) is \(\surfacecode{d}_{1/5}\)-avoiding.
    Now consider the faulty transversal measurement circuit, which can be decomposed as 
    \begin{align}
        \circuitMap{\MEASZ^{\otimes n}[\fault_\PAULIZ]}
        &= \fault_{\PAULIZ, 2} \circ \MEASZ^{\otimes n} \circ \fault_{\PAULIZ, 1}
    \end{align}
    for Pauli superoperators \(\fault_{\PAULIZ, 1}, \fault_{\PAULIZ, 2}\). 
    Note that the faults must respect the input and output types of the affected gates. Since the measurement gate output classical bits, \(\fault_{\PAULIZ, 2}\) must be a diagonal Pauli fault which is a tensor product of \(\PAULIX\) operators acting on a collection of bits \(F_2\). 
    To analyze \(\error\) and \(\fault_{\PAULIZ, 1}\), because we measured all physical qubits individually in the \(\PAULIZ\) basis, by the decoherence of errors lemma (\cref{lemma:deterministic-errors}) \(\fault_{\PAULIZ, 1}\circ \error\) decoheres into single-qubit Pauli \(\PAULIX\) errors acting on a collection of qubits \(F_1\cup E\subseteq [n]\), where \(F_1\) is a subset of the support of \(\fault_{\PAULIZ, 1}\) and \(E\) is a subset of the support of \(\error\).\footnote{Technically, we apply \cref{lemma:deterministic-errors} with the trivial quantum code whose stabilizer group consists of all \(n\) single-qubit \(\PAULIX\) operators.}

    Since the support of \(\error\) is \(\surfacecode{d}_{1/5}\)-avoiding, \(E\) is \(\surfacecode{d}_{1/5}\)-avoiding. 
    Let \(\badfaults_{\MEASZ} = \surfacecode{d}_{1/5}\) denote the bad fault paths on the transversal measurement. 
    If \(\fault_\PAULIZ\) is \(\badfaults_{\MEASZ}\)-avoiding, we see that \(F = F_1\cup F_2\) is \(\surfacecode{d}_{1/5}\)-avoiding.
    Let \(1_E, 1_F\) denote the indicator vectors for sets \(E, F\), and let \(e = 1_E+1_F\). 
    We see that the support of \(e\) is \(\surfacecode{d}_{2/5}\)-avoiding, and the measurement outcome is \(v+r+e\).
    Since we used a minimum weight decoder \(D_Z\), the correction \(c_Z\) satisfies \(c_Z+e\in \rowspan(H_X)\) (\cref{prop:successful_correction}). Therefore, \(v+r+e+c_Z\in v+\rowspan(H_X)\) and the output classical bit is again the correct logical measurement outcome.
    We conclude that the proposed gadget fault-tolerantly implement \(\MEASZ\), with bad fault paths and weight enumerate bound given by \cref{lemma:qldpc:transversal-gates}.
\end{proof}
We remark that the transversal measurement gadgets can be straightforwardly generalized to any CSS code, measuring out all logical qubits in the \(\PAULIX\) or \(\PAULIZ\) basis.

\subsubsection{State injection}
To employ magic state distillation \cite{bravyi2005universal}, we need a source of noisy logical magic states.
Here, we will employ a brute-force approach to construct a state-injection gadget for the surface code.
We utilize the following unitary circuit for growing a surface code, which was presented in Figure~19 of \cite{dennis2002topological}.

\begin{fact}[Surface code growth]\label{prop:surface-code-const-growth}
    For \(\ell \in \mathbb{N}\), there exists a depth \(O(\ell)\) unitary circuit \(U_{(d+\ell)\leftarrow d}\) that acts on a \(\surfacecode{d}\) codeblock and \(O(\ell d+\ell^2)\) qubits initialized to \(\ket{0}\) and returns a \(\surfacecode{d+\ell}\) codeblock.
    That is, as superoperators, we have that
    \begin{align}
        U_{(d+\ell)\leftarrow d} \circ \encoder[\surfacecode{d}] = \encoder[\surfacecode{d+\ell}]~.
    \end{align}
\end{fact}

We can use the constant-depth circuit in \cref{prop:surface-code-const-growth} alternately with the error correction gadget to prepare an arbitrarily large surface code block.
The main idea is that the error correction gadget has exponential error suppression in the block length, so there are relatively few ways that the earlier parts of this procedure can fail.
\begin{lemma}[Surface code state injection gadget \cite{dennis2002topological}]\label{lemma:state-injection}
    There exists a constant \(\epsilon_{*,\inject} \in (0,1)\) such that,
    given a constant-sized circuit \(C\) producing a state \(\ket{\psi}\), there exists a gadget \(\inject^{(d)}_{\ket{\psi}}\) with output code type \(\sctype{d}{1/5}\) that prepares \(\encoder[\surfacecode{d}](\ketbra{\psi})\) and bad fault paths \(\badfaults_{\inject}^{d}\) with weight enumerator bounded on \(x \in [0,\epsilon_{*,\inject}]\) as
    \begin{align}
        \weightenum{\badfaults_{\inject}^{d}}{x} \le c\cdot x
    \end{align}
    for some absolute constant \(c > 0\).
    The depth is \(O(d^2)\) and the width is \(O(d^2)\).
\end{lemma}
\nomenclature[F, 11]{\(\inject^{(d)}_{\ket{\psi}}\)}{Surface code state injection gadget, see \cref{lemma:state-injection}.}
\begin{construction}
    Let \(\grow_{d' \leftarrow d}\) be the circuit that introduces the necessary ancilla qubits and applies the growth unitary \(U_{d'\leftarrow d}\) from \cref{prop:surface-code-const-growth}.
    We define the \(\inject^{(d)}\) gadget inductively as
    \begin{align}
        \inject^{(d)} &= \grow_{d \leftarrow (d-1)} \circ \ecgadget[\surfacecode{d-1}] \circ \inject^{(d-1)} \\
        &= \grow_{d \leftarrow (d-1)} \circ \ecgadget[\surfacecode{d-1}] \circ \dots \circ \grow_{3 \leftarrow 2} \circ \ecgadget[\surfacecode{2}] \circ \grow_{2 \leftarrow 1}~.
    \end{align}
    This is a gadget for identity. The gadget that injects a given state \(\ket{\psi}\) is simply the circuit \(C\) followed by injection.
    \begin{align}
        \inject^{(d)}_{\ket{\psi}}
        &= \inject^{(d)} \circ C. 
    \end{align}
\end{construction}
\begin{proof}
    For a code distance \(s\), we analyze the gadget \(\grow_{s+1\leftarrow s} \circ \ecgadget[\surfacecode{s}]\). 
    Since the growth circuit from \cref{prop:surface-code-const-growth} has constant depth, by \cref{lemma:qldpc:constant-depth-gadget}, this is a friendly gadget that implements identity with input code type \(\sctype{s}{1/5}\) and output code type \(\sctype{s+1}{1/5}\), with bad fault paths \(\badfaults_s\). 

    We now define the bad fault paths \(\badfaults_{\inject}^{d}\). 
    Let \(\bar{d}\) be a absolute constant to be specified later.
    Let \(V\) be the set of all locations in the circuit \(\inject^{(\bar{d})}\). Define \(\badfaults_1\) to be all singletons of \(V\). 
    \begin{align}
        \badfaults_1 &= \{\{v\}: v\in V\}.
    \end{align}
    Since \(\bar{d}\) is a constant, the number of locations in \(V\) is also a constant. 
    For \(s\geq \bar{d}\), recall that \(\badfaults_s\) is defined by \cref{lemma:qldpc:constant-depth-gadget}. 
    We define \(\badfaults_{\inject}^{d}\) to be the sum of all of these fault paths where \(\mathcal{F}_s\) should be understood as subsets of locations of the \(s\)-th step in the growth sequence and \(\mathcal{F}_{1}\) as subsets of locations of the first \(\bar{d}\)-steps in the growth sequence.
    \begin{align}
        \badfaults_{\inject}^{d}
        &= \badfaults_1 \boxplus_{s=\bar{d}}^{d} \badfaults_s.
    \end{align}
    Intuitively, this definition enforces that the injection gadget suffers no faults until the surface code size reaches a large enough constant, after which the exponential error suppression of the EC gadgets will dominate and ensure that the weight enumerator is sufficiently bounded. 
 
    We proceed to the analysis. From \cref{lemma:qldpc:constant-depth-gadget}, there exists a polynomial \(\poly\) and absolute constants \(c_1, c_2\) such that the weight enumerator of \(\badfaults_s\) can be bounded as follows.
    \begin{align}
        \weightenum{\badfaults_s}{x} &\leq \poly(s)(c_1 x^{1/c_2})^{s}.
    \end{align}
    We want to restrict \(x\) and \(s\) such that the following inequalities hold
    \begin{align}\label{eq:sc-injection-bound-2}
         \poly(s)(c_1 x^{1/c_2})^{s} \leq 2^{-s}x^{\frac{s}{2c_2}}\leq 2^{-s}x. 
    \end{align}
    Note that the first inequality is equivalent to
    \begin{align}\label{eq:sc-injection-bound-1}
        \poly(s)(2c_1 x^{\frac{1}{2c_2}})^{s} \leq 1.
    \end{align}
    Let \(\epsilon_{*,\inject} = (\frac{1}{4c_1})^{2c_2}\). 
    For all \(x \in [0, \epsilon_{*,\inject})\), it holds that \(2c_1 x^{\frac{1}{2c_2}} < 1/2\). 
    Now let be \(\bar{d}\) an integer which is at least \(2c_2\), such that for all \(s \ge \bar{d}\), it holds that \(\poly(s)2^{-s} \le 1\).
    Then \cref{eq:sc-injection-bound-1} holds, which implies \cref{eq:sc-injection-bound-2}. 
    The weight enumerator can then be bounded using the preceding equations.
    \begin{align}
        \weightenum{\badfaults_{\inject}^{d}}{x}
        &\leq \weightenum{\badfaults_1}{x} + \sum_{s=\bar{d}}^{d} \weightenum{\badfaults_s}{x} \\
        &\leq |V|\cdot x + \sum_{s=\bar{d}}^{d} \poly(s)(c_1 x^{1/c_2})^{s} \\
        &\leq |V|\cdot x + \sum_{s=\bar{d}}^{d} 2^{-s}x \\
        &\leq (|V|+1) \cdot x.
    \end{align}
    The depth of the gadget is \(O(d^2)\) since every error correction gadget has depth bounded by \(O(d)\) and every growth circuit has constant depth. This completes the analysis of \(\inject^{(d)}\). 

    For a given state \(\ket{\psi}\) produced by a constant-sized circuit \(C\), let \(V_C\) be the set of locations of \(C\), and define \(\badfaults_C\) to be all singletons of \(V\). 
    The collection of bad fault paths for \(\inject^{(d)}_{\ket{\psi}}\) is 
    \begin{align}
        \badfaults_{\inject, \ket{\psi}}^{d}
        &= \badfaults_{\inject}^{d} \boxplus \badfaults_C,
    \end{align}
    and its weight enumerator is bounded by \((|V|+|V_C|+1)\cdot x\).
\end{proof}

In practice, it is better to grow the lattice all at once~\cite{li2015magic}, but the analysis of the lattice growth leads to some involved combinatorics.
Interested readers should refer to \cite{lodyga2015simple,strikis2023quantum}.

\subsubsection{Magic state distillation}
We will use the magic states \(\ket{\mathsf{T}} \equiv \mathsf{T}\ket{+}\).
\begin{fact}\label{fact:distillation-code}
    There exists a quantum CSS code \(\qcode[\mathrm{distill}]\) with parameters \(\dsl 15, 1, 3 \dsr\) such that the encoding map is implementable by a Clifford circuit and transversal \(\mathsf{T}\) applies logical \(\mathsf{T}\).
    That is, 
    \begin{align}
        \decoder[\qcode[\mathrm{distill}]] \mathsf{T}^{\otimes 15} \encoder[\qcode[\mathrm{distill}]] = \mathsf{T}~.
    \end{align}
\end{fact}

\begin{fact}[\(\mathsf{T}\) teleportation circuit]\label{fact:teleportation-circuit}
    It is a standard result that \(\ket{\mathsf{T}}\) can be consumed by a Clifford circuit containing a classically-controlled \(\mathsf{SX}\) operation in order to implement \(\mathsf{T}\). For reference, see Section~10.6.2, Figure~10.25 of \cite{NielsenAndChuang}.
\end{fact}

\begin{proposition}[Single-level state distillation]\label{prop:single-level-distill}
    Let \(\mathcal{S}_{(2,15)}\) be all size-2 subsets of the set \([15]\).
    There exists a constant-sized Clifford circuit with measurements and classical processing acting on a \(15\)-qubit state such that when the input is \(\mathcal{S}_{(2,15)}\)-deviated from \(\ketbra{\mathsf{T}}^{\otimes 15}\), the output is proportional to (trivially deviated from) \(\ketbra{\mathsf{T}}\).
\end{proposition}
\begin{construction}
    The circuit consists of four steps.
    \begin{enumerate}
        \item \(\mathsf{PREP}\): Prepare the logical \(\ket{+}\) state of the code \(\qcode[\mathrm{distill}]\) from \cref{fact:distillation-code}.
        \item \(\mathsf{TP}\): Use the teleportation circuit from \cref{fact:teleportation-circuit} to apply \(\mathsf{T}\) to the logical \(\ket{+}\) state.
        \item \(\mathsf{EC}_{\mathrm{distill}}\): Perform syndrome extraction and apply a minimum weight correction.
        \item Apply \(\decoder[\mathrm{distill}]\) and output the remaining qubit.
    \end{enumerate}
\end{construction}
\begin{proof}
    We assumed that the circuit is executed without fault and the input is \(\errorin(\ketbra{\mathsf{T}}^{\otimes 15})\), where \(\errorin\) is supported on a single-qubit.
    Consider the state 
    \begin{align}
        \sigma &= \mathsf{TP} \circ \mathsf{PREP} \circ \errorin(\ketbra{\mathsf{T}}^{\otimes 15}) 
    \end{align}
    The error \(\errorin\) is propagated to an error \(\error\) acting on one of the qubits of \(\qcode[\mathrm{distill}]\), while the remaining 14 qubits have the correct \(\mathsf{T}\) gates applied. 
    \begin{align}
        \sigma &= (\error\otimes \mathsf{T}^{\otimes 14}) \circ \encoder[\mathrm{distill}](\ketbra{+}) \\
        &= (\error\circ \mathsf{T}^{\dagger})\circ \mathsf{T}^{\otimes 15} \circ \encoder[\mathrm{distill}](\ketbra{+}) \\
        &= (\error\circ \mathsf{T}^{\dagger}) \circ \encoder[\mathrm{distill}](\ketbra{\mathsf{T}})
    \end{align}
    Upon application of \(\mathsf{EC}_{\mathrm{distill}}\), by the decoherence of errors lemma (\cref{lemma:deterministic-errors}), the error superoperator \(\error\circ \mathsf{T}^{\dagger}\) decomposes into a linear combination of (diagonal) Pauli errors supported on the affected qubit, which is then corrected by the minimum weight correction. 
    \begin{align}
        \mathsf{EC}_{\mathrm{distill}}(\sigma)
        &\propto \encoder[\mathrm{distill}](\ketbra{\mathsf{T}})
    \end{align}
    Therefore applying \(\decoder[\mathrm{distill}]\) gives us the desired \(\ketbra{\mathsf{T}}\) state. 
\end{proof}

\begin{lemma}[Logical level state distillation gadget \cite{bravyi2005universal}]\label{lemma:surface:logical-t-prep}
    Let \(\mathsf{PREPARE}_{\ket{T}}\) be the superoperator that outputs a \(\ket{T}\) state.
    There exists a constant \(\epsilon_{*,\mathsf{PREPARE}} \in (0,1)\) such that for \(d \in \mathbb{N}\), there exists a family of gadgets \(\mathsf{\overline{PREPARE}}^{(d)}_{\ket{T}}\) for \(\mathsf{PREPARE}_{\ket{T}}\) that has output code type \(\sctype{d}{1/5}\) and a family of bad fault paths \(\mathcal{F}_{d}^{\mathsf{\overline{PREPARE}}_{\ket{T}}}\) satisfying the upper bound
    \begin{align}
        \weightenum{\mathcal{F}_{d}^{\mathsf{\overline{PREPARE}}_{\ket{T}}}}{x} = O\left(\poly(d)\cdot e^{-\beta d}\right)
    \end{align}
    on \(x \in [0,\epsilon_{*,\mathsf{PREPARE}}]\).

    \(\mathsf{\overline{PREPARE}}^{(d)}\) has width \(O(d^{2+\log_2(15)})\) and depth \(O(d^2)\).
\end{lemma}
\nomenclature[F, 12]{\(\mathsf{PREPARE}_{\ket{T}}\)}{Logical level state distillation gadget, which prepares a \(\ket{T}\) state by repeated rounds of 15-to-1 distillation.}
\begin{construction}
    We fix a constant \(\ell = \lceil \log_2 \beta d \rceil\) corresponding to the number of rounds of distillation.
    Run the state injection gadget \(\inject^{(d)}_{\ket{\mathsf{T}}}\) (\cref{lemma:state-injection}) for the circuit \(\mathsf{T}\ket{+}\) in parallel \(N_0=15^\ell\) times followed by the error correction gadget\footnote{The EC gadget is not necessary, but we perform it for conceptual clarity.} \(\ecgadget[\surfacecode{d}]\) to obtain the state \(\rho_0\) which is a noisy copy of \(\encoder[\surfacecode{d}](\ketbra{\mathsf{T}})^{\otimes N_0}\).
    Let \(\mathsf{DISTILL}\) be the distillation circuit for \(\ket{T}\) (\cref{prop:single-level-distill}) where every gate is replaced by \(\surfacecode{d}\) transversal gate gadgets (\cref{lemma:surface:transversal-gates}).
    Perform \(\ell\) rounds of the following procedure to obtain the sequence of states \(\rho_1, \dots \rho_\ell\):
    In round \(i \in [\ell]\), execute \(\mathsf{DISTILL}\) \(15^{\ell-i}\) times in parallel on \(\rho_{i-1}\) to obtain \(\rho_i\).
    It follows that \(\rho_i\) is \(N_{i} = 15^{\ell-i-1}\) noisy copies of \(\encoder[\surfacecode{d}](\ketbra{\mathsf{T}})\).
    Output \(\rho_\ell\).
\end{construction}
\begin{proof}
    \(\inject^{(d)}_{\ket{\mathsf{T}}}\) has depth \(O(d^2)\) while the transversal gate gadgets have depth \(O(d)\). There are \(O(\ell)\) rounds of distillation, each one having \(O(d)\) depth, so the overall depth is \(O(d^2 + \ell \cdot d) = O(d^2)\).
    Overall there are at most \(O(15^{\ell})\) gadgets at any time, each of width \(O(d^2)\), so the total width is \(O\left(d^{2+\log_2(15)}\right)\).

    \paragraph{Weight enumerators}
    Let \(\Omega_{\inject}\) be the set of \(\inject^{(d)}_{\ket{\mathsf{T}}}\) gadgets and \(\Omega_{\mathsf{TRANSVERSAL}}\) the set of all error correction and transversal gate gadgets in the distillation steps (that is, including the error correction gadget after the distillation gadget but excluding any error correction gadgets within the injection gadgets).
    Let \(\Omega_{\mathsf{DISTILL}}\) be all instances of the distillation circuit \(\mathsf{DISTILL}\).
    The gadget will have two families of bad fault paths corresponding to failure of too many injection gadgets \(\badfaults_1\) and failure of a transversal gate of the distillation circuit \(\badfaults_2\), respectively.

    \(\badfaults_2\) will be the sum of the families of bad faults associated with each transversal gate gadget.
    In the following expression, we use \(\badfaults_{\qGate}\) to mean the family of bad fault paths associated with the transversal gate gadgets \(\qGate \in \Omega_{\mathsf{TRANSVERSAL}}\).
    \begin{align}
        \badfaults_2 = \wtsum_{\qGate \in \Omega_{\mathsf{TRANSVERSAL}}} \badfaults_{\qGate}
    \end{align}
    There are \(O(\ell~ 15^\ell)\) transversal gate gadgets, so utilizing the bound in \cref{lemma:surface:transversal-gates}, on values of \(x \in [0,\epsilon_{*,\mathsf{TRANSVERSAL}}]\), the weight enumerator of \(\badfaults_2\) obeys the bound
    \begin{align}
        \weightenum{\badfaults_2}{x} \le O(\ell~15^\ell~\poly(d)~e^{-\beta d})
    \end{align}
    Let \(\mathcal{S}_{(2,15)}\) be all size-2 subsets of \([15]\).
    For \(i \in \{0, \dots, \ell\}\), the code blocks of \(\rho_i\) are in bijection with the \((\ell-i)\)-fold Cartesian product \([15]^{\times (\ell-i)}\).
    For conciseness, we define \([15]^{\times 0} \equiv \{\varnothing\}\) (singleton set) and \(() \in \{\varnothing\}\) (empty tuple).
    The distillation circuits \(\Omega_{\mathsf{DISTILL}}\) are in bijection with \(\{\varnothing\} \cup [15] \cup [15]^{\times 2} \cup \dots [15]^{\ell-1}\) such that the distillation circuit at coordinates \((i_1, \dots, i_{\ell-i}) \in [15]^{\times (\ell-i)}\) consumes the code blocks with coordinates \((i_1, \dots, i_{\ell-i}) \times [15]\).
    Let \(\mathcal{S} \subseteq P([15]^{\times \ell})\) be the set (recall the definition of \(\bullet\) from \cref{def:composition})
    \begin{align}
        \mathcal{S} = \underbrace{\mathcal{S}_{(2,15)} \bullet \dots \bullet \mathcal{S}_{(2,15)}}_{\ell \text{ times}}~.
    \end{align}
    That is, for an \(\mathcal{S}\)-avoiding set \(X\subseteq [15]^{\times \ell}\) and every \(i \in \{0, \dots, \ell\}\), \(X\) is \((\mathcal{S}_{(2,15)})^{\bullet i}\)-avoiding on every \(i\)-rectangle \((i_1, \dots i_{(\ell-i)}) \times [15]^{\times i} \subseteq [15]^{\times \ell}\) except for an \((\mathcal{S}_{(2,15)})^{\bullet(\ell-i)}\)-avoiding subset \(I \subseteq [15]^{\times (\ell-i)}\) of \(i\)-rectangles.

    Let \(\badfaults_{\inject}\) be the bad fault paths of the injection gadget \(\inject^{(d)}_{\ket{\mathsf{T}}}\).
    We will define \(\badfaults_1 \subseteq \Omega_{\inject}\) as the fault paths that lead to the failure of a set of injection gadgets that is not \(\mathcal{S}\)-avoiding.
    \begin{align}
        \badfaults_1 = \mathcal{S} \bullet \badfaults_{\inject}
    \end{align}
    Using the composition evaluation formula (\cref{lemma:composition-upper-bound}) and the simple weight enumerator \(\weightenum{\mathcal{S}_{(2,15)}}{x} = 105 x^2\), a short calculation\footnote{For a function \(f(x) = c\cdot x^2\), the \(\ell\)-fold composition \(f^{(\ell)}(x)\) obeys the recursion \(f^{(\ell)}(x) = \left(\sqrt{c}\cdot f^{(\ell-1)}(x)\right)^2\) and has the closed form \(f^{(\ell)}(x) = c^{2^\ell - 1}x^{2^\ell}\).} shows that the weight enumerator of \(\mathcal{S}\) is
    \begin{align}
        \weightenum{\mathcal{S}}{x} = (105)^{2^\ell - 1} x^{2^\ell}~.
    \end{align}
    Thus, on \(x \in [0,\epsilon_{*,\inject}]\), the weight enumerator of \(\badfaults_1\) can be bounded as
    \begin{align}
        \weightenum{\badfaults_1}{x} &\le \weightenum{\mathcal{S}}{\weightenum{ \badfaults_{\inject}}{x}} \le  (c\cdot x)^{2^\ell}
    \end{align}
    for some constant \(c > 0\) that depends on the absolute constant in the weight enumerator upper bound of \(\badfaults_{\inject}\) (\cref{lemma:state-injection}).
    Thus, there exists a constant 
    \begin{align}
        \epsilon_{*,\mathsf{PREPARE}} = \min\left(\epsilon_{*,\inject}, \frac{1}{e\cdot c}, \epsilon_{*,\mathsf{TRANSVERSAL}}\right)
    \end{align}
    such that on \(x\in [0,\epsilon_{*,\mathsf{PREPARE}}]\) the weight enumerator of \(\badfaults_1\) is bounded as
    \begin{align}
        \weightenum{\badfaults_1}{x} \le e^{-2^\ell}~.
    \end{align}

    The family of bad fault paths for the gadget is the sum
    \begin{align}
        \badfaults^{d}_{\mathsf{\overline{PREPARE}}_{\ket{T}}} = \badfaults_1 \wtsum \badfaults_2~.
    \end{align}
    Using the choice of \(\ell = \lceil \log_2 \beta d \rceil\), for \(x \in [0,\epsilon_{*,\mathsf{PREPARE}}]\) the weight enumerator is bounded as
    \begin{align}
        \weightenum{\badfaults^{d}_{\mathsf{\overline{PREPARE}}_{\ket{T}}}}{x} &= O\left(\poly(d,\ell)~15^\ell~e^{-\beta d}+ e^{-2^\ell}\right) \\
        &= O\left(\poly(d) e^{-\beta d}\right)~.
    \end{align}

    \paragraph{Correctness}
    It remains to show that the output satisfies the code type \(\sctype{d}{1/5}\) when the fault path is \(\badfaults^{d}_{\mathsf{\overline{PREPARE}}_{\ket{T}}}\)-avoiding.
    Let \(\fault\) be a \(\badfaults^{d}_{\mathsf{\overline{PREPARE}}_{\ket{T}}}\)-avoiding Pauli fault.
    Recall that \(\rho_0\) is the state after state injection and error correction \((\ecgadget[\surfacecode{d}] \circ \inject_{\ket{T}}^{(d)})^{\otimes N_0} [\fault]\).

    From the construction of \(\badfaults^{d}_{\mathsf{\overline{PREPARE}}_{\ket{T}}}\), \(\fault\) is \(\badfaults_1\)-avoiding and \(\badfaults_2\)-avoiding.
    It follows that there is an \(\mathcal{S}\)-avoiding set \(B \subseteq \Omega_{\inject} \equiv [15]^\ell\) of \(\inject^{(d)}_{\ket{\mathsf{T}}}\) gadgets that have an \(\mathcal{F}_{\inject}\)-avoiding fault path.
    The state after the injection gadgets is therefore \(\mathcal{S}\bullet \sctype{d}{1/5}\)-deviated from \(\encoder[\surfacecode{d}](\ketbra{\mathsf{T}}^{\otimes N_0})\).
    Using the friendliness property\footnote{Informally, the blocks that are in the \(\mathcal{S}\)-avoiding set are arbitrarily damaged and are corrected to some arbitrary codestate.} of the error correction gadgets, it follows that \(\rho_0\) is \(\sctype{d}{1/5}\)-deviated on all blocks from a state \(\encoder[\surfacecode{d}](\bar{\rho}_0)\) where \(\bar{\rho_0}\) is \(\mathcal{S}\)-deviated from \(\ketbra{\mathsf{T}}^{\otimes{N_0}}\).
    \begin{align}
        \rho_0 &\precapprox_{\sctype{d}{1/5}} \encoder[\surfacecode{d}](\bar{\rho}_0) \\
        \bar{\rho}_0 &\precapprox_{\mathcal{S}} \ketbra{\mathsf{T}}^{\otimes{N_0}}
    \end{align}
    Using the property that \(\mathcal{F}_2\) is the sum of the bad fault paths for each transversal gate gadget of \(\Omega_{\mathsf{DISTILL}}\) and applying the gate the gadget definition (\cref{def:gadget}), for each distillation level \(i \in [\ell]\), \(\rho_i\) is \(\sctype{d}{1/5}\)-deviated on each block from a state \(\encoder[\surfacecode{d}](\bar{\rho}_i)\) where \(\bar{\rho}_i\) is the result of applying \(i\) layers of the distillation protocol \cref{prop:single-level-distill} recursively to \(\bar{\rho}_0\).

    We will now apply the distillation property of \cref{prop:single-level-distill} inductively to establish that each \(\bar{\rho}_i\) is deviated from a tensor product of \(\ketbra{\mathsf{T}}\).
    Consider layer-\(i\) of executions of the distillation protocol,
    \begin{align}
        \Omega_{\mathsf{DISTILL}}^{(i)} \equiv [15]^{\times (\ell-i)} \subseteq \Omega_{\mathsf{DISTILL}}.
    \end{align}

    \(\Omega_{\mathsf{DISTILL}}^{(i-1)}\) indexes the input qubits which are grouped together in sets labeled by indices in \(\Omega_{\mathsf{DISTILL}}^{(i)}\) (see construction).
    Suppose that there exists a \(\left(\mathcal{S}_{(2,15)}\right)^{\bullet (\ell-i)}\)-avoiding subset \(I \subseteq \Omega_{\mathsf{DISTILL}}^{(i)}\) such that for each collection of inputs \(A = (j_1, \dots, j_{\ell-i})\times [15] \in \Omega_{\mathsf{DISTILL}}^{(i-1)}\), either \(\bar{\rho}_{i-1}\) is \(\mathcal{S}_{2}^{(15)}\)-deviated from \(\ketbra{\mathsf{T}}^{\otimes 15}\) on \(A\),\footnote{More accurately, \(\bar{\rho}_{i-1}\) is \(\mathcal{S}_{2}^{(15)}\)-deviated from a state \(\ketbra{\mathsf{T}}^{\otimes 15} \otimes \sigma\) where \(\sigma\) is an arbitrary state of the complement \(A^{c}\) and the deviation superoperator may act arbitrarily on \(\sigma\). Here, \(\sigma\) should be considered as an input to an environment circuit.} or \((j_1, \dots, j_{\ell-i})\in I\).
    Then, applying \cref{prop:single-level-distill} to every distillation protocol execution in 
    \begin{align}
        I^c \equiv \Omega_{\mathsf{DISTILL}}^{(i)}\setminus I
    \end{align}
    implies that \(\bar{\rho}_{i}\) is \(\ketbra{\mathsf{T}}^{\otimes N_i}\) except for a superoperator supported on \(I\).
    In other words, \(\bar{\rho}_{i}\) is \(\left(\mathcal{S}_{(2,15)}\right)^{\bullet (\ell-i)}\)-deviated from \(\ketbra{\mathsf{T}}^{\otimes N_i}\).
    Since \(\mathcal{S} \equiv \left(\mathcal{S}_{(2,15)}\right)^{\bullet (\ell)}\), the base case is satisfied and \(\left(\mathcal{S}_{(2,15)}\right)^{\bullet 0} \equiv \{\{1\}\}\) so \(\bar{\rho}_\ell\) is proportional to \(\ketbra{\mathsf{T}}\).
    The output state \(\rho_\ell\) was previously shown to be \(\sctype{d}{1/5}\)-deviated from \(\encoder[\surfacecode{d}](\bar{\rho}_\ell)\), so this completes the proof.
\end{proof}

\subsubsection{Assembly}
\begin{theorem}[Threshold theorem for surface code quantum computation]\label{thm:surface-code-threshold}
    There exists a constant \(\epsilon_* \in (0,1)\) such that, for any Clifford+T circuit \(C\) of width \(W\) and depth \(D\) and any \(\epsilon \in (0,1)\), there exists a circuit \(\overline{C}\) that is a fault-tolerant gadget for \(C\) with bad fault paths \(\mathcal{F}\).
    When \(C\) is a circuit with only classical inputs and classical outputs, \(\overline{C}\) has trivial input and output types.
    The type of any quantum inputs and outputs is \(\sctype{d}{1/5}\) where \(d = O(\log(V) \mathrm{polyloglog}(V))\).

    Let \(V = \frac{WD}{\epsilon}\).
    Then, \(\overline{C}\) has width \(\overline{W}\) and depth \(\overline{D}\) satisfying the bounds
    \begin{align}
        \overline{W} &= O(W \log^{5.91}(V) \mathrm{polyloglog}(V)) \\
        \overline{D} &= O(D \log^2(V) \mathrm{polyloglog}(V))~.
    \end{align}
    On \(x \in [0,\epsilon_*]\), the weight enumerator of \(\mathcal{F}\) satisfies the bound
    \begin{align}
        \weightenum{\mathcal{F}}{x} \le \epsilon~.
    \end{align}
\end{theorem}
\begin{construction}
    We will defer the choice of the surface code size \(d \in \mathbb{N}\) to the proof.
    Let \(\epsilon_* = \min\left(\epsilon_{*,\mathsf{PREPARE}}, \epsilon_{*,\mathsf{TRANSVERSAL}}\right)\).

    \(C\) is Clifford+T, at the loss of a constant multiplicative factor in the depth, we can assume it contains gates and their classically controlled analogs (except for measurement) from the following set 
    \begin{align}
        \mathsf{GATESET} = \{\mathsf{T}, \IDENT, \PAULIX, \PAULIY \PAULIZ, \CNOT, \HAD, \PHAS, \MEASX, \MEASZ, \INITZ\}~.
    \end{align}
    We begin by construction a set of gadgets that implement gates in \(\mathsf{GATESET}\) and have identical depth.
    We construct a gadget \(\mathsf{G}_{\mathsf{T}}\) for \(\mathsf{T}\) in the following way: First prepare a magic state with the gadget \(\mathsf{\overline{PREPARE}}^{(d)}_{\ket{T}}\) (\cref{lemma:surface:logical-t-prep}), then consume it to apply a teleported \(\mathsf{T}\) gate using the transversal gate gadgets (\cref{lemma:surface:transversal-gates}) implementing the \(\mathsf{T}\) teleportation circuit (\cref{fact:teleportation-circuit}).
    Let \(\tau\) be the duration of this gadget.

    For all other gates \(\mathsf{g} \in \mathsf{GATESET}\setminus \{\mathsf{T}\}\), the operation is in \(\mathsf{TRANSVERSAL}\).
    Therefore, we can construct the corresponding gadget \(\mathsf{G}_\mathsf{g}\) by padding gadgets from \cref{lemma:surface:transversal-gates} to length \(\tau\) with an additional error correction gadget.

    Construct the family of circuits \(\overline{C}_d\) by replacing every gate \(\mathsf{g}\) of \(C\) with the corresponding gadget \(\mathsf{G}_{\mathsf{g}}\).
\end{construction}
\begin{proof}
    First, since \(\mathsf{G}_{\mathsf{T}}\) is the composition of a constant number of gadgets with the state preparation gadget, \(\tau = O(d^2)\) (i.e. the logical cycle time is not too long).
    Consider an operation \(\mathsf{g} \in \mathsf{GATESET}\) and let \(\mathcal{F}_\mathsf{g}\) be the weight enumerator for the gadget \(\mathsf{G}_{\mathsf{g}}\).
    Applying the weight enumerator upper bounds (\cref{lemma:surface:transversal-gates}, \cref{lemma:surface:logical-t-prep}), it follows that the previously constructed gadgets are \((\mathsf{g}, \mathcal{F}_{\mathsf{g}})\)-FT gadgets in the sense of \cref{def:gadget} where for \(x \in [0,\epsilon_*]\), the weight enumerator of \(\mathcal{F}_{\mathsf{g}}\) is bounded as 
    \begin{align}
        \weightenum{\mathcal{F}_{\mathsf{g}}}{x} \le \poly(d) e^{-\beta d}~.
    \end{align}

    Let \(\bar{V}\) be the set of gadgets of \(\overline{C}_d\).
    For each gadget \(g \in \bar{V}\) of \(\overline{C}_d\), let \(\mathcal{F}_{g}\) be its bad fault paths.
    We define the bad fault paths of \(\overline{C}_d\) to
    \begin{align}
        \mathcal{U}_d = \boxplus_{g \in \bar{V}} \mathcal{F}_{g}~.
    \end{align}
    It follows from the construction and the application of \cref{lemma:gadget-composition} inductively on a topological sort of the bundled circuit corresponding to \(\overline{C}_d\) that \(\overline{C}_d\) is a \((C,\mathcal{U}_d)\)-FT gadget.
    Using the weight enumerator upper bounds and \cref{lemma:enumerator-ring}, it follows that for \(x \in [0,\epsilon_*]\), 
    \begin{align}
        \weightenum{\mathcal{U}_d}{x} \le \bar{V} \poly(d) e^{-\beta d}~.
    \end{align}
    Since \(\bar{V} = O(WD)\), there exists \(d = O(\log(V)~\mathrm{polyloglog}(V))\) such that \(\weightenum{\mathcal{U}_d}{x} \le \epsilon\).
    \(\overline{C}\) is \(\overline{C}_d\) with the minimum such choice of \(d\).

    By using the depth and width bounds of the corresponding gadgets, depth \(\overline{D}\) and width \(\overline{W}\) of \(\overline{C}\) satisfies
    \begin{align}
        \overline{D} = O(D \tau) &= O(D \log(V)^2~\mathrm{polyloglog}(V))\\
        \overline{W} = O(W d^{2+\log_2(15)}) &= O(W \log(V)^{5.91}~\mathrm{polyloglog}(V))~. \qedhere
    \end{align}
\end{proof}

\begin{remark}[Spacetime overhead]
    It is possible to obtain a depth overhead of \(\widetilde{O}(\log(V))\) by growing the lattice all at once~\cite{dennis2002topological,li2015magic}, but this requires some more complicated combinatorics.
    It is possible to reduce the width overhead to \(\widetilde{O}(\log^{2+\gamma}(V))\) where \(\gamma\) is the exponent of the magic state distillation scheme distillation cost by using a different magic state distillation schemes (e.g. with post selection or different codes).
    For example, it is possible to achieve \(\gamma \le 2.47\)~\cite{bravyi2005universal}, \(\gamma \le 1.6\)~\cite{bravyi2012magic}, \(\gamma \le 0.68\)~\cite{hastings2018distillation}.
    For states different from \(\ket{T}\) but convertible to \(\ket{T}\) with a catalyst state~\cite{beverland2020lower}, it is possible to achieve \(\gamma=o_{d\to \infty}(1)\) without substantially harming the time complexity~\cite{wills2024constant,golowich2025asymptotically,nguyen2025good,nguyen2024quantum}. 
\end{remark}

A direct corollary of the above theorem is that we can obtain a surface code fault tolerance scheme against local stochastic noise with a constant threshold.
Using~\Cref{prop:coherent-noise}, we can also obtain a fault tolerance scheme against coherent noise, albeit the threshold is sub-constant.

\begin{corollary}[Sub-constant threshold against coherent noise] There exists a polylogarithmically decaying function \(\delta_* = O(1/\log^{6.91}(WD/\epsilon)) \) such that, for any Clifford+T circuit \(C\) of width \(W\) and depth \(D\) and any \(\epsilon \in (0,1)\), there exists a circuit \(\overline{C}\) that is a fault-tolerant gadget for \(C\), whose parameters are given in~\Cref{thm:surface-code-threshold}.
    When \(C\) is a circuit with classical inputs and classical outputs, \(\overline{C}\) has trivial input and output types. 
    The outputs of $C$ and $\overline{C}$ are $\epsilon$-close in trace distance when $\overline{C}$ is subject to coherent noise $\fault$ of strength at most $\delta_*$.
\end{corollary}
\begin{proof}
    The construction is given in \cref{thm:surface-code-threshold}, and the proof is almost the same except we use~\Cref{prop:coherent-noise} to turn the upper bound on the weight enumerator of the bad fault paths into a bound on the trace distance. 
    Note that the condition in~\Cref{prop:coherent-noise}, which asks that the bad fault paths of the whole circuit is a sum of bad fault paths of individual gadgets, is satisfied. 
    The number of gadgets \(m\) satisfies \(m\le WD\). 
    We want
    \begin{align}\label{eq:sc-coherent-goal}
        m(1+\eta)^m\eta\leq \epsilon,
    \end{align}
    where \(\eta = \max_{i\in m}\weightenum{\mathcal{F}_i}{2(1+x)^{\ratio-1}x}\), and \(\ratio\) is defined as 
    \begin{align}
        \ratio = \max_{i\in [m]}\max_{F\in \badfaults_i} |V_i|/|F|,
    \end{align}
    for \(V_i\) being the set of locations in the \(i\)-th gadget. 

    To satisfy \cref{eq:sc-coherent-goal}, it suffices for us to choose \(\eta < \frac{\epsilon}{2m} = O(\epsilon/WD)\), as that would imply
    \begin{align}
        m(1+\eta)^m\eta &\leq \epsilon e^{\eta m}/2 \leq \epsilon e^{\epsilon/2}/2\le \epsilon.
    \end{align} 
    
    According to \cref{lemma:surface:transversal-gates} and \cref{lemma:surface:logical-t-prep}, the gadgets in the circuit has bad fault paths \(\badfaults_i\) such that for \(x \in [0,\epsilon_*]\), the weight enumerator is bounded as 
    \begin{align}
        \weightenum{\badfaults_i}{x} \le \poly(d) e^{-\beta d}
    \end{align}
    for some absolute constant \(\beta\).\footnote{Technically, the constant \(\beta\) depends on the value of \(x\), but we can always take the minimal value of \(\beta\) for this upper bound.}
    Let us choose \(d = O(WD/\epsilon)\) such that 
    \begin{align}
        \eta < \frac{\epsilon}{2m} \le \poly(d) e^{-\beta d}.
    \end{align}
    It suffices for us to choose \(\delta_*\) such that for all \(x\leq \delta_*\), we have
    \begin{align}
        2(1+x)^{\ratio-1}x \leq \epsilon_*.
    \end{align}
    From \cref{lemma:surface:transversal-gates} and \cref{lemma:surface:logical-t-prep}, we know that the volume of each gadget is at most $O(d^{7.91})$, and the bad sets are all of weight at least $\Omega(d)$.
    This implies that $\ratio = O(d^{6.91})$.
    Choose \(\delta_* = \frac{\epsilon_*\log\ratio}{2\ratio}\), then for all \(x\leq \delta_*\), it holds that
    \begin{align}
        2(1+x)^{\ratio-1}x 
        \leq \epsilon_* e^{\ratio x}/\ratio \leq \epsilon_* e^{\epsilon_* \log\ratio/2}/\ratio \leq \epsilon_*.
    \end{align}
    The threshold delays polylogarithmically in the circuit volume, namely, \(\delta_* = O(\log d/d^{6.91})\) where \(d = O(WD/\epsilon)\).
\end{proof}

\subsection{Fault-tolerant quantum output}
\begin{lemma}[Surface code unencoding gadget]\label{lemma:sc-unencode}
    There exists a constant \(\epsilon_{*,\unencode} \in (0,1)\) such that we have a gadget \(\unencode^{(d)}\) for identity with input code type \(\sctype{d}{1/5}\) and trivial output code type and bad fault paths \(\badfaults_{\unencode}^{d}\) with weight enumerator bounded on \(x \in [0,\epsilon_{*,\unencode}]\) as
    \begin{align}
        \weightenum{\badfaults_{\unencode}^{d}}{x} \le c\cdot x
    \end{align}
    for some absolute constant \(c > 0\).
    The depth is \(O(d^2)\) and the width is \(O(d^2)\).
\end{lemma}
\nomenclature[F, 13]{\(\unencode^{(d)}\)}{Surface code unencoding gadget, see \cref{lemma:sc-unencode}.}
\begin{construction}
    Recall the growth unitary \(U_{d'\leftarrow d}\) from \cref{prop:surface-code-const-growth}, let \(\shrink_{d'\rightarrow d}\) be the channel which applies the inverse unitary \(U_{d'\leftarrow d}^\dagger\) and traces out the unencoded ancilla qubits. We define the unencoding gadget by alternatively apply \(\shrink\) and \(\ecgadget\). 
    \begin{align}
        \unencode^{(d)}
        &= \shrink_{d\rightarrow d-1} \circ \ecgadget[\surfacecode{d-1}] \circ \unencode^{(d-1)}\\
        &=  \shrink_{d\rightarrow d-1} \circ \ecgadget[\surfacecode{d-1}] \circ \dots \circ \shrink_{3\rightarrow 2} \circ \ecgadget[\surfacecode{2}] \circ \shrink_{2\rightarrow 1}
    \end{align}
\end{construction}
\begin{proof}
    Note that \(\unencode^{(d)}\) is defined analogously to the injection gadget \(\inject^{(d)}\), with the growth unitary inverted. 
    The definition and analysis for its bad fault paths follow the same arguments as in \cref{lemma:state-injection}.
\end{proof}

\begin{lemma}[Quantum output]\label{lemma:quantum-output}
    There exists a constant \(\epsilon_* \in (0,1)\) such that, for any \(\epsilon \in (0,1)\) and any Clifford+T circuit \(C\) of width \(W\) and depth \(D\) outputting \(n\) qubits, there exists a circuit \(\overline{C}\) with trivial input and output types.
    For any family of bad error supports \(\mathcal{U}\subseteq P([n])\), there is a collection of bad fault paths \(\mathcal{F}_{\mathcal{U}}\) such that for any \(\mathcal{F}_{\mathcal{U}}\)-avoding Pauli fault \(\fault\), for any input state \(\rho\), it holds that
    \begin{align}
        \circuitMap{\overline{C}[\fault]}(\rho)
        &\dev{\mathcal{U}} \circuitMap{C}(\rho).
    \end{align}    

    Let \(V = \frac{WD}{\epsilon}\).
    Then, \(\overline{C}\) has width \(\overline{W}\) and depth \(\overline{D}\) satisfying the bounds
    \begin{align}
        \overline{W} &= O(W \log^{5.91}(V) \mathrm{polyloglog}(V)) \\
        \overline{D} &= O(D \log^2(V) \mathrm{polyloglog}(V))~.
    \end{align}
    On \(x \in [0,\epsilon_*]\), the weight enumerator of \(\mathcal{F}_{\mathcal{U}}\) satisfies the bound
    \begin{align}
        \weightenum{\mathcal{F}_{\mathcal{U}}}{x} \le \epsilon + \weightenum{\mathcal{U}}{c\cdot x}
    \end{align}
    where \(c>0\) is an absolute constant.
\end{lemma}
\begin{construction}
    Apply \cref{thm:surface-code-threshold} with circuit \(C\) to get a fault-tolerant gadget \(\cft\) for \(C\) with bad fault paths \(\badfaults\). 
    Let \(d\) be the distance value chosen in \cref{thm:surface-code-threshold}.
    Apply the unencoding gadget \(\unencode^{(d)}\) to all the output logical qubits of \(\cft\) to obtain the circuit \(\overline{C}\). 
    \begin{align}
        \overline{C} &= \left(\unencode^{(d)} \right)^{\otimes n}\circ \cft
    \end{align}
\end{construction}
\begin{proof}
    Recall that \(\badfaults_{\unencode}^{d}\) is the bad fault paths of the unencoding gadget \(\unencode^{(d)}\). 
    We define the bad fault paths of \(\overline{C}\) to be 
    \begin{align}
        \badfaults_{\mathcal{U}} &= \badfaults \boxplus (\mathcal{U}\bullet \badfaults_{\unencode}^{d}).
    \end{align}
    The upper bound on weight enumerator follows from \cref{lemma:composition-upper-bound}.
    \begin{align}
        \weightenum{\mathcal{F}_{\mathcal{U}}}{x} &\le \weightenum{\badfaults}{x} + \weightenum{\mathcal{U}}{c\cdot x} \\
        &\le \epsilon + \weightenum{\mathcal{U}}{c\cdot x} 
    \end{align}
    
    For correctness, decompose \(\fault\) into \(\fault_1\), which is supported on the locations of \(\cft\) and \(\fault_2\), which is supported on the locations of the unencoding gadgets. 
    From \cref{thm:surface-code-threshold}, we know that \(\cft\) is a fault-tolerant gadget for \(C\) with output type \(\surfacecode{d}_{1/5}\) on every output logical qubit. 
    Let \(\sigma\) be the output state of \(C\) given input \(\rho\), then on every surface code block, the output state \(\bar{\sigma} := \circuitMap{\cft[\fault_1]}(\rho)\) is \(\surfacecode{d}_{1/5}\)-deviated from \(\encoder[\surfacecode{d}]^{\otimes n}(\sigma)\). 
    Notationally, we write 
    \begin{align}
        \forall i\in [n], \bar{\sigma}_i &= \error_i \circ \encoder[\surfacecode{d}](\sigma_i)
    \end{align}
    for some \(\surfacecode{d}_{1/5}\)-avoiding superoperator error \(\error_i\).

    Now let \(F\subseteq [n]\) denote the set of unencoding gadgets \(\unencode^{(d)}\) on which the fault \(\fault_2\) is not \(\badfaults_{\unencode}^{d}\)-avoiding. 
    From the definition of composition (the \(\bullet\) operation), we see that \(F\) is \(\mathcal{U}\)-avoiding. 
    For all \(i\notin F\), the \(i\)-th unencoding gadget correctly implements identity. For \(i\in F\), let \(\fault_{2,i}\) denote the restriction of \(\fault_2\) to the locations of the \(i\)-th unencoding gadget, and define \(\noise_i := \circuitMap{\unencode^{(d)}[\fault_{2,i}]}\). 
    \begin{align}
        \circuitMap{\overline{C}[\fault]}(\rho)
        &= \circuitMap{\left(\unencode^{(d)} \right)^{\otimes n}[\fault_2]}\circ \circuitMap{\cft[\fault_1]}(\rho) \\
        &= \circuitMap{\left(\unencode^{(d)} \right)^{\otimes n}[\fault_2]} \circ \left( \bigotimes_{i\in [n]} \error_i \circ \encoder[\surfacecode{d}] \right) (\sigma) \\
        &= \left(\bigotimes_{i\notin F} \decoder[\surfacecode{d}]\circ \error_i \circ \encoder[\surfacecode{d}]\right)\otimes \left( \bigotimes_{i\in F} \noise_i\circ \error_i \circ \encoder[\surfacecode{d}]\right) (\sigma) \\
        &= \left( \bigotimes_{i\in F} \noise_i\circ \error_i \circ \encoder[\surfacecode{d}]\right) (\sigma). 
    \end{align}
    Since \(F\) is \(\mathcal{U}\)-avoiding, we see that \(\circuitMap{\overline{C}[\fault]}(\rho)\dev{U} \sigma\). 
    The bounds on the width and depth of \(\overline{C}\) follows a simple calculation from \cref{thm:surface-code-threshold} and \cref{lemma:sc-unencode}.
\end{proof}

\subsection{Constant space overhead fault-tolerance}
Having constructed a fault-tolerance scheme with space and time overhead, we can use it to work our way up to a scheme with constant quantum space overhead~\cite{gottesman2014fault,fawzi2020constant,tamiya2024polylog,nguyen2024quantum}.
Here, we follow roughly the construction of \cite{tamiya2024polylog} with surface codes replacing concatenated codes and non-single-shot error correction gadgets instead of single-short ones.
It should be noted that the theorem proved in \cite{tamiya2024polylog} is stronger in that the model of computation here does not restrict the size of the classical compution: Effectively, we have assumed the existence of a classical minimum-weight decoder which runs instantaneously.
The exponent of the depth overhead in \cref{thm:ldpc-constant-space} can be reduced to \(1+o(1)\) in a model of computation with very limited or even constant depth classical computation~\cite{nguyen2024quantum} at the cost of a substantially more involved construction.

We first need the following fact about the existence of an asympotically good qLDPC code family shown by \cite{panteleev2022asymptotically}, building upton a line of work~\cite{hastings2021fiber,panteleev2021quantum,breuckmann2021balanced}. 
The work of \cite{panteleev2022asymptotically} is later extended by \cite{leverrier2022quantum,dinur2023good}.

\begin{fact}[Good qLDPC codes]\label{fact:good-qldpc}
    There exists an efficiently constructable qLDPC code family indexed by \(i \in \mathbb{N}\) such that for some absolute constant \(c>1\), the \(i\)-th member of the family \(\ldpcCode{i}\) has parameters \(n_i = \Theta(c^{i})\),
    \begin{align}
        \dsl n_i, \Theta(n_i), \Theta(n_i) \dsr
    \end{align}
\end{fact}

\begin{definition}[Good qLDPC code family]
    We use \(\ldpctype{i}{c}\) to refer to the code type constructed from the \(i\)-th member of the family of \cref{fact:good-qldpc} with the encoding map and bad sets constructed as in \cref{def:surface-code}.
\end{definition}
\nomenclature[F, 40]{\(\ldpcCode{i}\)}{The \(i\)-th member of a family of asymptotically good codes. Its code type is denoted \(\ldpctype{i}{c}\).}

\begin{proposition}[Transversal Clifford gates for CSS qLDPC codes]\label{lemma:qldpc:transversal-clifford-gates}
    There exists a constant \(\epsilon_{*,\mathsf{qLDPCTRANSVERSAL}} \in (0,1)\) and friendly gate gadgets for the set of operations
    \begin{align}
    \mathsf{CSSTRANSVERSAL} = \{\IDENT, \PAULIX, \PAULIZ, \CNOT, \MEASX, \MEASZ\}
    \end{align}
    (and their classically controlled analogs, if unitary) such that 
    \begin{itemize}[topsep = 0pt]
        \item The gadgets for \(\IDENT, \PAULIX, \PAULIZ, \CNOT\) and their classically controlled analogs has input and output code types \(\ldpctype{i}{1/5}\). The gadgets for \(\MEASX, \MEASZ\) has input code type \(\ldpctype{i}{1/5}\) and trivial output code type.
        \item Each gadget has width \(O(k_i)\) and (identical) depth \(O(d_i)\).
        \item For each operation, \(\qGate \in \mathsf{CSSTRANSVERSAL}\), the family of bad fault paths of the corresponding gadget \(\badfaults_{\qGate}\) has weight enumerator bounded on \(x\in [0,\epsilon_{*,\mathsf{qLDPCTRANSVERSAL}}]\) as
        \begin{align}
            \weightenum{\badfaults_{\qGate}}{x} \le e^{-\beta d_i} O(\poly(d_i))
        \end{align}
        for some constant \(\beta > 0\).
    \end{itemize}
\end{proposition}
\begin{proof}
    All CSS codes have transversal implementations of the operations in \(\mathsf{CSSTRANSVERSAL}\).
    The construction and proof is identical to \cref{lemma:surface:transversal-gates} with parameters modified in the obvious way.
\end{proof}

\begin{lemma}[qLDPC resource state gadget]\label{lemma:qldpc-resource-state}
    Let \(k_i\) be the number of qubits encoded by \(\ldpcCode{i}\) and \(d_i\) its distance.
    Consider \(m = O(1)\) registers of size \(k_i\) and let \(\psi\) be an \(m k_i\)-qubit state prepared by a Clifford+T circuit \(C_{\psi}\) of width \(O(m k_i)\) and depth \(O(1)\).\footnote{Depth and \(m\) \(O(1)\) here means an absolute constant such as \(10\).}

    There exists an absolute constant \(\epsilon_{*,\mathsf{qLDPC~PREP}} \in (0,1)\) and a circuit \(\mathsf{PREP}_{\psi}^i\) such that \(\mathsf{PREP}_{\psi}^i\) is a \((C_\psi, \mathcal{F})\)-FT gadget with output code types \(\ldpctype{i}{1/5}\) for each of the \(m\) registers.
    \(\mathsf{PREP}_{\psi}^i\) has depth \(D\) and width \(W\) where
    \begin{align}
        W &= O(k_i ~ d_i^{5.91}~\mathrm{polylog}(d_i)) \\
        D &= O(d_i^2 ~ \mathrm{polylog}(d_i))~.
    \end{align}
    On \(x \in [0,\epsilon_{*,\mathsf{qLDPC~PREP}}]\),
    \begin{align}
        \weightenum{\mathcal{F}}{x} \le \poly(d_i) e^{-d_i / 10}~.
    \end{align}
\end{lemma}
\nomenclature[F, 41]{\(\mathsf{PREP}_{\psi}^i\)}{Resource state preparation gadget for qLDPC codes, see \cref{lemma:qldpc-resource-state}.}
\begin{construction}
    Let \(\encoder[\ldpcCode{i}]^{\mathsf{CLIFF}}\) be a Clifford encoder for \(\ldpcCode{i}\) which can selected to be constant depth~\cite{gottesman2014fault} and width \(O(m k_i)\).\footnote{The decoder must be consistent with the decoder unitary, see \cref{rmk:decoder-unitary}. The precise choice of logical basis for \(\ldpcCode{i}\) is arbitrary until this point.}
    Let \(C'\) be a circuit that executes \(C_{\psi}\) and then encodes the output (non-fault tolerantly) into \(m\) \(\ldpcCode{i}\) blocks using \(\encoder[\ldpcCode{i}]^{\mathsf{CLIFF}}\).
    \begin{align}
        C' = \left(\encoder[\ldpcCode{i}]^{\mathsf{CLIFF}}\right)^{\otimes m} \circ C_{\psi}
    \end{align}
    The gadget is the application of \cref{lemma:quantum-output} to \(C'\) with \(\epsilon \equiv \epsilon_{\mathrm{prep}} \in (0,1)\) to be selected later and \(\mathcal{U} = \sum_{i=1}^m \ldpctype{i}{1/5}\)\footnote{We attach the code type \(\ldpctype{i}{1/5}\) to bundles of qubits originating from each application of \(\encoder[\ldpcCode{i}]^{\mathsf{CLIFF}}\) \(\mathcal{U}\) is the sum of \(\ldpctype{i}{1/5}\) imposed on each bundle of output qubits.} to obtain a \((C', \mathcal{F}_{\mathcal{U}})\)-FT gadget \(\mathsf{PREP}_{\psi}^i\).
\end{construction}
\begin{proof}
    Let \(\epsilon_{*,\mathsf{PREP}} \in (0,1)\) be the threshold value from the statement of \cref{lemma:quantum-output}.
    From \cref{lemma:quantum-output}, for some absolute constant \(c > 0\), the weight enumerator of \(\mathcal{F} \equiv \mathcal{F}_{\mathcal{U}}\) can be bounded on \(x \in [0,\epsilon_{*,\mathsf{PREP}}]\) as
    \begin{align}
        \weightenum{\mathcal{F}}{x} \le \epsilon_{\mathrm{prep}} + m \cdot \weightenum{\ldpctype{i}{1/5}}{c\cdot x}~.
    \end{align}
    Let \(d_i\) be the distance of \(\ldpcCode{i}\), we can apply \cref{lemma:weight-enum-baderrors-c} to bound the weight enumerator of the code type as
    \begin{align}
        \weightenum{\ldpctype{i}{1/5}}{c\cdot x} \le \poly(d_i) \left(O_{\Delta}(x)\right)^{d_i/10}~.
    \end{align}
    Let \(\epsilon' \in (0,\epsilon_{*,\mathsf{PREP}})\) be a small enough constant so that the term in the base of the exponent is at most \(1/e\), so that on \(x \in [0,\epsilon']\)
    \begin{align}
        \weightenum{\ldpctype{i}{1/5}}{c\cdot x} \le \poly(d_i) \left(O_{\Delta}(x)\right)^{d_i/10} \le \poly(d_i) e^{-d_i/10} \equiv f(d_i)
    \end{align}
    We now set \(\epsilon_{\mathrm{prep}} = f(d_i)\) and \(\epsilon_{*,\mathsf{qLDPC~PREP}} = \min(\epsilon', \epsilon_{*,\mathsf{PREP}})\), so that overall, the weight enumerator of the bad fault paths on \(x \in [0,\epsilon_{*,\mathsf{qLDPC~PREP}}]\) is exponentially small in the code distance after using the restriction on \(m\)
    \begin{align}
        \weightenum{\mathcal{F}}{x} \le O(1) ~ f(d_i)~.
    \end{align}

    \cref{lemma:quantum-output} was invoked on a constant depth circuit of width \(O(m k_i) = O(\poly(d_i))\) with \(\epsilon_{\mathrm{prep}} = f(d_i)\).
    Let \(V' = \frac{m k_i}{f(d_i)} = e^{d_i \mathrm{polylog}(d_i)}\).
    It follows that \(\mathsf{PREP}_{\psi}^i\) has depth \(D\) and width \(W\) where  
    \begin{align}
        W &= O(m k_i ~ \log^{5.91}(V') ~ \mathrm{polyloglog}(V')) = O(k_i ~ d_i^{5.91}~\mathrm{polylog}(d_i)) \\
        D &= O(\log^{2}(V') ~ \mathrm{polyloglog}(V')) = O(d_i^2 ~ \mathrm{polylog}(d_i))~.\qedhere
    \end{align}
\end{proof}

\begin{corollary}[qLDPC gate gadget]\label{lemma:qldpc-gate-gadget}
    There exists a constant \(\epsilon_{*,\mathsf{GATE}} \in (0,1)\) such that for a one or two-qubit operation \(\mathsf{g}\)  in the set
    \begin{align}
        \mathsf{GATESET} = \{\mathsf{T}, \IDENT(A), \PAULIX, \PAULIY, \PAULIZ, \CNOT, \HAD, \PHAS, \MEASX(A), \MEASZ(A), \INITZ(A)\}
    \end{align}
    with support on qubits from at most two registers of size \(k_i\),
    there exists a circuit \(\mathsf{GATE}_\mathsf{g}\) that is a \((\mathsf{g},\mathcal{F}_{\mathsf{g}})\)-FT gadget with input and output types \(\ldpctype{i}{1/5}\).
    On \(x \in [0,\epsilon_{*,\mathsf{GATE}}]\) the weight enumerator of \(\mathcal{F}_{\mathsf{g}}\) is bounded as
    \begin{align}
        \weightenum{\mathcal{F}_{\mathsf{g}}}{x} \le \poly(d_i) ~ e^{- \beta d_i}~.
    \end{align}
    \(\mathsf{GATE}_\mathsf{g}\) has width \(W\) and depth \(D\) where
    \begin{align}
        W &= O(k_i ~ d_i^{5.91}~\mathrm{polylog}(d_i)) \\
        D &= O(d_i^2 ~ \mathrm{polylog}(d_i))~.
    \end{align}
    \(\mathsf{GATE}_{\IDENT(A)}\) (identity) has width \(W = O(k_i)\).
    The depth is independent of \(\mathsf{g}\).
\end{corollary}
\nomenclature[F, 42]{\(\mathsf{GATE}_\mathsf{g}\)}{qLDPC gate gadget for a universal set of gates, see \cref{lemma:qldpc-gate-gadget}.}
\begin{construction}
    For a gate \(\mathsf{g} = \mathsf{GATESET} \setminus \{\IDENT(A),\INITZ(A),\MEASX(A), \MEASZ(A)\}\) supported on \(m\) registers, we use \cref{lemma:qldpc-resource-state} to prepare \(\ket{\mathsf{g}} = (\mathsf{g}\otimes I)\ket{\phi_+}^{m k_i}\) where \(\mathsf{g}\) should be interpreted as acting on one side of the Bell pair.
    Let \(t \le 3\) be the level of the Clifford hierarchy that \(\mathsf{g}\) is in.
    \(\ket{\mathsf{g}}\) can be consumed by a teleportation circuit using operations only in \(\mathsf{CSSTRANSVERSAL}\) and one operation in level \(t' = \max(t-1, 1)\) of the Clifford hierarchy~\cite{gottesman1999demonstrating}.
    Let \(C_{\mathsf{g}}\) be the circuit that inductively performs the previous procedure until \(t'=1\) and then directly applies an operation \(\{\IDENT, \PAULIX, \PAULIY, \PAULIZ\}\).
    The gadget is \(C_{\mathsf{g}}\) with every gate replaced by a gadget from \cref{lemma:qldpc:transversal-clifford-gates}.
    For \(\mathsf{g} \in \{\IDENT(A),\INITZ(A),\MEASX(A), \MEASZ(A)\}\), \(C_{\mathsf{g}}\) is the corresponding circuit from \(\mathsf{CSSTRANSVERSAL}\).
    The gadgets are padded with error correction gadgets \(\ecgadget[\ldpcCode{i}]\) (\cref{lemma:qldpc:ec-gadget}) such that they are of uniform size.
\end{construction}
\begin{proof}
    Set \(\epsilon_{*,\mathsf{GATE}} = \min(\epsilon_{*,\mathsf{qLDPCTRANSVERSAL}}, \epsilon_{*, \mathsf{qLDPC~PREP}})\).
    The width and depth bounds follow from the respective gadgets
    The weight enumerator bounds follow from having composed a constant number of gadgets, each with weight enumerator bounded as \(\poly(d_i) e^{-\beta d_i}\) on \([0,\epsilon_{*,\mathsf{GATE}}]\), and from the depth bound (for padding).
\end{proof}

\begin{theorem}[Constant space overhead threshold theorem]\label{thm:ldpc-constant-space}
    There exists a constant \(\epsilon_* \in (0,1)\) such that, for any Clifford+T circuit \(C\) with classical input and classical output of width \(W\), depth \(D\), and \(\epsilon \in (0,1)\) satisfying\footnote{This condition rules out exponentially deep circuits but is otherwise not very restrictive.} \(W \ge \log^7\left(\frac{WD}{\epsilon}\right)\), there exists a circuit \(\overline{C}\) that is a fault-tolerant gadget for \(C\) with trivial input and output and bad fault paths \(\mathcal{F}\).
    Let \(V = \frac{WD}{\epsilon}\).
    Then, \(\overline{C}\) has width \(\overline{W}\) and depth \(\overline{D}\) satisfying the bounds
    \begin{align}
        \overline{W} &= O(W) \\
        \overline{D} &= O(D \log^{8.91}(V) \mathrm{polyloglog}(V))~.
    \end{align}
    On \(x \in [0,\epsilon_*]\), the weight enumerator of \(\mathcal{F}\) satisfies the bound
    \begin{align}
        \weightenum{\mathcal{F}}{x} \le \epsilon~.
    \end{align}
\end{theorem}
\begin{construction}
    Let \(i \in \mathbb{N}\) be a parameter to be determined later.
    Without loss of generality, we may assume \(C\) contains only gates from the following gate set
    \begin{align}
        \mathsf{GATESET} = \{\mathsf{T}, \IDENT, \PAULIX, \PAULIZ, \CNOT, \HAD, \PHAS, \MEASX, \MEASZ, \INITZ\}~.
    \end{align}
    Otherwise, we may rewrite \(C\) in terms of the gates from \(\mathsf{GATESET}\) at constant depth cost.

    For each timestep \(t \in [D]\) of \(C\), let \(W_t\) be the number of qubits.
    We will have \(m_t = \lceil W_t / k_i \rceil + 3 = \Theta(W_t / k_i) \) registers of \(k_i\) qubits.
    Arbitrarily\footnote{This can be computed greedily.} make an injective assignment of the quantum inputs (or output in the case of \(\INITZ\)) of the gates of \(C\) at timestep \(t\) to the qubit coordinates \([k_i] \times [m_t]\) indexing the qubits of the \(m_t\) registers.
    We restrict the assignment such that when one of \(\{\MEASX, \MEASZ, \INITZ\}\) is assigned to a register, no other operations (including \(\IDENT\)) are assigned to that register.
    This assignment also induces an assignment of the outputs.

    For any permutation of \(N\) qubits, there exists an efficiently computable depth-\(2\) circuit of \(\SWAP\) gates that implements it.
    We now construct a new circuit \(C'\) that implements the following steps for every timestep \(t\in [D]\) of \(C\):
    \begin{enumerate}
        \item Implement the permutation that routes the outputs of gates at timestep \(t-1\) of \(C\) to the appropriate location for timestep \(t\). Note that there are ``holes'' in the register coordinates \([m_{t-1}]\) corresponding to blocks containing qubits that are acted on by (\(\MEASX\)/\(\MEASZ\)) and in the coordinates \([m_{t}]\) corresponding to blocks containing qubits that will be initialized by an \(\INITZ\).
        \item Execute all gates of \(C\) at time \(t\) on the qubits with the targets given by the previous assignment to qubit coordinates. 
        When a qubit of the block is acted on by one of \(\{\MEASX, \MEASZ\}\) perform the operation on all qubits of the block (i.e. \(\MEASX(A)\) or \(\MEASZ(A)\)). 
        Likewise perform the \(\INITZ(A)\) to initialize a block when there is a qubit assigned to that block originiating from \(\INITZ\).
        For any qubit that does not correspond to a qubit of \(C\), execute \(\IDENT\).
    \end{enumerate} 

    After replacing the \(\{\INITZ, \MEASX, \MEASZ\}\) operations supported on blocks with \(\{\INITZ(A), \MEASX(A), \MEASZ(A)\}\), every operation in \(C'\) is an operation for which there is a corresponding gadget from \cref{lemma:qldpc-gate-gadget}.

    Let \(D'\) be the depth of \(C'\) and \(W'\) its width.
    Let \(\ell \in [1,W'/k_i]\) be a parameter to be selected later that will control the maximum number of gates executed in a timestep.
    The circuit \(\overline{C}\) will be constructed from \(C'\) be executing at most \(\ell\) gate gadgets in each timestep.

    Let \(t \in [D']\) be a timestep of \(C'\).
    The operations will be executed over \(\tau\) steps such that no step contains more than \(\ell\) operations.
    This can be done by first constructing a maximum-degree \(k_i\) multigraph with a vertex for each register and an edge for each gate supported on qubits contained in both registers.
    An edge coloring \(c\) of the multigraph can be computed using at most \(O(k_i)\) colors.
    We can compute a new edge coloring \(c'\) from \(c\) by subdividing each color into at most \(O((W/k_i)/\ell)\) new colors such that the number of edges in each color is at most \(\ell\).
    \(c'\) has at most \(O(W/\ell)\) colors (see \cite[Serialization Lemma]{nguyen2024quantum}).
    From \(c'\), a partitioning of the gates\footnote{Single qubit gates are executed greedily subject to the constaints.} of \(C'\) in timestep \(t\) can be computed such that: 
    \begin{enumerate}
        \item There are at most \(\ell\) non-trivial gates (not in \(\IDENT\)) per partition.
        \item In each partition, no non-trivial two gates are supported on the same register.
        \item There are \(O(W/\ell)\) partitions.
    \end{enumerate}

    \(\overline{C}\) is constructed by the following procedure: For each timestep \(t\) of \(C'\) and for each gate \(\mathsf{g}\) in the partitioning associated with timestep \(t\), use the gadget \(\mathsf{GATE}_{\mathsf{g}}\) (from \cref{lemma:qldpc-gate-gadget}) to execute \(\mathsf{g}\) and execute \(\mathsf{GATE}_{\IDENT(A)}\) on all other blocks.
\end{construction}
\begin{proof}
    Set \(\epsilon_* = \epsilon_{*,GATE}\) from \cref{lemma:qldpc-gate-gadget}
    Let \(m = \max_{t \in [D]} m_t = O(W/k_i)\) be the maximum number of registers of \(C\).
    \(\overline{C}\) is composed entirely of \(\mathsf{GATE}\) gadgets from \cref{lemma:qldpc-gate-gadget}, so the width and depth bounds will follow.
    Let \(\overline{W}\) and \(\overline{D}\) be the width and depth of \(\overline{C}\), respectively.
    In terms of a function \(f(d_i) = \Theta(d_i^{5.91}~\mathrm{polylog}(d_i))\), (essentially corresponding to the width overhead of \(\mathsf{GATE}\)) the width and depth satisfy the bounds
    \begin{align}
        \overline{W} &= O\left(\ell~k_i~f(d_i) + (m - \ell) k_i\right)\\
        \overline{D} &= O\left(D \frac{W}{\ell}~d_i^2~\mathrm{polylog}(d_i)\right)~.
    \end{align}
    Using \cref{lemma:gadget-composition}, we can show that \(\overline{C}\) is a \((C,\mathcal{F})\)-FT gadget where \(\mathcal{F}\) is the sum of the bad fault paths of all \(\mathsf{GATE}\) gadgets in the circuit.
    There are \(O(m \cdot W/\ell\cdot D)\) of these.
    Thus, for \(x \in [0,\epsilon_*]\), we have the bound
    \begin{align}
        \weightenum{\mathcal{F}}{x} \le O(m \cdot W/\ell\cdot D)~\poly(d_i)~e^{-\beta d_i}~.
    \end{align}

    We now set \(\ell = \lceil m/f(d_i) \rceil\) and use the linear distance \(k_i = \Theta(d_i)\), so that these bounds become
    \begin{align}
        \weightenum{\mathcal{F}}{x} %
        &\le O\left(WD~f(d_i)~\poly(d_i)~e^{-\beta d_i}\right)\\
        &\le O\left(WD~\poly(d_i)~e^{-\beta d_i}\right)~.
    \end{align}
    It follows that \(d_i = O(\log(V)~ \mathrm{polyloglog}(V))\) is sufficient such that the inequality \(\weightenum{\mathcal{F}}{x} \le \epsilon\) holds.
    In other words, we pick \(i = \Theta(\mathrm{loglog}(V))\) (see \cref{fact:good-qldpc}), utilizing the exponential spacing of the code family.
    The bounds on the width and depth follow from the parameter choices and the assumption that the circuit width satisfies \(W \ge \log^7(V)\).
    \begin{align}
        \overline{W} %
        &= O\left(W + k_i f(d_i)\right)\\
        &= O\left(W + \log^{6.91}(V) \mathrm{polyloglog}(V)\right)\\
        &= O\left(W\right)\\
        \overline{D} %
        &= O\left(D k_i~f(d_i)~d_i^2~\mathrm{polylog}(d_i)\right) \\
        &= O\left(D d_i^{8.91}~\mathrm{polylog}(d_i)\right) \\
        &= O\left(D \log^{8.91}(V)~\mathrm{polyloglog}(V)\right)
    \end{align}
    This completes our threshold proof for a qLDPC code- and teleportation-based constant overhead fault-tolerant scheme.
\end{proof}

\newpage
\printnomenclature

\newpage

\appendix

\newpage 
\printbibliography

\end{document}